\documentclass[
 reprint, superscriptaddress,
 amsmath,amssymb,
 aps]{revtex4-2}
\pdfoutput=1 

\usepackage{cancel}

\usepackage{graphicx}
\usepackage[dvipsnames]{xcolor}
\usepackage{dcolumn}
\usepackage{bm}
\usepackage{bbold}
\usepackage{physics}
\usepackage{amsmath,amsthm,amssymb,amsfonts,physics,bbm,graphicx}
\usepackage{xfrac, hyperref}
\hypersetup{colorlinks,citecolor=blue,filecolor=blue,linkcolor=blue,urlcolor=blue}
\usepackage{natbib}
\usepackage{mathtools}
\usepackage{comment}
\usepackage{enumerate}

\usepackage{tikz}
\usepackage{mathtools}
\usepackage{bm}
\usetikzlibrary{arrows.meta}
\usetikzlibrary{decorations.markings,quantikz2}
\usetikzlibrary{fit,shapes.arrows}
\usepackage{pgfplots}

\newcommand{\hmax}{{H_{\textnormal{max}}}}

\newcommand{\psecr}{{p_{\textnormal{secr}}}}
\newcommand{\hmins}{{H^{*}_{\textnormal{min}}}}
\newcommand{\hmaxs}{{H^{*}_{\textnormal{max}}}}
\newcommand{\pg}{{P_{\textnormal{guess}}}}
\newcommand{\pgs}{{P^{*}_{\textnormal{guess}}}}
\newcommand{\pgq}{{P^{\, \textnormal{Q}}_{\textnormal{guess}}}}
\newcommand{\pgsq}{{P}^{\, \textnormal{Q *}}_{\textnormal{guess}}}
\newcommand{\pgc}{{P^{\, \textnormal{C}}_{\textnormal{guess}}}}

\newcommand{\m}{{\mathcal{M}}}
\newcommand{\ve}{{\varepsilon}}

\usepackage[thinc]{esdiff}
\newcommand{\vecr}{\vec{r}}
\newcommand{\id}{\mathbbm{1}}
\DeclareMathOperator{\lmin}{\lambda_{min}}
\DeclareMathOperator{\lmax}{\lambda_{max}}
\usepackage{amsthm}   
\usepackage{chngcntr}

\DeclareMathOperator*{\argmax}{arg\,max}
\pgfplotsset{compat=1.18}

\DeclareMathOperator{\diag}{\textnormal{diag}}

\newtheorem{theorem}{Theorem}
\newtheorem{lemma}{Lemma}
\newtheorem{corollary}{Corollary}

\newtheorem{sublemma}{Sublemma}

\usepackage{ragged2e} 
\usepackage{caption} 
\DeclareCaptionJustification{justified}{\justifying}
\usepackage{subcaption}
\usepackage[normalem]{ulem}

\allowdisplaybreaks

\begin{document}

\title{Maximal intrinsic randomness of noisy quantum measurements}

\author{Fionnuala Curran}
\affiliation{ICFO-Institut de Ci\`encies Fot\`oniques, The Barcelona Institute of Science and Technology,
Av. Carl Friedrich Gauss 3, 08860 Castelldefels (Barcelona), Spain} 

\author{Morteza Moradi}
\affiliation{Institute of Informatics, Faculty of Mathematics, Informatics and Mechanics,
University of Warsaw, Banacha 2, 02-097 Warsaw, Poland}

\author{Gabriel Senno}
\affiliation{Quside Technologies S.L., C/Esteve Terradas 1, 08860 Castelldefels,
Barcelona, Spain}

\author{Magdalena Stobinska}
\affiliation{Institute of Informatics, Faculty of Mathematics, Informatics and Mechanics,
University of Warsaw, Banacha 2, 02-097 Warsaw, Poland}

\author{Antonio Acín}
\affiliation{ICFO-Institut de Ci\`encies Fot\`oniques, The Barcelona Institute of Science and Technology,
Av. Carl Friedrich Gauss 3, 08860 Castelldefels (Barcelona), Spain}
\affiliation{ICREA - Instituci\'o Catalana de Recerca i Estudis Avan\c{c}ats, 08010 Barcelona, Spain}

\date{\today}

\begin{abstract}
Quantum physics exhibits an intrinsic and private form of randomness with no classical counterpart. Any setup for quantum randomness generation involves measurements acting on quantum states. In this work, we consider the following question: Given a quantum measurement, how much randomness can be generated from it? In real life, measurements are noisy and thus contain an additional, extrinsic form of randomness due to ignorance. This extrinsic randomness is not private since, in an adversarial model, it takes the form of quantum side information held by an eavesdropper who can use it to predict the measurement outcomes. Randomness of measurements is then quantified by the guessing probability of this eavesdropper, when minimized over all possible input states. This optimization is in general hard to compute, but we solve it here for any two-outcome qubit measurement and for projective measurements in arbitrary dimension mixed with white noise. We also construct, for a given measured probability distribution, different realizations with (i) a noisy state and noiseless measurement (ii) a noiseless state and noisy measurement and (iii) a noisy state \emph{and} measurement, and we show that the latter gives an eavesdropper significantly higher guessing power.

\end{abstract}

\maketitle

\textit{Introduction}.---
Randomness is a useful peculiarity of quantum physics. Say we have a quantum system represented by a pure state $\ket{\psi}$, and we measure it in a basis to which it is unbiased. Our outcomes would not just be statistically random (independent and identically distributed with uniform probability); they would also be wholly unpredictable, even to a potential eavesdropper who tried to guess them. This private or \emph{intrinsic} quality of quantum randomness is exploited in the design of quantum random number generators \cite{Mannalatha_2023, Herrero_Collantes_2017} and quantum cryptography protocols \cite{RevModPhys.94.025008}.

Quantum measurements can, however, be mixed, so they possess an additional, \emph{extrinsic} form of randomness. A mixed (or non-extremal) measurement can be seen as a probabilistic mixture of other measurements \cite{D_Ariano_2005} or as the incomplete marginal of a projective measurement carried out on a larger system, as captured by Naimark's dilation theorem \cite{Peres_1990}. In fact, real-life measurements, plagued by noise, are necessarily mixed. Moreover, from a fundamental perspective, pure projective measurements are impossible to implement with finite thermodynamic resources \cite{Guryanova_2020}.

While we usually think of noise in our measurement as inconvenient but benign, in a cryptographic scenario, to ensure the security of our outcomes, we must ascribe all of our ignorance to \emph{side information} held by an adversary. For convenience, we call our randomness-generating experimentalist Alice and her hypothetical eavesdropper Eve. Eve's side information may be classical, when she knows which measurement from a probabilistic decomposition is \emph{actually} performed by Alice in each round, or quantum, when she shares entanglement with Alice's devices. If Alice performs a measurement and records the outcome $x$ from a register $X$, Eve's knowledge is intuitively quantified by the probability that she correctly guesses $x$. This guessing probability is directly related to the min-entropy of $X$ conditioned on Eve's information \cite{Konig_2009}, which in turn lower bounds the number of private and statistically random bits that can be extracted from $X$ \cite{renner2006securityquantumkeydistribution, Konig_2011}. 

In this work, we address the following question, of relevance both from a fundamental and an applied point of view: Given a quantum measurement $\m$, what is the maximal amount of intrinsic randomness one can generate from it? For that, we should determine the optimal state on which Alice should apply the given measurement. We solve this question and therefore derive analytic solutions for the maximal intrinsic randomness of two classes of realistic noisy measurements: any qubit measurement with two outcomes and isotropic noisy projective measurements  in any dimension $d$, resulting from mixing a rank-one projective measurement with isotropic noise. Our results represent a step forward in quantifying randomness in quantum measurements and are complementary to recent works that characterized the intrinsic randomness of  quantum states~\cite{Meng_2024,anco2024securerandomnessquantumstate}.

\textit{Setting the problem}.---
Any randomness generation process involves measurements acting on states. We work throughout in a Hilbert space of finite dimension $d$. A quantum state $\rho$ is a positive semidefinite operator with trace one, while a quantum measurement is defined by a positive operator-valued measure (POVM), a set of positive-semidefinite operators $\m= \{M_x\}_{x}$ that sum to identity, $\sum_{x} M_x=\id$. If we measure $\rho$ using $\m$, the probability of obtaining the outcome $x$ is given by the Born rule, $p(x)= \tr \left( \rho M_x \right)$.

A projective measurement $\Pi= \{\Pi_{x}\}_x$ is a special class of POVM satisfying $\Pi_x \Pi_y = \delta_{xy} \Pi_x$.  
A general non-projective POVM $\m_S$ (in what follows, we sometimes label Alice's system by $S$ for clarity) is realized by a \emph{generalized Naimark extension}. Concretely, a global projective measurement $\Pi_{SA}$ is applied to Alice's state $\rho$ and an auxiliary state $\sigma_A$ inside Alice's measuring device. For such an extension to be valid, it must satisfy
\begin{equation}
    \tr_A \Big( \Pi_{x, SA} \, \id_{S} \otimes \sigma_{A} \Big) = M_{x, S} \;\; \text{for all} \; x\,.
\end{equation}
For a given measurement, there exist many possible extensions. We use an adversarial model introduced in \cite{frauchiger_2013} and generalized in \cite{Senno_2023} (see also \cite{dai2023intrinsic}), wherein the eavesdropper Eve can, in principle, share entanglement both with Alice's state $\rho$ \emph{and} with her POVM $\m_S$. However, since here we are interested in maximizing the randomness that can be generated from a given measurement, Alice's optimal state can be taken pure without loss of generality, $\rho=\ketbra{ \phi}{\phi}_{S}$, and we need only consider Eve's entanglement with $\m$. This quantum side information is encapsulated by a purification $\ket{\varphi}_{AE}$ of the auxiliary state $\sigma_A$ in Alice's device. To guess Alice's outcome, Eve applies a measurement $\{M_{x, E}\}_{x}$ to her particle $E$. Eve's guessing probability, conditioned on her quantum side information, is defined by the optimization of all the different realizations compatible with Alice's measurement,
\begin{equation}\label{eqn: pguess_quantum_pure}
\begin{aligned}
&\pgq \big( \ket{\phi}_S,\, \m_S \big)= \; 
\\
&  \max_{\{\Pi_{x, SA}\}_{x}, \, \ket{\varphi}_{AE}, \, \{M_{x, E}\}_{x}}
& & \!\!\!\!\! \sum_{x}\bra{\Phi}\Pi_{x, SA}\otimes M_{x, E}\ket{\Phi}_{SAE} 
\\
& \qquad \quad  \text{subject to} 
& & \!\!\!\!\! \!\!\!\!\! \tr_A \Big( \Pi_{x, SA} \big( \id_S\otimes \tr_{E} \left( \ketbra{ \varphi}{\varphi}_{AE} \right) \big) \Big)  
\\
& &&  \qquad \qquad  =M_{x, S} \; \text{ for all } \; x\,,
\end{aligned}
\end{equation}
where $\ket\Phi_{SAE}=\ket\phi_S\ket\varphi_{AE}$.

Happily, it is known \cite{Senno_2023} that when Alice chooses a pure state, the quantum guessing probability \eqref{eqn: pguess_quantum_pure} is equal to the guessing probability of a `classical' Eve, who is correlated to the measurements only through classical side information. More precisely, any non-extremal POVM can be decomposed as a convex combination of other, possibly extremal, measurements,
\begin{equation}
    \sum_{j} p(j) N_{x, j} = M_x\,, \qquad \sum_{x} N_{x, j} = \id\,,
\end{equation}
where $p(j)$ is the mixing probability distribution and $\{\mathcal{N}_j\}$ a set of POVMs. In the adversarial scenario, Eve knows which POVM $\mathcal{N}_{j}=\{N_{x, j}\}_{x}$ is being applied in each round (notice that we no longer use the subscript $S$ for Alice's system). It is convenient to define the subnormalized POVMs $\mathcal{K}_{j}= p(j) \, \mathcal{N}_{j}$, such that the corresponding (classical) guessing probability is given by the semidefinite programming problem 
\begin{equation}\label{eqn: pguess_main}
\begin{aligned}
&\pg \big( \ket{\phi},\, \m \big)= \; 
\\
 &    \qquad \max_{\{ \mathcal{K}_j\} }
& & \!\!\!\!\! \sum_{j}  \bra{\phi} K_{j,j}\ket{\phi}
\\
&  \quad \; \text{subject to}
& &  \!\!\!\!\! K_{x, j} \geq 0  \; \text{ for all } \; x, j
\\
& 
& &  \!\!\!\!\! \sum_{x}  d K_{x, j} =  \id \sum_{x}  \tr K_{x, j} \; \text{ for all } \; j 
\\
& 
& &  \!\!\!\!\! \sum_{j} K_{x,j} = M_{x} \text{ for all } x\,. 
\end{aligned}
\end{equation}
Our figure of merit to measure the randomness generated by a measurement is then the quantity $\pg \big( \ket{\phi},\, \m \big)$ minimized over all states $\ket{\phi}$, which gives us the min-max problem
\begin{equation}\label{eqn: pguess_opt}
\pgs \big(\m \big)= \min_{\ket{\phi}} \, \pg \big( \ket{\phi}, \, \m \big)\,,
\end{equation}
from which the maximal conditional min-entropy can be found \footnote{We take all logarithms with respect to base two.} as 
\begin{equation}
    \hmins \left( X | E, \, \m \right) = - \log \, \pgs \big(\m \big)\,. 
\end{equation}

\textit{Main results}--- We present our main contributions here, followed by three implications, with full details and derivations in the Appendix.
\begin{theorem}\label{thm: qubit 2-outcome}
Let $\m= \{M_1,M_2\}$ be any qubit POVM with two outcomes, where $\tr M_1 \leq \tr M_2$. Then
\begin{equation}\label{eqn: thm1}
  \pgs \big(  \m  \big) =  1 - \tr M_1 + \frac{1}{2} \Big( \tr \sqrt{M_1} \Big)^2 \,.
\end{equation}
\end{theorem}
We prove the theorem using two lemmas. 
\begin{lemma}\label{lemma: qubit 2-outcome lower}
Let $\m= \{M_1,M_2\}$ be any qubit POVM with two outcomes, where $\tr M_1 \leq \tr M_2$. The following lower bound holds for any state $\ket{\phi}$,
 \begin{equation}\label{eqn: lemma 1}
  \pg \big( \ket{\phi}, \,  \m  \big) \geq  1 - \tr M_1 + \frac{1}{2} \Big( \tr \sqrt{M_1} \Big)^2 \,.
\end{equation}  
\end{lemma}
We sketch the proof, providing full details in Appendix \ref{app: qubit decomposition} and an alternative proof using Bloch vectors in Appendix \ref{app: qubit 2-outcome lower}. For any state $\ket{\phi}$ chosen by Alice, Eve can decompose the POVM $\m$ as $\mathcal{K}_j=\{K_{1, j}, \, K_{2, j}\}$, where
\begin{equation}\label{eqn: qubit_decomp}
\begin{aligned}
 K_{1, 1}&=  \sqrt{M_{1}} \ketbra{\phi}{\phi} \sqrt{M_{1}}\,, 
    \\
    K_{2, 1} &= p (1) \id -  \sqrt{M_{1}} \ketbra{\phi}{\phi} \sqrt{M_{1}}\,, 
    \\
    K_{1, 2} &=     M_{1} - \sqrt{M_1} \ketbra{\phi}{\phi} \sqrt{M_1}   \,,
    \\
    K_{2, 2} &=p(2) \id -   M_{1} + \sqrt{M_1} \ketbra{\phi}{\phi} \sqrt{M_1}   \,,  
\end{aligned}    
\end{equation}
with $p(x)= \langle \phi | M_x | \phi \rangle$. This decomposition gives rise to the lower bound 
\begin{equation}\label{eqn: inter phi}
\pg \big( \ket{\phi}, \, \m  \big) \geq 1 - 2 p \left( 1\right) + 2 \langle \phi | \sqrt{M_1} | \phi \rangle^2\,.   
\end{equation}
We can further show that the right-hand side of \eqref{eqn: inter phi} is lower-bounded by the right-hand side of \eqref{eqn: lemma 1}, with equality achieved only when $\ket{\phi}$ is unbiased to the diagonal basis $\{\ket{x}\}$ of $\m$, i.e. when $\abs{\langle \phi | x \rangle}^2 = 1/2$ for all $x$.

\begin{lemma}\label{lemma: qubit 2-outcome upper}
Let $\m= \{M_1,M_2\}$ be any qubit POVM with two outcomes, where $\tr M_1 \leq \tr M_2$, and let $\ket{\psi}= \frac{1}{\sqrt{2}} \left(1, \, 1 \right)^{T}$. Then
\begin{equation}\label{eqn: lemma2}
  \pg \big( \ket{\psi}, \, \m  \big) \leq   1 - \tr M_1 + \frac{1}{2} \Big( \tr \sqrt{M_1} \Big)^2 \,.
\end{equation}  
\end{lemma}
We provide the proof in Appendix \ref{app: qubit 2 outcome upper}. Lemma \ref{lemma: qubit 2-outcome upper} shows that the lower bound in Lemma \ref{lemma: qubit 2-outcome lower} is achievable, concluding the proof of Theorem \ref{thm: qubit 2-outcome}. Note that while the state $\ket{\phi}$ must be unbiased to the basis $\{\ket{x}\}$ to achieve maximal intrinsic randomness, such a state does not result in a uniform distribution of outcomes (perfect statistical randomness) unless $\tr M_1= \tr M_2$. Since Theorem \ref{thm: qubit 2-outcome} holds for all qubit POVMS with two outcomes, it holds for any POVM arising from a noisy channel (e.g. the depolarizing, phase flip or amplitude damping channels) applied to a rank-one qubit projective measurement (see \cite[8.3]{nielsen00} for more examples).

\begin{figure}[h]
    \centering
    \begin{tikzpicture}[line cap=round, line join=round, >=Triangle]
     \centering
    \begin{scope}[scale=1]
  \clip(-2.9,-2.6) rectangle (2.8,2.8);
  \draw(0,0) circle (2cm);
  \draw[dashed] (0,0)-- (1.7, 0);
   \draw[dashed] (0,0)-- (-1.7, 0);
\draw [->] (0,0)-- (1.7, 1.05357);
  \draw [->] (0,0)-- (1.7, -1.05357);
  \draw [->] (0,0)-- (-1.7, 1.05357);
  \draw [->] (0,0)-- (-1.7, -1.05357);
  \draw [dotted] (1.7, -1.05357)-- (1.7, 1.05357);
  \draw [dotted] (-1.7, -1.05357)-- (-1.7, 1.05357);
  \draw (-2,0) node[anchor=east] {\small{$\ket{{1}}$}};
  \draw (2,0) node[anchor=west] {\small{$\ket{{2}}$}};
  \draw (0,2) node[anchor=south] {\small{$ {|\psi\rangle}$}};
  \draw (0,-2) node[anchor=north] {\small{$|{\psi^{\perp}}\rangle$}};
  \draw [fill] (1.7,0) circle (1.5pt);
  \draw [fill] (-1.7,0) circle (1.5pt);
  \end{scope}
\end{tikzpicture}
    \caption{\emph{Eve's optimal decomposition}. Schematic of the qubit noisy projective measurement $\m_2= \{M_1, \, M_2\}$, with noise parameter $\ve=0.15$, where Alice chooses the state $\ket{\psi}$ or $|{\psi^{\perp}}\rangle$. The filled circles represent Alice's POVM elements $\{M_1, \, M_2\}$, while the arrows represent Eve's decompositions $\{K_{1, j}, \, K_{2, j}\}$.}
    \label{fig: noisy_qubit_decomp}
\end{figure}
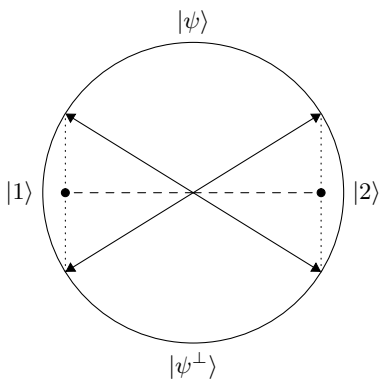

Beyond qubits, for any two-outcome POVM $\m$ of arbitrary dimension, the following upper bound to $\pgs \big(  \m  \big)$ follows readily from Lemma \ref{lemma: qubit 2-outcome upper} by projecting to a suitable qubit subspace. 

\begin{corollary}\label{corr: qudit 2-outcome}
Let $\m=\{M_x\}_x$ be any POVM with two outcomes, with $\lmin \left( M_1\right) + \lmax \left( M_1 \right) \leq \lmin \left( M_2\right) + \lmax \left( M_2 \right)$, where $\lmax \left( . \right)$ and $\lmin \left( . \right)$ denote the largest and smallest eigenvalues of the argument. Then 
\begin{equation}\label{eqn: pguess_bound_2outcome}
 \pgs \big(  \m  \big) \leq 1 - \frac{1}{2} \bigg( \sqrt{\lmax\left( M_1\right)} - \sqrt{\lmin \left( M_1\right)} \bigg)^2\,.
\end{equation}
\end{corollary}
The proof is given in Appendix \ref{app: corr1 proof}. This result covers two-outcome measurements defined by a projector, $\Pi$, and its complement $\id-\Pi$, used for instance in photonic experiments involving orbital angular momentum~\cite{Molina_07}.

Another relevant and natural measurement for systems of arbitrary dimension $d$ corresponds to a rank-one projective measurement of $d$ outcomes affected by white noise. Let $\{\ket{x}\}$ be an orthonormal basis with $x= 1, ..., d$, let $\ve \in \left(0, 1 \right)$ and consider the depolarizing channel $\Delta_{\ve} \left( . \right)$ which acts on any operator $\sigma$ as
 \begin{equation}
    \Delta_{\ve} \left( \sigma \right) = \left( 1 - \ve \right) \sigma + \frac{\ve}{d}\tr(\sigma) \id\,.
\end{equation}
We refer to the POVM $\m_{d} = \{M_x\}_x$ with
\begin{equation}
  M_x = \Delta_{\ve} \left( \ketbra{x}{x} \right) = \left( 1 - \ve \right) \ketbra{x}{x} + \frac{\ve}{d} \id\,   
\end{equation}
as the  noisy projective measurement in dimension $d$.
\begin{theorem}\label{thm: qudit noisy}
Let $\m_d$ be the noisy projective measurement in dimension $d$. Then 
\begin{equation}\label{eqn: thm2}
  \pgs \big(  \m_d  \big) =  \frac{1}{d} \Big( \tr \sqrt{M_1} \Big)^2 \,.
\end{equation}
\end{theorem}
Note that the choice of $M_1$ in \eqref{eqn: thm2} is arbitrary, as 
\begin{equation}
    \tr \sqrt{M_1} = \tr \sqrt{M_k} \;\;  \textnormal{for all} \;\; k=1, ..., d\,.
\end{equation}
We prove Theorem \ref{thm: qudit noisy} using two lemmas.

\begin{lemma}\label{lemma: qudit lower}
Let $\m_d$ be the noisy projective measurement in dimension $d$. The following lower bound holds for any state $\ket{\phi}$,
 \begin{equation}\label{eqn: lemma3}
  \pg \big( \ket{\phi}, \,  \m_d  \big) \geq  \frac{1}{d} \Big( \tr \sqrt{M_1} \Big)^2 \,.
\end{equation}   
\end{lemma}
In Appendix \ref{app: upper qudit pguess}, we present a decomposition $\{\mathcal{K}_j\}$ that generalizes the qubit `square root decomposition' from \eqref{eqn: qubit_decomp} to dimension $d$. Using $\{\mathcal{K}_j\}$, we can lower bound Eve's guessing probability for any state $\ket{\phi}$ by
\begin{equation}\label{eqn: qudit decomp bound}
\begin{aligned}
\pg \left( \ket{\phi}, \, \m_d \right) &\geq  \sum_{x} \langle{\phi | \sqrt{M_{x}} | \phi } \rangle^2 \geq  \frac{1}{d} \Big( \tr \sqrt{M_1} \Big)^2\,,    \end{aligned}
  \end{equation}
with equality if and only if $\ket{\phi}$ is unbiased to the basis $\{\ket{x}\}$, i.e. $\abs{ \langle \phi | x \rangle }^2 =1/d$ for all $x$. This proves \eqref{eqn: lemma3}.

\begin{lemma}\label{lemma: qudit upper}
Let $\m_d$ be the noisy projective measurement in dimension $d$, and let $\ket{\psi}= \frac{1}{\sqrt{d}}\left(1, ..., 1 \right)^{T} $. Then
\begin{equation}\label{eqn: lemma4}
  \pg \big( \ket{\psi}, \, \m  \big) \leq  \frac{1}{d} \Big( \tr \sqrt{M_1} \Big)^2 \,.
\end{equation}  

\end{lemma}

 We prove this lemma in Appendix \ref{app: upper qudit pguess} by finding feasible variables for the dual version of the semidefinite program \eqref{eqn: pguess_main} that achieve the bound \eqref{eqn: lemma4}, with an alternative proof in Appendix \ref{app: upper qudit pguess perm} that exploits the permutation symmetry of $\m_d$.  
 Combining Lemmas \ref{lemma: qudit lower} and \ref{lemma: qudit upper}, we prove Theorem \ref{thm: qudit noisy}. Note that a state $\ket{\phi}$ that is unbiased to $\{\ket{x}\}$ gives a uniform distribution of outcomes, so maximal intrinsic \emph{and} statistical randomness always co-occur for the noisy projective measurement. From Lemma \ref{lemma: qudit upper}, we find that the `square root decomposition' $\{\mathcal{K}_j\}$ is optimal for the unbiased state $\ket{\psi}$. When $d=2$, this decomposition has an intuitive form as both the subnormalized POVMs $\mathcal{K}_j=\{K_{1, j}, \, K_{2, j} \}$ are proportional to a rank-one projective measurement, as shown in Figure \ref{fig: noisy_qubit_decomp}.
\begin{figure*}
\begin{subfigure}{0.49\textwidth}
\centering
    \begin{tikzpicture}[line cap=round, line join=round, >=Triangle]
     \centering
    \begin{scope}[scale=1]
  \clip(-2.9,-2.6) rectangle (2.8,2.8);
  \draw(0,0) circle (2cm);
  \draw (-2,0) node[anchor=east] {\small{$\ket{{1}}$}};
  \draw (2,0) node[anchor=west] {\small{$\ket{{2}}$}};
  \draw (0,2) node[anchor=south] {\small{$ {|\psi\rangle}$}};
  \draw (0,-2) node[anchor=north] {\small{$|{\psi^{\perp}}\rangle$}};
\draw[dashed] (0,0)-- (0, 1.7);
\draw [->] (0,0)-- (1.05357, 1.7);
  \draw [->] (0,0)-- (-1.05357, 1.7);
  \draw [dotted] (-1.05357, 1.7)-- (1.05357, 1.7);
\draw [fill] (0, 1.7) circle (1.5pt);
\draw[dashed] (0,0)-- (1.7, 0);
   \draw[dashed] (0,0)-- (-1.7, 0);
\draw [->] (0,0)-- (1.7, 1.05357);
  \draw [->] (0,0)-- (1.7, -1.05357);
  \draw [->] (0,0)-- (-1.7, 1.05357);
  \draw [->] (0,0)-- (-1.7, -1.05357);
  \draw [dotted] (1.7, -1.05357)-- (1.7, 1.05357);
  \draw [dotted] (-1.7, -1.05357)-- (-1.7, 1.05357);
\draw [fill] (1.7,0) circle (1.5pt);
  \draw [fill] (-1.7,0) circle (1.5pt);
\end{scope}
\end{tikzpicture}
\caption{$\ve= 0.15$}
\label{subfig: mixed noise TikZ}
\end{subfigure}
\begin{subfigure}{0.49\textwidth}
\centering
    \begin{tikzpicture}[line cap=round, line join=round, >=Triangle]
     \centering
    \begin{scope}[scale=1]
  \clip(-2.9,-2.6) rectangle (2.8,2.8);
  \draw(0,0) circle (2cm);
  \draw[dashed] (0,0)-- (1.4, 0);
   \draw[dashed] (0,0)-- (-1.4, 0);
\draw [->] (0,0)-- (1.4,1.4);
  \draw [->] (0,0)-- (1.4,-1.4);
  \draw [->] (0,0)-- (-1.4,1.4);
  \draw [->] (0,0)-- (-1.4,-1.4);
  \draw [dotted] (1.4,-1.4)-- (1.4,1.4);
  \draw [dotted] (-1.4,-1.4)-- (-1.4,1.4);
  \draw (-2,0) node[anchor=east] {\small{$\ket{{1}}$}};
  \draw (2,0) node[anchor=west] {\small{$\ket{{2}}$}};
  \draw (0,2) node[anchor=south] {\small{$ {|\psi\rangle}$}};
  \draw (0,-2) node[anchor=north] {\small{$|{\psi^{\perp}}\rangle$}};
  \draw [fill] (1.4,0) circle (1.5pt);
  \draw [fill] (-1.4,0) circle (1.5pt);

  \draw [fill] (0, 1.4) circle (1.5pt);
  \draw[dashed] (0,0)-- (0, 1.4);
\draw [->] (0,0)-- (1.4, 1.4);
  \draw [->] (0,0)-- (-1.4, 1.4);
  \draw [dotted] (-1.4, 1.4)-- (1.4, 1.4);
\end{scope}
\end{tikzpicture}
\caption{$\ve= \ve^{*}= 1- 1/\sqrt{2}$.}
\label{subfig: mixed noise TikZ ep*}
\end{subfigure}

\caption{\label{fig:wide} \emph{Joint decomposition of state and measurement}. A schematic of a possible attack by Eve when Alice holds a noisy state $\rho_\psi$ and a noisy projective measurement $\m_2=\{M_1, \, M_2\}$, each with noise parameter $\ve$. The filled circles represent $\rho_{\psi}$ and $\{M_1, \, M_2\}$, while the arrows represent Eve's decompositions $\ket{\varphi_{i, \lambda}}$ and $\{N_{1, j, \lambda}, \, N_{2, j, \lambda}\}$}.
\end{figure*}
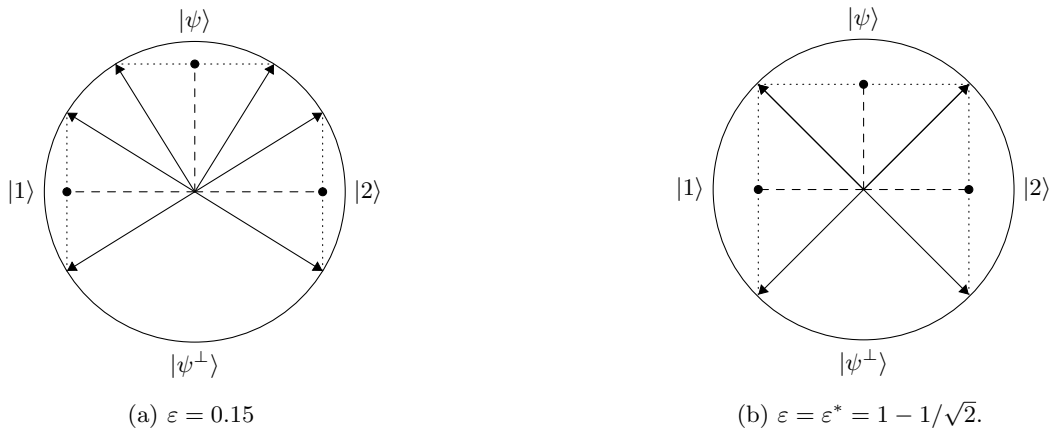

\textit{ Other entropies}.---
We have solved Eve's guessing probability for the noisy projective measurement $\m_d$, minimized over all pure states. In terms of the conditional min-entropy, we have
\begin{equation}\label{eqn: Hmin bound}
 \hmins \left(X | E, \, \m_d \right) = \log d - \log \Big(  \tr \sqrt{M_1}\Big)^2\,.  
\end{equation}
From \cite{Meng_2024}, we find that $\pgs \left( \m_d\right)$ is equal to the optimal guessing probability $\pgs \left( \rho_\psi \right)$ of the analogous noisy pure state $\rho_\psi$, 
\begin{equation}
    \rho_{\psi} = \left( 1 - \ve \right) \ketbra{\psi}{\psi} + \frac{\ve}{d}\id\,,
\end{equation}
when minimized over all measurements in orthogonal bases.
This begs the question: Is the maximal intrinsic randomness of $\m_d$ and $\rho_\psi$ equivalent when characterized by \emph{any} conditional entropy? 

As a first step towards solving this problem, we present upper bounds for the conditional von Neumann and max-entropies of $\m_d$, given the unbiased state $\ket{\psi}$. The conditional von Neumann entropy is relevant in the asymptotic limit of many measurement rounds \cite{Tomamichel_2009, Dupuis_2020, dai2023intrinsic}, while the conditional max-entropy has an operational interpretation as the security of $X$ when used as a secret key \cite{Konig_2009}. In Appendix \ref{app: dilation}, we use the `square root decomposition' $\{\mathcal{K}_j\}$ from Lemma \ref{lemma: qudit lower} to write a generalized Naimark dilation for $\m_d$ that maximizes Eve's guessing probability when Alice chooses $\ket{\psi}$. Since we allow Eve to minimize the conditional entropy over all possible dilations, the value arising from this dilation upper bounds the optimal one.   

Concretely, for the conditional von Neumann entropy, we find
\begin{equation}\label{eqn: vonN bound}
\begin{aligned}
 H \left(X | E, \,\ket{\psi}, \, \m_d \right) \leq  H_2 &\left( \pgs \big(  \m_d  \big) \right) 
 \\
  + &\left( 1 - \pgs \big(  \m_d  \big) \right) \log \left( d-1\right)\,,    
\end{aligned}
\end{equation}
where $H_2 \left(x \right)= - x \log x - \left( 1-x\right) \log \left( 1-x\right)$ is the binary Shannon entropy. This quantity is generally larger than the maximal conditional von Neumann entropy of the state $\rho_{\psi}$, $H^* \left( \rho_{\psi} \right) = \log d - S \left( \rho_\psi\right)$, where $S (\sigma)= - \tr \big( \sigma \log \sigma \big) $ is the von Neumann entropy. In the qubit case, the right-hand side of \eqref{eqn: vonN bound} is the coherence of formation of the state $\rho_{\psi}$ with respect to the unbiased basis $\{\ket{x}\}$ \cite{Yuan_2015}. For the conditional max-entropy, meanwhile, we find 
\begin{equation}\label{eqn: hmax bound}
    \hmax \left( X | E, \, \ket{\psi}, \, \m_d \right) \leq \log \left( d- (d-1)\ve \right) = \hmaxs \left( \rho_{\psi} \right)\,.
\end{equation}
We leave open several questions: (i) whether, as we might expect, the unbiased state $\ket{\psi}$ is the optimal choice for Alice to maximize \emph{any} conditional entropy; (ii) whether our `square root decomposition' $\{\mathcal{K}_j\}$ is indeed optimal for Eve, such that the bounds \eqref{eqn: vonN bound} and \eqref{eqn: hmax bound} are saturated; and, if not, (iii) whether there exists a universal dilation with which Eve can minimize all the conditional entropies of $\m_d$ at once. Note that if conditions (i) and (ii) were true, an isotopic state would have the same max- and min-entropy as an isotropic measurement, for the same noise value, but a different von Neumann entropy.

\textit{Coarse-graining}.\textbf{---}Imagine that, when $d$ is even and greater than 2, rather than perform the POVM $\m_d$ directly, Alice amalgamates the first half of its outcomes into a single outcome, 1, neglecting to distinguish between them, and acts similarly on the second half, which form the outcome 2. This coarse-graining process yields a new two-outcome POVM, $\hat{\m}$. From Corollary \ref{corr: qudit 2-outcome}, we find that $\pgs \big( \hat{\m}\big) \leq \pgs \left( \m_2 \right)$. We prove in Appendix \ref{app: coarse grain} that this inequality is saturated, so the higher dimension of $\hat{\m}$ offers Alice no extra randomness compared to the qubit POVM $\m_2$. Moreover, this case study provides an insight into the eavesdropper's behaviour. We find that it is not necessarily optimal for Eve to follow Alice's suit and coarse-grain her best attack for the original POVM $\m_d$; at least when Alice uses the unbiased state $\ket{\psi}$, Eve should instead `inflate' her optimal attack for $\m_2$ to fit the higher dimension $d$.

\textit{Noisy states and measurements}.--- Having solved for the guessing probability of a noisy projective measurement optimized over all pure states, we now ask, what changes if both the state and the measurement are noisy? Of course this comparison is especially meaningful when dealing with different realizations of the same measured probability distribution. In what follows, we study this question in the case where Alice holds the state $\rho_\psi$ and the measurement $\m_2$, both parameterized by $\ve$. 

We can consider an adversarial model where Eve can \emph{simultaneously} decompose $\rho_\psi$ into states $\ket{\varphi_{i, \lambda}}$ and $\m_2$ into POVMs $\mathcal{N}_{j, \lambda}=\{N_{1, j, \lambda}, \, N_{2, j, \lambda}\}$, according to some joint probability distribution $p(i, j, \lambda )$, where $\lambda$ is a classical hidden variable (see Figure \ref{subfig: mixed noise TikZ} for a schematic of such a decomposition). When the noise reaches $\ve = \ve^{*}$, where $\ve^{*}= 1 - 1/\sqrt{2}$, Eve can achieve perfect guessing probability, as she can choose rank-one measurements $\mathcal{N}_{j, \lambda}=\{N_{1, j, \lambda}, \, N_{2, j, \lambda}\}$ that align with the states $\ket{\varphi_{i, j}}$, as shown in Figure \ref{subfig: mixed noise TikZ ep*}. Further, when $\ve > \ve^{*}$, Eve can maintain perfection while faithfully reproducing $\rho_\psi$ and $\m_2$ if she randomizes her attack using $\lambda$: for a suitable distribution $p(\lambda)$, when $\lambda=1$, she continues to use the decomposition shown in Fig. \ref{subfig: mixed noise TikZ ep*}, while for $\lambda=2$, she employs its `mirror' decomposition, with all the vectors of Fig. \ref{subfig: mixed noise TikZ ep*} flipped across the horizontal axis (see Appendix \ref{app: mixed noise} for details). 
We find then that in this scenario, there is no intrinsic randomness whatsoever when $\ve \geq \ve^{*}$. This contrasts sharply with the case where $\m_2$ acts on the pure state $\ket{\psi}$ (or where the basis $\{\ket{x}\}$ acts on $\rho_\psi$ \cite{Meng_2024}), as then Eve can only achieve perfect guessing probability if the measurement (state) is maximally mixed.

\begin{figure}[h]
    \centering
    \begin{tikzpicture}
\begin{scope}[scale=0.75]
    \begin{axis}[
        axis lines=middle,
        xlabel={$\delta$},
        ylabel={$\pg$},
        xlabel style={at={(ticklabel* cs:1.02)}, anchor=west},
        ylabel style={at={(ticklabel* cs:1.02)}, anchor=south},
        ymin=0, ymax=1,
        xmin=0, xmax=1,
        xtick={0,0.5,1},
        ytick={0,0.5,1},
        domain=0:1,
        samples=200,
        legend pos=south west,
        legend style={font=\small}]
        \addplot[
            thick,
            black
        ] {x < 0.5 ? 0.5 * (1 + 2 * sqrt(x) * sqrt(1 - x)) : 1};
        \addlegendentry{shared noise lower bound};
        \addplot[
            dashed,
            black
        ] {0.5 * (1 + sqrt(x * (2 - x)))};
        \addlegendentry{single noise};
    \end{axis}
\end{scope}
\end{tikzpicture}
    \caption{\emph{Shared versus single noise}. Eve's optimal guessing probability for the single noise case (dashed) and a lower bound on her guessing probability in the shared noise case (undashed).}
\label{subfig: noise plot}
\end{figure}
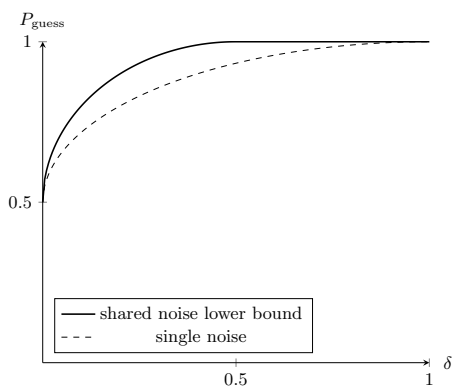

To compare these `single' and `shared' noise scenarios, we introduce a  parameter to quantify the total noise in Alice's measuring devices, $\delta= \ve \left( 2- \ve\right)$, so chosen because it preserves the Born rule between the two scenarios, i.e.
\begin{equation}
   \tr \big( \Delta_{\ve} \left( \rho \right) \, \Delta_{\ve} \left( M_x \right)\big)= \tr \big( \Delta_{\delta} \left( \rho \right) \, M_x \big) = \tr \big( \rho \,\Delta_{\delta} \left( M_x \right) \big) \,
\end{equation}
for all states $\rho$ and all POVMS $\m = \{M_x\}_x$ that satisfy $\tr M_x =1$ for all $x$. Notice that, for all $\ve$, $p(1)=\tr (\Delta_{\ve} \left( \ket{\psi}\bra{\psi} \right) \, \Delta_{\ve} \left( M_1 \right))=1/2$, so there is maximal statistical randomness. We find a lower bound on Eve's guessing probability in the shared noise case that is larger, for all values of $\delta$, than her guessing probability in the single noise case. We compare these values in Figure \ref{subfig: noise plot}, where the plateau of perfect probability emerges at $\delta=1/2$.

\textit{Discussion}.--- In this work, we solved the maximal intrinsic randomness of two classes of noisy quantum measurements, both of practical relevance. Notably, the maximal intrinsic randomness of a noisy projective measurement is equal to that of a noisy pure state (optimized over projective measurements), so when measuring a pure state in an unbiased orthogonal basis, we are free to consider either the state \emph{or} the measuring device to be the noisy one. However, our qubit case study shows that it is not equivalent to consider partial noise shared between the devices, as this scenario may give an eavesdropper significantly more guessing power. Our results may prove relevant in bounding the functionality of device-dependent QRNGs. Furthermore, we expect they will pave the way to better understand the generalized Naimark dilation in an adversarial scenario and to solve the maximal intrinsic randomness of more general quantum measurements.
\\

\section*{Acknowledgments}
We thank Victoria J. Wright for her early contributions to this work and Máté Farkás and Leonardo Zambrano for helpful discussions. FC and AA acknowledge funding from the Government of Spain (Severo Ochoa CEX2019-000910-S, Torres Quevedo PTQ2021-011870, NextGenerationEU PRTR-C17.I1 and FUNQIP), the European Union (QSNP 101114043 and NEQST 101080086 and Quantera Veriqtas), Fundació Cellex, Fundació Mir-Puig, Generalitat de Catalunya (CERCA program), ERC AdG CERQUTE, the AXA Chair in Quantum Informacion Science and Ayuda PRE2022-101448 financiada por MCIN/AEI/ 10.13039/501100011033 y por el FSE+. M. M. and M. S. were supported by the European Union's Horizon 2020 research and innovation programme under the Marie Skłodowska-Curie project `AppQInfo' No. 956071. M. S. was supported by the National Science Centre `Sonata Bis' Project No. 2019/34/E/ST2/00273, and by the QuantERA II Programme which has received funding from the European Union's Horizon 2020 research and innovation programme under Grant Agreement No. 101017733, Project `PhoMemtor' No. 2021/03/Y/ST2/00177.

\bibliography{newrefs}

\begin{widetext}

\newpage    

\appendix
\setcounter{lemma}{0}
\setcounter{corollary}{0}
\setcounter{figure}{0}

\section{The guessing probability }
\label{app: SDP}

\subsection{The primal problem }\label{app: primal}
For a fixed POVM $\m_S=\{M_{x, S}\}_{x}$ in finite dimension $d$, we wish to find the maximum randomness that can be generated from it, as quantified by the probability that an eavesdropper, Eve, correctly guesses the outcomes of the POVM when an experimentalist, Alice, applies it to the optimal state $\rho_S$, with which Eve may share quantum correlations. To formulate this guessing probability, we use a generalized Naimark dilation, introduced in \cite{frauchiger_2013}, which differs from a typical Naimark dilation as the auxiliary state may be mixed, and therefore correlated with Eve. 
We use as our figure of merit the following `quantum' guessing probability for Eve, from \cite{Senno_2023},
\begin{equation}\label{eqnA: pguess_quantum}
\begin{aligned}
\pgq \big( \rho_S,\, \m_S \big)= \; &  \max_{\{\Pi_{x, SA}\}_{x}, \, \ket{\phi}_{SAE}, \, \{M_{x, E}\}_{x}}
& & \sum_{x}\bra{\phi}\Pi_{x, SA}\otimes M_{x, E}\ket{\phi}_{SAE} \\
& \qquad \quad \text{subject to}
& & \tr_{AE} \left( \ketbra{\phi}{\phi}_{SAE} \right) =\rho_S \\
& 
& & \tr_A \Big( \Pi_{x, SA}\big(\id_S\otimes \tr_{SE} \left( \ketbra{\phi}{\phi}_{SAE} \right) \big) \Big)=M_{x, S} \; \text{ for all } \; x \\
& 
& & \langle{\phi|\Pi_{x, SA}\otimes \id_E|\phi} \rangle_{SAE}=\tr\big( M_{x, S} \, \rho_S \big)\,.
\end{aligned}
\end{equation}
We want to determine the minimum of this guessing probability over all of Alice's possible states $\rho_S$, i.e. to solve
\begin{equation}\label{eqnA: pguess_quantum_opt}
    {\pgsq} \big( \m_S \big) = \min_{\rho_{S}}\, \pgq \big( \rho_S, \,\m_S  \big)\,.
\end{equation}
The following lemma tells us that the optimal state $\rho_S$ in \eqref{eqnA: pguess_quantum_opt} can be taken pure.
\begin{lemma}
    For every POVM $\m_S$, the function $\pgq \big( \rho_S,\, \m_S \big)$ is concave in the variable $\rho_S$.
\end{lemma}
\begin{proof}
Let $\rho_S= \lambda \rho^{0}_{S}+ \left( 1 - \lambda \right) \rho^{1}_{S}$, where $0 \leq \lambda \leq 1$, and let $\left\{ \{ \Pi^{y}_{x, SA}\}_x\,, \; \ket{\phi^{y}}_{SAE}\,, \; \{M^{y}_{x, E}\}_x \right\}$ be an optimal solution to $\pgq \big( \rho^{y}_S,\, \m_S \big)$, such that $\tr_{AE} \ketbra{\phi^{y}}{\phi^{y}}_{SAE}= \rho^{y}_S$. Define 
\begin{equation}\label{eqnA: decomp mixed}
\begin{aligned}
    \rho_{SAE F_1 F_2} &= \lambda \ketbra{\phi^{0}}{\phi^{0}}_{SAE} \otimes \ketbra{0 \, 0}{0 \, 0}_{F_1 F_2} + \left( 1 - \lambda\right) \ketbra{\phi^{1}}{\phi^{1}}_{SAE} \otimes \ketbra{1\, 1}{1 \,1}_{F_1 F_2}\,,
    \\
    \Pi_{x, \,SAF_1} &= \Pi^{0}_{x, \, SA} \otimes \ketbra{0}{0}_{F_1} + \Pi^{1}_{x,  SA} \otimes \ketbra{1}{1}_{F_1}\,,
    \\
    M_{x, E F_2} &= M^{0}_{x, \, E} \otimes \ketbra{0}{0}_{F_2} + M^{1}_{x, E} \otimes \ketbra{1}{1}_{F_2}\,,
\end{aligned}    
\end{equation}
where $A F_1$ and $E F_2$ constitute the new auxiliary for the measurement and Eve's new system, respectively. For convenience, we write $\rho_{SAEF_1 F_2}$ as a mixed state in \eqref{eqnA: decomp mixed}, but we assume that Eve holds its purification. This decomposition is valid because 
\begin{align}
    \tr_{AEF_1 F_2} \rho_{SAEF_1 F_2} &= \lambda \rho^{0}_S + \left( 1 - \lambda \right) \rho^{1}_S = \rho_S\,, 
    \\
 \tr_{AF_1} \Big( \Pi_{x, SAF_1} \big( \id_{S} \otimes \tr_{SEF_2 } \rho_{SAEF_1 F_2}\big) \Big)  &= \lambda \tr_A \Big( \Pi^{0}_{x, SA}\big(\id_S\otimes \tr_{SE} \left( \ketbra{\phi^{0}}{\phi^{0}}_{SAE} \right) \big) \Big) \nonumber
 \\
 & \quad 
 + \left( 1 - \lambda \right) \tr_A \Big( \Pi^{1}_{x, SA}\big(\id_S \otimes \tr_{SE} \left( \ketbra{\phi^{1}}{\phi^{1}}_{SAE} \right) \big) \Big) = M_{x, S} \;\; \textnormal{for all} \;\; x\,, 
 \\
\tr \left( \Pi_{x, SA} \otimes \id_{EF_2} \; \rho_{SAEF_1 F_2} \right) &=  \lambda \langle{\phi^{0}|\Pi^{0}_{x, SA}\otimes \id_E|\phi^{0}} \rangle_{SAE} + \left( 1 - \lambda \right) \langle{\phi^{1}|\Pi^{1}_{x, SA}\otimes \id_E|\phi^{1}} \rangle_{SAE} \nonumber
\\
&= \tr\big( M_{x, S} \, \rho_S \big)\;\;\;\; \textnormal{for all} \;\; x\,.
 \end{align}
We aim to show that $\left\{ \{ \Pi_{x, SAF_1}\}_x\,, \; \rho_{SAEF_1 F_2}\,, \; \{M_{x, E F_2}\}_x \right\}$ is a solution to $\pgq \big( \rho_S,\, \m_S \big)$ such that
\begin{equation}
    \pgq \left( \rho_S,\, \m_S , \, \left\{ \{ \Pi_{x, SAF_1}\}_x\,, \; \rho_{SAEF_1 F_2}\,, \; \{M_{x, E F_2}\}_x \right\}\right) = \lambda  \pgq \big( \rho^{0}_S,\, \m_S \big) + \left( 1 - \lambda \right) \pgq \big( \rho^{1}_S,\, \m_S \big)\,.
\end{equation}
We have 
\begin{align}
&\pgq \left( \rho_S,\, \m_S , \, \left\{ \{ \Pi_{x, SAF_1}\}_x\,, \; \rho_{SAEF_1 F_2}\,, \; \{M_{x, E F_2}\}_x \right\} \right) = \sum_x \tr \Big( \Pi_{x, SAF_1} \otimes M_{x, EF_2} \; \rho_{SAEF_1 F_2}  \Big)
\\
&= \lambda \sum_{x}\bra{\phi^{0}}\Pi^{0}_{x, SA}\otimes M^{0}_{x, E}\ket{\phi^0}_{SAE}  +  \left( 1 - \lambda \right) \sum_{x}\bra{\phi^{1}}\Pi^{1}_{x, SA}\otimes M^{1}_{x, E}\ket{\phi^{1}}_{SAE} 
\\
&=  \lambda \pgq \big( \rho^{0}_S,\, \m_S \big) + \left( 1 - \lambda \right) \pgq \big( \rho^{1}_S,\, \m_S \big)\,.
\end{align}
Since 
\begin{equation}
    \pgq \big( \rho_S,\, \m_S \big) \geq \pgq \left( \rho_S,\, \m_S , \, \left\{ \{ \Pi_{x, SAF_1}\}_x\,, \; \rho_{SAEF_1 F_2}\,, \; \{M_{x, E F_2}\}_x \right\} \right)\,,
\end{equation}
we have proven that   
\begin{equation}
    \pgq \big( \rho_S,\, \m_S \big) \geq \lambda \pgq \big( \rho^{0}_S,\, \m_S \big) + \left( 1 - \lambda \right) \pgq \big( \rho^{1}_S,\, \m_S \big)\,
\end{equation}
for any convex decomposition of $\rho_{S}$, so the function $\pgq \big( \rho_S,\, \m_S \big)$ is concave in $\rho_S$ for any POVM $\m_S$.
\end{proof}
Since $\pgq \big( \rho_S,\, \m_S \big)$ is concave in $\rho_S$ for any POVM $\m_S$ and since the set of states $\{\rho_S\}$ is convex, for any fixed $\m_S$, the minimization of $\pgq \big( \rho_S,\, \m_S \big)$ over all $\rho_S$ must be achieved by a pure state $\rho_S = \ketbra{\phi}{\phi}_S$. When Alice chooses a pure state $\ket{\phi}_S$, the guessing probability \eqref{eqnA: pguess_quantum} reduces to
\begin{equation}\label{eqnA: pguess_quantum_pure}
\begin{aligned}
\pgq \big( \ket{\phi}_S,\, \m_S \big)= \; &  \max_{\{\Pi_{x, SA}\}_{x}, \, \ket{\varphi}_{AE}, \, \{M_{x, E}\}_{x}}
& & \sum_{x}\bra{\phi}\Pi_{x, SA}\otimes M_{x, E}\ket{\phi}_{SAE} 
\\
& \qquad \quad  \text{subject to}
& & \tr_A \Big( \Pi_{x, SA} \big( \id_S\otimes \tr_{E} \left( \ketbra{ \varphi}{\varphi}_{AE} \right) \big) \Big)=M_{x, S} \;  \text{ for all } \; x\,.
\end{aligned}
\end{equation}
Since $\ket{\phi}_S$ is pure, the guessing probability \eqref{eqnA: pguess_quantum_pure} reduces to the guessing probability of a classically correlated eavesdropper \cite{Senno_2023} who decomposes the POVM $\m_{S}= \{M_{x, S}\}_{x}$  into `sub'-POVMs $\mathcal{N}_{j, S}=\{N_{x, j, S}\}_x$ indexed by $j$, according to some probability distribution $\{p(j)\}$. We have the following definition for Eve's classical guessing probability,  
\begin{equation}\label{eqnA: pguess_classical_pure}
\begin{aligned}
\pgc \big( \ket{\phi}_S,\, \m_S \big)= \; &   \max_{\{p(j)\}, \, \{ \mathcal{N}_{j, S}\}}
& & \sum_{j}p(j) \max_x \bra{\phi} N_{x,j, S}\ket{\phi}_S
\\
& \;\;\; \text{subject to}
& &  N_{x, j, S} \geq 0 \; \text{for all} \; x, j 
\\
& 
& & \sum_{x} N_{x, j, S} = \id_S
\\
& 
& & \sum_{j} p(j)N_{x,j, S} = M_{x, S} \; \text{ for all } \; x\,. \\
\end{aligned}
\end{equation}
Our figure of merit is now the minimization
\begin{equation}\label{eqnA: pguess_optimal}
    \pgs \left( \m_{S} \right) = \min_{\ket{\phi}} \pg \big( \ket{\phi},\, \m_S  \big)\,.
\end{equation}
From here on, we drop the `classical' superscript C and the subscript $S$ denoting Alice's system. The optimization \eqref{eqnA: pguess_classical_pure} can be cast as a semidefinite programming (SDP) problem by defining the subnormalized decomposition $\{\mathcal{K}_j\}$ with elements
\begin{equation}
    K_{x, j} = p(j) N_{x, j}\,,
\end{equation}
as we have 
\begin{equation}\label{eqnA: pguess_unnorm}
\begin{aligned}
\pg \big( \ket{\phi},\, \m \big)= \; &   \quad \;  \max_{\{ \mathcal{K}_j\}}
& & \sum_{j} \max_x \bra{\phi} K_{x,j}\ket{\phi}
\\
& \;\;\; \text{subject to}
& &  K_{x, j} \geq 0  \; \text{ for all } \; x, j
\\
& 
& & \sum_{x}  d K_{x, j} =  \id \sum_{x}  \tr K_{x, j} \; \text{ for all } \; j 
\\
& 
& & \sum_{j} K_{x,j} = M_{x} \text{ for all } x\,. 
\end{aligned}
\end{equation}
To prove equivalence, note that we can recover \eqref{eqnA: pguess_classical_pure} from \eqref{eqnA: pguess_unnorm} by defining the variables
\begin{equation}
    p(j)= \frac{1}{d} \sum_{x} \tr K_{x, j}\,, \qquad N_{x, j} = \frac{1}{p(j)} K_{x, j}\,
\end{equation}
when $\sum_{x} \tr K_{x, j} \neq 0$. Furthermore, following an argument made in \cite{Law_2014}, we are free to consider the number of POVMs, $n$, in Eve's decomposition to be equal to the number of outcomes, $m$, in $\m$. To see this, imagine that $n > m$ and that the POVMs $\mathcal{K}_k=\{K_{x, k}\}_x$ and $\mathcal{K}_l=\{K_{x, l}\}_x$ correspond to the same most probable outcome, $y$, i.e.
\begin{equation}
    y = \argmax_{x} \; \langle \phi | K_{x, k} | \phi \rangle = \argmax_{x} \; \langle \phi | K_{x, l} | \phi \rangle\,.
\end{equation}
Eve could then combine these to form a new POVM $\mathcal{K}'_j=\{ {K}'_{x, j}   \}_{x}$, with
\begin{equation}
    \{ {K}'_{x, j} \}_{x} = \{ K_{x, k} \}_{x} + \{ K_{x, l} \}_{x}\,,
\end{equation}
without affecting her guessing probability. She could then repeat this process until she had whittled her $n$ original POVMs down to $m$. 
On the other hand, if $n < m$, Eve loses nothing by splitting one of her POVMs $\mathcal{K}_j=\{K_{x, j}\}_x$ into a convex mixture of identical POVMs $\mathcal{K}'_k=\{K_{x, k}\}_x$ and $\mathcal{K}'_{l}= \{K_{x, l} \}_x$, i.e.
\begin{equation}
    \{K'_{x, k}\}_x= \frac{1}{2}\{K_{x, j}\}_x\,, \qquad  \{K'_{x, l}\}_x= \frac{1}{2}\{K_{x, j}\}_x\,. 
\end{equation}
As before, she could repeat this process until she reached $n=m$ POVMs. Fixing $n=m$, we are then free to drop the inner maximization over $x$ in \eqref{eqnA: pguess_unnorm} by labelling Eve's POVMs such that
\begin{equation}
 j= \argmax_x \bra{\phi} K_{x,j}\ket{\phi} \; \text{for all} \; x \,.
\end{equation}
Our final SDP has the more tractable form
\begin{equation}\label{eqnA: pguess_main}
\begin{aligned}
\pg \big( \ket{\phi},\, \m \big)= \; &   \quad \;  \max_{ \{ \mathcal{K}_j\}}
& & \sum_{j}  \bra{\phi} K_{j,j}\ket{\phi}
\\
& \;\;\; \text{subject to}
& &  K_{x, j} \geq 0  \; \text{ for all } \; x, j
\\
& 
& & \sum_{x}  d K_{x, j} =  \id \sum_{x}  \tr K_{x, j} \; \text{ for all } \; j 
\\
& 
& & \sum_{j} K_{x,j} = M_{x} \text{ for all } x\,. 
\end{aligned}
\end{equation}

\subsection{The dual problem } \label{app: dual}
Here we derive the dual for the SDP problem \eqref{eqnA: pguess_main} (we refer to \cite{Skrzypczyk_2023} for a detailed description on deriving dual problems).
The Lagrangian for \eqref{eqnA: pguess_main} can be written as 
\begin{equation}
    \begin{aligned}
    \mathcal{L}&= \sum_{j}  \tr \left( K_{j, j}\, \ketbra{\phi}{\phi} \right) + \sum_{x} \tr \bigg(Y_{x} \Big(  M_{x}- \sum_{j}K_{x, j}\Big) \bigg) + \sum_{j} \tr \bigg(F_{j}  \sum_{x}\big(\id \tr -d \big) K_{x, j} \bigg) + \sum_{x, j} \tr \left(Z_{xj} \,K_{x, j} \right) 
    \\
    &= \sum_{x, j} \tr \bigg( K_{x, j} \Big(\delta_{x, j}\ketbra{\phi}{\phi} - Y_{x} + \left( \id\tr -d \right)F_j + Z_{xj}\Big)\bigg) + \sum_{x} \tr \big( Y_{x}M_x \big)\,,    
    \end{aligned}
\end{equation}
where we introduced the dual variables $Y_x$, $F_j$ and $Z_{xj}$. We impose $Z_{xj} \geq 0$ such that 
\begin{equation}
    \mathcal{L} \geq \sum_{j}  \tr \left( K_{j, j}\, \ketbra{\phi}{\phi} \right)
\end{equation}
holds for any feasible point. If we now restrict the dual variables to satisfy
\begin{equation}
   \delta_{x, j}\ketbra{\phi}{\phi} - Y_{x} + \left( \id\tr -d \right)F_j + Z_{xj} =0 \;\; \textrm{for all} \; x,\,j\,, 
\end{equation}
we arrive at 
\begin{equation}
\mathcal{L}= \sum_{x} \tr \big( Y_{x}\, M_x \big)\,.    \end{equation}
The dual problem is then 
\begin{equation}\label{eqnA: dualF}
\begin{aligned}
\pg \big( \ket{\phi},\, \m \big)= \; &   \quad  \min_{ \{Y_x\}, \, \{F_j\}}
& & \sum_{x}  \tr \big( Y_x \, M_x \big)
\\
& \;\;\; \text{subject to}
& &  Y_{x} \geq \delta_{xj} \ketbra{\phi}{\phi}+ \left(\id \tr -d \right)F_j \;\; \text{for all} \;\; x,  j\,,
\end{aligned}
\end{equation}
where $\{F_j\}$ are slack variables. Note that, without loss of generality, we can describe any traceless Hermitian matrix $G_j$ in terms of a generic Hermitian matrix $F_j$ as
\begin{equation}
    G_{j} = \left(\id \tr -d \right)F_j\,,
\end{equation}
so we can simplify the dual to 
\begin{equation}\label{eqnA: dualG}
\begin{aligned}
\pg \big( \rho,\, \m \big)= \; &   \quad  \min_{\{Y_x\}, \, \{G_j\}}
& & \sum_{x}  \tr \big( Y_x \, M_x \big)
\\
& \;\;\; \text{subject to}
& &  Y_{x} \geq \delta_{xj} \ketbra{\phi}{\phi}+ G_j \;\; \text{for all} \;\; x,  j\,,
\\
& 
& &  \tr G_j =0 \;\; \text{for all} \;\; j\,.
\end{aligned}
\end{equation}
If we set $Y_x=\alpha\id$ and $F_{j}=0$ for all $x$ and $j$, with $\alpha>1$, we obtain a feasible point for \eqref{eqnA: dualF} that satisfies the strict constraint $Y_{x} > \delta_{xj} \ketbra{\phi}{\phi}+ \left(\id \tr -d \right)F_j$ for all $x$ and $j$. Moreover, the variables $K_{x, j}= \frac{1}{d} M_x$ are strictly positive for any POVM whose elements are all full-rank, and they satisfy the requirements for the primal SDP \eqref{eqnA: pguess_main}. From \cite[Conic Duality Theorem]{Nesterov1994}, if both dual and primal problems are strictly feasible, the duality gap for an optimal solution is zero, so for a full-rank POVM, we are free to use either the primal or dual formulations to solve our problem. 

Finally, the complementary slackness condition (see e.g. \cite[Section 5.5.3]{Boyd_2004}) is a useful necessary condition for a set of primal and dual variables to be optimal. For our problem, the condition is
\begin{equation}\label{eqnA: slackness}
    K_{x, j} \Big( Y_{x} - \delta_{xj} \ketbra{\phi}{\phi} - G_j    \Big) = 0 \quad \text{for all }\,\, x, j\,.
\end{equation}
We will use this condition later to prove that a set of dual variables $\{Y_x\}$ and $\{G_j\}$ satisfy all of their necessary constraints.

\section{Qubit measurements with two outcomes}\label{app: qubit two_outcome}

\subsection{Decomposition proof of Lemma \ref{lemma: qubit 2-outcome lower app}}\label{app: qubit decomposition}
\begin{lemma}\label{lemma: qubit 2-outcome lower app}
Let $\m= \{M_1,M_2\}$ be any qubit POVM with two outcomes, where $\tr M_1 \leq \tr M_2$. The following lower bound holds,
 \begin{equation}
  \pgs \big(  \m  \big) \geq  1 - \tr M_1 + \frac{1}{2} \Big( \tr \sqrt{M_1} \Big)^2 \,.
\end{equation}  
\end{lemma}
\begin{proof}
We note that in any two-outcome measurement, the POVM elements must share an eigenbasis. We denote this basis by $\{\ket{x}\}$, and represent $M_1$ by
\begin{equation}
    M_1 = m_1 \ketbra{1}{1} + m_2 \ketbra{2}{2}\,,
\end{equation} 
setting $m_1 \geq m_2$.
For a fixed state $\ket{\phi}$, Eve's optimal guessing probability can be bounded from below by taking any given decomposition of the POVM. Defining $p(x)=\langle \phi | M_x | \phi \rangle$, consider the decomposition
\begin{equation}\label{eqn: qubit_decomp_gen}
\begin{aligned}
 K_{1, 1}&=  \sqrt{M_{1}} \ketbra{\phi}{\phi} \sqrt{M_{1}}\,, 
    \\
    K_{2, 1} &= p(1) \id - K_{1, 1}  = p(1) \id -  \sqrt{M_{1}} \ketbra{\phi}{\phi} \sqrt{M_{1}}\,, 
    \\
    K_{1, 2} &=   M_1 - K_{1, 1} =   M_{1} - \sqrt{M_1} \ketbra{\phi}{\phi} \sqrt{M_1}   \,,
    \\
    K_{2, 2} &= p(2) \id - K_{1, 2}  = p(2) \id -   M_{1} + \sqrt{M_1} \ketbra{\phi}{\phi} \sqrt{M_1}   \,.   
\end{aligned}    
\end{equation}
The normalization condition is satisfied, 
\begin{equation}
    K_{1, 1}+ K_{2, 1}= p(1) \id\,, \qquad
    K_{1, 2}+ K_{2, 2}= p(2) \id\,,
\end{equation}
while we also have 
\begin{equation}
   K_{1, 1}+ K_{1, 2}= M_1\,, \qquad K_{2, 1}+ K_{2, 2}= \id - M_1 = M_2\,. 
\end{equation}
All that's left is to prove that each element is positive-semidefinite. $K_{1, 1}$ is rank-one, so it is positive-semidefinite because its trace is non-negative, 
\begin{equation}
    \tr K_{1, 1} = \langle \phi | M_1 | \phi\rangle = p(1)\,.
\end{equation}
$K_{2, 1}$ is rank-one and positive semi-definite because
\begin{equation}
    \lmin \left( K_{2, 1}\right) = \lmin \Big( p(1) \id - K_{1, 1} \Big) = p(1) - \tr K_{1, 1} = 0\,,
\end{equation}
where $\lmin\left(. \right)$ denotes the smallest eigenvalue of its argument. When $M_1$ is rank-one, $K_{2, 1}$ is rank-one and positive semidefinite because
\begin{equation}
    K_{1, 2} = M_{1} - \sqrt{M_1} \ketbra{\phi}{\phi} \sqrt{M_1} = m_1 \Big( 1 - \abs{\langle \phi | 1 \rangle }^2  \Big) \ketbra{1}{1} \geq 0\,,
\end{equation}
while when $M_1$ is full-rank, $K_{1, 2}$ has a zero eigenvector given by $M_{1}^{-1/2} \ket{\phi}$,
\begin{equation}
    K_{1, 2} \,  M_{1}^{-1/2} \ket{\phi} = \Big( 1 - \langle \phi | \id | \phi \rangle
 \Big)\sqrt{M_1} \ket{\phi}=0\,,
\end{equation}
so it is positive-semidefinite because its trace is non-negative,
\begin{equation}
    \tr K_{1, 2} = \tr M_1 - \langle \phi | M_1 | \phi \rangle \geq 0\,.
\end{equation}
Finally, $K_{2, 2}$ is positive-semidefinite because $K_{1, 2}$ is rank-one, so 
\begin{equation}
 \lmin \left( K_{2, 2} \right) = \lmin \Big( p(2) \id - \tr K_{1, 2} \Big) = p(2) - \tr K_{1, 2} = 1 - \tr M_1 \geq 0\,. 
\end{equation}
We see that $K_{2, 2}$ is rank-one only in the special case $\tr M_1 =1$. We represent $\ket{\phi}$ in the eigenbasis $\{\ket{x}\}$ of the POVM as $\ket{\phi}= \left(\phi_1, \phi_2 \right)^{T}$, where we are free to choose real $\phi_1$ and $\phi_2$ (if they were complex, a unitary of the form $U = \sum_{k} e^{i \theta_{k}} \ketbra{k}{k}$ could be applied to make them real, leaving the POVM $\m$ unchanged). From \eqref{eqnA: pguess_main}, the guessing probability is then bounded by
\begin{align}\label{eqn: qubit_pguess_decomp}
    \pg \left(\ket{\phi}, \,  \m \right) &\geq \langle {\phi | K_{1, 1} + K_{2, 2}  | \phi} \rangle  = \langle {\phi | M_2 | \phi} \rangle- \langle{\phi | M_1 | \phi} \rangle + 2\langle {\phi | \sqrt{M_1} | \phi} \rangle^2 
    \\
   &= 1 - 2\langle{\phi | M_1 | \phi} \rangle + 2\langle{\phi | \sqrt{M_1} | \phi} \rangle^2 = 1 - 2 \phi_1^2 \left( 1 -\phi_1^2 \right)^2 \big( \sqrt{m_1} - \sqrt{m_2} \big)^2\,,
\end{align}
where in the last line we use $\phi_2^2 = 1 - \phi_1^2$. The left-hand side of \eqref{eqn: qubit_pguess_decomp} is minimized when Alice chooses 
\begin{equation}
    \min_{\phi_1}\, \phi_1^2 \left( 1 -\phi_1^2 \right)^2 = \frac{1}{4}\,,
\end{equation}
which gives the bound
\begin{equation}
    \pgs \left( \m \right) \geq 1 - \frac{1}{2} \big( \sqrt{m_1} - \sqrt{m_2} \big)^2= 1 - \tr M_1 + \frac{1}{2} \left( \tr \sqrt{M_1} \right)^2\,,
\end{equation}
which is saturated for this decomposition if and only if $\phi_1^2= \phi_2^2= 1/2$ i.e. if $\ket{\phi}$ is unbiased to the eigenbasis $\{\ket{x}\}$ of the POVM, or in the trivial case where $M_1 \propto \id$, as $m_1=m_2$ and Eve has perfect guessing probability.
\end{proof}

\subsection{Vector proof of Lemma \ref{lemma: qubit 2-outcome lower app}}\label{app: qubit 2-outcome lower}
\begin{proof}
First, we prove by contradiction that Eve's optimal decomposition of any POVM cannot contain a full-rank element $K_{a, b}$ for $a \neq b$. If $K_{a, b}$ were full-rank, we could define 
\begin{equation}
  \epsilon =  \lmin \big( K_{a, b} \big) > 0\,,
\end{equation}
where $\lmin(.)$ denotes the smallest eigenvalue of its argument. We could create a new decomposition with POVMs $\mathcal{K}'_j=\{\tilde{K'}_{x, j} \}_{x}$ where all of the elements are unchanged apart from
\begin{equation}
    \tilde{K'}_{a, b}=  K_{a, b} - \epsilon \id\,, \qquad \tilde{K'}_{a, a} = K_{a, a} + \epsilon \id\,.
\end{equation}
This decomposition satisfies all the necessary conditions,
\begin{equation}
    \tilde{K'}_{x, j} \geq 0\,, \qquad \sum_{j}\tilde{K'}_{x, j} = M_x\,, \qquad
    \sum_{x} \tilde{K'}_{x, b} = \sum_{x} K_{x, b} - \epsilon \id \propto \id\,, \qquad
    \sum_{x} \tilde{K'}_{x, a} = \sum_{x} K_{x, a} + \epsilon \id \propto \id\,.
\end{equation}
Eve would achieve a strictly higher guessing probability using $\{\mathcal{K}'_j\}$, as 
\begin{equation}
    \pg \big( \ket{\phi},\,  \m, \, \{\mathcal{K}'_j \}\big) = \bra{\phi} \Big( \sum_{j} K_{j, j} + \epsilon \id \Big) \ket{\phi} = \pg \big( \ket{\phi}, \, \m, \, \{\mathcal{K}_j \}\big) + \epsilon\,, 
\end{equation}
so the decomposition $\{\mathcal{K}_j\}$ must not have been optimal.

In the case of a qubit POVM with two outcomes, this result allows us to restrict to decompositions where the elements $K_{x, j}$ are rank-one for $x \neq j$. We denote the POVM elements by $M_1$ and $M_2$, and assume throughout that $0 < \tr M_1 \leq \tr M_2$. Note that in any two-outcome measurement, the POVM elements must share an eigenbasis. We denote this basis by $\{\ket{x}\}$, and represent $M_1$ by 
\begin{equation}
    M_1 = m_1 \ketbra{1}{1} + m_2 \ketbra{2}{2}\,,
\end{equation}
fixing $m_1 \geq m_2$.
Choosing the conventions $\lambda_1 \geq \lambda_2 \geq 0$ and $\mu_1 \geq \mu_2 \geq 0$, we can then define the decomposition
\begin{equation}
    \begin{aligned}
      K_{1, 1} &= \lambda_{1}\ketbra{\lambda_1}{\lambda_1} + \lambda_2 \ketbra{\lambda_2}{\lambda_2} \,, 
    \\
    K_{2, 1} &= \big( \lambda_1 - \lambda_2 \big) \ketbra{\lambda_2}{\lambda_2} \,,
    \\
    K_{1, 2} &= \big( \mu_1 - \mu_2 \big) \ketbra{\mu_2}{ \mu_2} \,, 
    \\
    K_{2, 2} &= \mu_1 \ketbra{\mu_1}{\mu_1} + \mu_2 \ketbra{\mu_2}{\mu_2} \,.  
    \end{aligned}
\end{equation}
The constraints in \eqref{eqnA: pguess_main} give us two conditions,
\begin{equation}
\begin{aligned}
    &\left( \lambda_1-\lambda_2\right) \ketbra{\lambda_1}{\lambda_1} - \left(\mu_1 - \mu_2 \right) \ketbra{\mu_1}{\mu_1} + \left( \lambda_2 + \mu_1 - \mu_2 \right) \id  = \left( m_1-m_2 \right) \ketbra{m_1}{m_1} + m_2\id\,,
    \\
    & \; \; \lambda_1 + \mu_1 = 1\,.
\end{aligned}
\end{equation}
It is convenient to switch to Bloch vector notation, so we introduce vectors $\vecr_{a}$ such that
 \begin{equation}
     \ketbra{a}{a} = \frac{1}{2} \left( \id + \vecr_{a} \cdot \vec{\sigma} \right)\, 
 \end{equation}
for $a= \lambda_1, \mu_1, 1$ and $\phi$, where $\abs{\vecr_{a}}^{2}=1$, $\vec{\sigma}$ is the vector of Pauli matrices and $\vecr_{\phi}$ represents Alice's chosen state. We can now rewrite Eve's POVMs as 
\begin{equation}
\begin{aligned}
    K_{1, 1} &= \frac{\lambda_1 + \lambda_2}{2} \id + \frac{\lambda_1-\lambda_2}{2} \, \vecr_{\lambda_1} \cdot \sigma\,, 
    \\
    K_{2, 1} &= \frac{\lambda_1 - \lambda_2}{2} \, \left(\id -\vecr_{\lambda_1} \cdot \sigma \right) 
 \,, 
    \\
    K_{1, 2} &= \frac{\mu_1 - \mu_2}{2} \, \left( \id - \vecr_{\mu_1} \cdot \sigma \right)   \,, 
    \\
    K_{2, 2} &= \frac{\mu_1 + \mu_2}{2}\id + \frac{\mu_1 - \mu_2}{2} \, \vecr_{\mu_1} \cdot \sigma  \,,
\end{aligned}    
\end{equation}
such that the optimal guessing probability is 
\begin{equation}\label{eqnA: vec_pguess_original}
\begin{aligned}
\pgs \big(  \m  \big) = & \min_{\vecr_{\phi}} \; \max_{\{\lambda_i, \, \mu_i, \, \vecr_{\lambda_1}, \, \vecr_{\mu_1} \} } 
& & \frac{\lambda_1  + \lambda_2 + \mu_1
+ \mu_2}{2}  + \Big( \frac{(\lambda_1-\lambda_2)\vecr_{\lambda_1}  + (\mu_1 - \mu_2)\vecr_{\mu_1}}{2} \Big) \cdot \vecr_{\phi} 
\\
& \quad \text{subject to}
& & \lambda_1 \geq \lambda_2 \geq 0 \,, \quad \mu_1 \geq \mu_2 \geq 0  
\\
& 
& & \lambda_1 + \mu_1 = 1\   
\\
& 
& & \lambda_1 + \lambda_2 + \mu_1 -\mu_2 = \tr M_1
\\
& 
& & \left( \lambda_1-\lambda_2 \right) \vecr_{\lambda_1} - \left( \mu_1 - \mu_2 \right)  \vecr_{\mu_1} = \left( m_1-m_2 \right) \vecr_{1}\,.
\end{aligned}
\end{equation}
Introducing the notation 
\begin{equation}
    a= \lambda_1 - \lambda_2\,, \quad b= \mu_1 - \mu_2\,, \quad z = m_1 -m_2\,,
\end{equation}
we can simplify \eqref{eqnA: vec_pguess_original} to
\begin{equation}\label{eqnA: pguess_vec_ab}
\begin{aligned}
\pgs \big(  \m  \big) = \; & \min_{\vecr_{\phi}} \; \max_{ a, \, b, \, \vecr_{\lambda_1}, \, \vecr_{\mu_1}  } 
& & 1 - \frac{a+b}{2} + \left( \frac{a\vecr_{\lambda_1} + b\vecr_{\mu_1}}{2} \right) \cdot \vecr_{\phi}
\\
& \quad \text{subject to}
& & a \geq 0\,,  \quad b \geq 0\,, \quad \tr M_1  \geq a+b  
\\
& 
& & a\vecr_{\lambda_1} - b\vecr_{\mu_1} = z\vecr_{1}\,.   
\end{aligned}
\end{equation}
To recover \eqref{eqnA: vec_pguess_original} from \eqref{eqnA: pguess_vec_ab}, note that, since $\tr M_1 \leq 1$, we have $a+b \leq \tr M_1 \leq 1$, so we can define
\begin{equation}
    \lambda_1 = \frac{1}{2} \Big( \tr M_1 + a -b \Big)
\end{equation}
such that
\begin{equation}
   1-b \geq \lambda_1 \geq a\,,  
 \end{equation}
and from there we can define 
\begin{equation}
    \mu_1 = 1 - \lambda_1 \geq 0\,, \quad \lambda_2 = \lambda_1 - a \geq 0\,, \quad \mu_2 = \mu_1 - b \geq 0\,, 
\end{equation}
so the remaining conditions follow.

We can bound \eqref{eqnA: pguess_vec_ab} from below by considering any decomposition by Eve that is valid for all $\vecr_{\phi}$. The case where Alice chooses $\vecr_{\phi}= \vecr_{1}$ is trivial, as Eve could achieve perfect guessing probability by taking $a=z, b=0, \vecr_{\lambda_1}=\vecr_{1}$ and arbitrary $\vecr_{\mu_1}$. For $\vecr_{\phi} \neq \vecr_{1}$, consider the decomposition 
\begin{equation}\label{eqnA: a and b}
    2a = \tr M_1 + z \cos \theta \sqrt{ \frac{ \big( \tr M_1 \big)^2 -z^2  }{ \big( \tr M_1 \big)^2 - z^2 \cos^2 \theta
  } } \,, \qquad
  2b = \tr M_1 - z \cos \theta \sqrt{ \frac{ \big( \tr M_1 \big)^2 -z^2  }{ \big( \tr M_1 \big)^2 - z^2 \cos^2 \theta
  } }\,,
\end{equation}
where we define $\cos \theta := \vecr_{\phi} \cdot \vecr_{1}$. These values saturate the inequality $\tr M_1 \geq a+b$, they do not diverge since $\cos \theta < 1$, and we see that they are non-negative by noting that 
\begin{equation}
    \tr M_1 = m_1 + m_2 \geq \abs{m_1-m_2} = \abs{z} \geq \abs{ z \, \cos\theta} \geq \pm \, z \cos\theta \sqrt{\frac{  \big(\tr M_1 \big)^2-z^2}{\big( \tr M_1 \big)^2-z^2 \cos^2\theta}}\,.
\end{equation}
Combining \eqref{eqnA: a and b} with the following vector decomposition for fixed $l$,
\begin{equation}\label{eqnA: vectors solved}
     a \, \vecr_{\lambda_1} = \frac{l \vecr_{\phi} + z \vecr_{1} }{2}\,, \qquad 
     b \, \vecr_{\mu_1} = \frac{l \vecr_{\phi}- z \vecr_{1} }{2}\,,
 \end{equation}
we see that $a$ and $b$ satisfy the following conditions which ensure the correct normalization of $\vecr_{\lambda_1}$ and $\vecr_{\mu_1}$,
\begin{equation}\label{eqnA: l}
    l^2 = 2 \left( a^2 + b^2 \right) -z^2\,, \qquad a^2 - b^2 = lz\cos\theta\,,
\end{equation}
so \eqref{eqnA: a and b}, \eqref{eqnA: vectors solved} and \eqref{eqnA: l} together form a valid decomposition for any $\vecr_{\phi} \neq \vecr_{1}$. These variables satisfy 
\begin{equation}
   a \vecr_{\lambda_1} + b \vecr_{\mu_1} = l \vecr_{\phi}\,,  
\end{equation}
so we find that the guessing probability is bounded by 
\begin{equation}
 \pgs \big(  \m  \big)  \geq 
 \min_{\vecr_{\phi}} \, \bigg(  1 - \frac{a+b}{2} + \frac{l}{2} \bigg)  = \min_{\vecr_{\phi}} \, \left(  1 - \frac{a+b}{2} + \frac{\sqrt{2 \left( a^2 + b^2 \right) -z^2   }}{2} \right)\,,
\end{equation}
with $a$ and $b$ given by \eqref{eqnA: a and b}. Using 
\begin{equation}
    2 \left( a^2 + b^2\right) \geq \left( a +b \right)^2\,,
\end{equation}
we achieve the following bound for all $\vecr_{\phi}$,

\begin{align}
 \pgs \big(  \m  \big)  &\geq     1 - \frac{a+b}{2} + \frac{\sqrt{\left( a + b \right)^2 -z^2   }}{2} =     1 - \frac{\tr M_1}{2} + \frac{\sqrt{\left( \tr M_1 \right)^2 -z^2   }}{2} 
 \\
 &=     1 - \frac{m_1 + m_2}{2} + \sqrt{m_1 m_2} =    1 - \tr M_1 + \frac{1}{2} \Big( \tr \sqrt{M_1} \Big)^2 \,,
\end{align}
proving Lemma \ref{lemma: qubit 2-outcome lower app}.
\end{proof}

\subsection{Proof of Lemma \ref{lemma: qubit 2-outcome upper app}}\label{app: qubit 2 outcome upper}
\begin{lemma}\label{lemma: qubit 2-outcome upper app}
Let $\m= \{M_x\}_x$ be any qubit POVM with two outcomes, where $\tr M_1 \leq \tr M_2$, and let $\ket{\psi}= \frac{1}{\sqrt{2}} \left(1, \, 1 \right)^{T}$. Then
\begin{equation}\label{eqn: lemma2 app}
  \pg \big( \ket{\psi}, \, \m  \big) \leq   1 - \tr M_1 + \frac{1}{2} \Big( \tr \sqrt{M_1} \Big)^2 \,.
\end{equation}  
\end{lemma}

\begin{proof}
Here, we borrow the vector formulation \eqref{eqnA: pguess_vec_ab} of the guessing probability from Appendix \ref{app: qubit 2-outcome lower}. Denote by $\vecr_{\psi}$ the Bloch vector of the state $\ket{\psi}$. Since $\ket{\psi}$ is unbiased to the basis $\{\ket{x}\}$, we must have $\vecr_{\psi} \cdot \vecr_{1}=0$. From the condition $a\vecr_{\lambda_1} - b\vecr_{\mu_1} = z\vecr_{1}$, we can then define 
\begin{equation}
   a \vecr_{\lambda_1} \cdot \vecr_{\psi} = b \vecr_{\mu_1} \cdot \vecr_{\psi} = h\,
\end{equation}
and write the guessing probability as
\begin{equation}\label{eqnA: pguess_vec_unbiased app}
\begin{aligned}
\pgs \big(  \m  \big) = \; & \min_{\vecr_{\phi}} \; \max_{a, \, b, \, \vecr_{\lambda_1}, \, \vecr_{\mu_1}  } 
& & 1 - \frac{a+b}{2} + h 
\\
& \quad \text{subject to}
& & a \geq 0\,,  \quad b \geq 0\,, \quad \tr M_1  \geq a+b  
\\
& 
& & h= a \vecr_{\lambda_1} \cdot \vecr_{\psi}
\\
& 
& & a\vecr_{\lambda_1} - b\vecr_{\mu_1} = z\vecr_{1}\,. 
\end{aligned}
\end{equation}
In the special case where $M_1 \propto \id$, $z=0$ and the second constraint of \eqref{eqnA: pguess_vec_unbiased app} disappears, so 
Eve is free to choose $\vecr_{\lambda_1}= \vecr_{\psi}$ to achieve perfect guessing probability, satisfying 
\begin{equation}
    \pg \big( \ket{\psi}, \, \m  \big) =   1 - \tr M_1 + \frac{1}{2} \Big( \tr \sqrt{M_1} \Big)^2=1 \,.
\end{equation}
We assume from now on that $z \neq 0$. From the second constraint of \eqref{eqnA: pguess_vec_unbiased app}, the vectors $a\vecr_{\lambda_1}$, $b\vecr_{\mu_1}$ and $z\vecr_{1}$ can be seen as the sides of a triangle (see Figure \ref{fig: SDP triangle}), whose height perpendicular to $z\vec{r}_{1}$ is given by $h$.

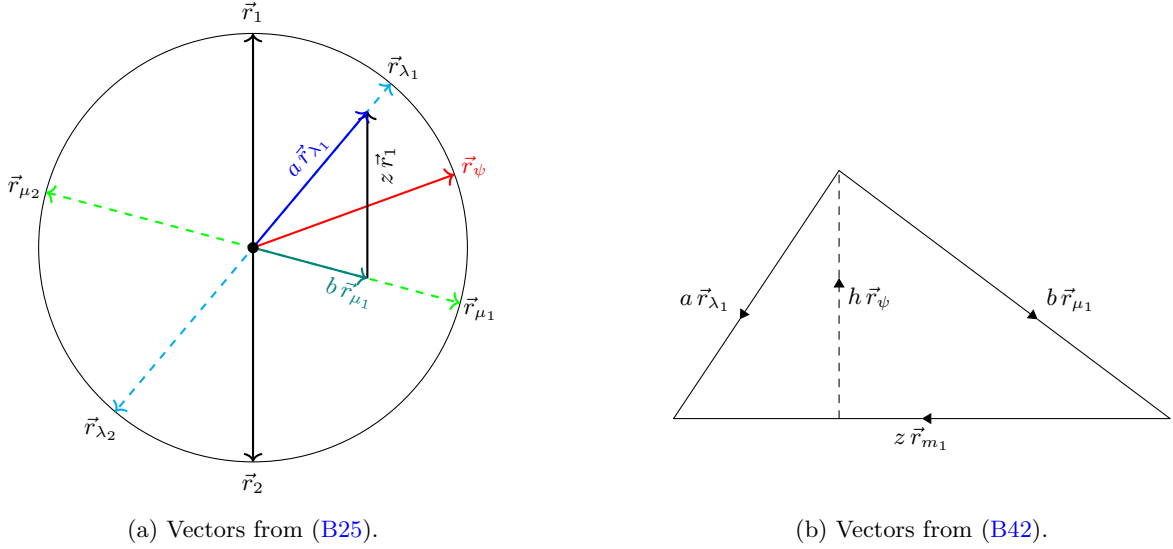
\begin{figure*}

\begin{subfigure}{0.49\textwidth}
\centering
    \begin{tikzpicture}
    \begin{scope}[scale=0.95]
        \draw (0,0) circle (3cm);

        \draw[thick, <->] (0,-3) -- (0,3);
        
        \draw[cyan, thick, dashed, <->] (50:3) -- (230:3);
        
        \draw[green, thick, dashed, <->] (165:3) -- (345:3);
        
        \draw[red, thick, ->] (0,0) -- (20:3);
        
        \draw[thick, ->] (1.6,-0.43) -- (1.6,1.91);
        
        \draw[teal, thick, ->] (0,0) -- (1.6,-0.43);
        
        \draw[blue, thick, ->] (0,0) -- (1.6,1.91);
        
        \node at (0,3.3) {$\vecr_{1}$};
        \node at (0,-3.3) {$\vecr_{2}$};
        \node at (50:3.3) {$\vecr_{\lambda_1}$};
        \node at (230:3.3) {$\vecr_{\lambda_2}$};
        \node at (345:3.3) {$\vecr_{\mu_1}$};
        \node at (165:3.3) {$\vecr_{\mu_2}$};    
        \node[red] at (20:3.3) {$\vecr_\psi$};
    
        \node[rotate=50, blue] at (60:1.5) {$a \,\vecr_{\lambda_1}$}; 
        \node[rotate=345, teal] at (335:1.5) {$b \,\vecr_{\mu_1}$};
        \node[rotate=90] at (1.85,0.75+ 0.4) {$z \, \vecr_1 $};

        \filldraw[black] (0,0) circle (2pt);
        \end{scope}
    \end{tikzpicture}
\caption{Vectors from \eqref{eqnA: vec_pguess_original}.}
\label{subfig: circle pic}
\end{subfigure}
\begin{subfigure}{0.49\textwidth}
\centering
    \begin{tikzpicture}[line cap=round, line join=round, >=Triangle]
    \begin{scope}[scale=1.1]
\clip(-3,-1) rectangle (3,3);
\coordinate (arrow1start) at (-1, 0);
    \coordinate (arrow2end) at (-1, 3);
    \coordinate (arrow2endless) at (-1, 1.7);
     \coordinate (arrow3end) at (3, 0);
     \coordinate (arrow3endless) at (0, 0);
      \coordinate (arrow4end) at (-3, 0);
      \coordinate (arrow4endless) at (-2.85, 0);
      \coordinate (leftless) at (-2.2, 1.2);
      \coordinate (rightless) at (1.4, 1.2);
\draw [->] (arrow2end) -- (leftless);
\draw [->] (arrow2end) -- (rightless);
\draw [->, dashed] (arrow1start) -- (arrow2endless);
\draw [-, dashed] (arrow2endless) -- (arrow2end);
\draw [->] (arrow3end) -- (arrow3endless);
\draw [-] (arrow3endless) -- (arrow4end);
\draw [-] (leftless) -- (arrow4end);
\draw [-] (rightless) -- (arrow3end);
\draw (arrow2endless) node[anchor=north west] {{$h \, \vecr_{\psi}$}};  
\draw (rightless) node[anchor=south west] {{$b\, \vecr_{\mu_1}$}};      
\draw (leftless) node[anchor= south east] {{$a\, \vecr_{\lambda_1}$}};
\draw (arrow3endless) node[anchor=north] {{$z\, \vecr_{m_1}$}};
\end{scope}
\end{tikzpicture}
\caption{Vectors from \eqref{eqnA: pguess_vec_unbiased app}.}
\label{subfig: triangle pic}
\end{subfigure}

\caption{Vector visualizations of the SDP constraints for a qubit two-outcome measurement.}
\label{fig: SDP triangle}
\end{figure*}
Heron's formula for the area $A$ of a triangle with side lengths $a$, $b$ and $z$, and semiperimeter $s=\frac{1}{2}\left(a + b + z \right)$, is 
\begin{equation}
    A = \sqrt{s(s-a)(s-b)(s-c)}\,, 
\end{equation}
so we can express $h$ as 
\begin{align}
    A &=  \frac{1}{4}\sqrt{(a+b+z)(a+b-z)(a-b+z)(-a+b+z)} 
    \\
    &=  \frac{1}{4}\sqrt{ \Big( (a+b)^2-z^2 \Big) \Big(z^2-(a-b)^2}\Big) = \frac{hz}{2}\,.
\end{align}
From the cosine rule for triangles, we have 
\begin{equation}
    - 2ab \leq a^2 + b^2 -z^2 \leq 2ab\,,
\end{equation}
from which we can extract the conditions
\begin{equation}
(a-b)^2 \leq z^2 \leq (a+b)^2\,.
\end{equation}
We then find the bound
\begin{equation}
    h = \frac{1}{2z}\sqrt{ \big( (a+b)^2-z^2 \big) \big(z^2-(a-b)^2}\big) \leq \frac{1}{2} \sqrt{(a+b)^2 - z^2}\,,
\end{equation}
with equality achieved only when $a=b$. The guessing probability can then be bounded from above by  
\begin{equation}\label{eqnA: pguess_vec_unbiased2 app}
\begin{aligned}
\pgs \big(  \m  \big) = \; & \min_{\vecr_{\phi}} \; \max_{a, \, b, \, \vecr_{\lambda_1}, \, \vecr_{\mu_1}  } 
& & 1 - \frac{a+b}{2} + \frac{1}{2} \sqrt{\left(a+b \right)^2-z^2}
\\
& \quad \text{subject to}
& & a \geq 0\,,  \quad b \geq 0\,, \quad \tr M_1  \geq a+b  
\\
& 
& & a\vecr_{\lambda_1} - b\vecr_{\mu_1} = z\vecr_{1}\,. 
\end{aligned}
\end{equation}
The objective function $f(a, b)$ of \eqref{eqnA: pguess_vec_unbiased2 app},
\begin{equation}
    f(a, b)= 1 - \frac{a+b}{2} + \frac{1}{2} \sqrt{\left(a+b \right)^2-z^2}\,,
\end{equation}
increases monotonically with $\left( a+b \right)$, as
\begin{equation}
    \diff{}{(a+b)} f (a, b) =  \frac{1}{2 \sqrt{ \left( a+b\right)^2 -z^2 } } \bigg( a+b - \sqrt{ \left( a+b\right)^2 -z^2 } \bigg) \geq 0\,,
\end{equation}
so we obtain the following bound by setting $a+b = \tr M_1$ in \eqref{eqnA: pguess_vec_unbiased2 app},
\begin{equation}\label{eqn: pguess_vec_unbiased3 app}
 \pgs \big( \ket{\psi}, \,  \m  \big) \leq 1 - \frac{\tr M_1}{2} + \frac{1}{2} \sqrt{\left(\tr M_1 \right)^2-z^2} = 1 - \tr M_1 + \frac{1}{2} \Big( \tr \sqrt{M_1} \Big)^2 \,,
\end{equation}
which proves \eqref{eqn: lemma2 app}.
\end{proof}

\subsection{Proof of Corollary \ref{corr: qudit 2-outcome app}}\label{app: corr1 proof}
\begin{corollary}\label{corr: qudit 2-outcome app}
Let $\m=\{M_x\}_x$ be any POVM with two outcomes, with $\lmin \left( M_1\right) + \lmax \left( M_1 \right) \leq \lmin \left( M_2\right) + \lmax \left( M_2 \right)$, where $\lmax \left( . \right)$ and $\lmin \left( . \right)$ denote the largest and smallest eigenvalues of the argument. Then 
\begin{equation}\label{eqnA: pguess_bound_2outcome}
 \pgs \big(  \m  \big) \leq 1 - \frac{1}{2} \bigg( \sqrt{\lmax \left( M_1\right)} - \sqrt{\lmin \left( M_1\right)} \bigg)^2\,.
\end{equation}
\end{corollary}
\begin{proof}
Consider a projector $\Pi$ and some state $\ket{\phi}$ that it stabilizes, i.e. $\Pi \ket{\phi}= \ket{\phi}$. The optimal guessing probability is bounded by 
\begin{equation}
 \pgs \big(  \m  \big) \leq \pg \big( \ket{\phi}, \, \m    \big) =  \pg \big( \Pi\ket{\phi},\,  \m   \big)
 \,.
\end{equation}
Note that any variables $\mathcal{K}_j=\{K_{x, j}\}_x$ that satisfy the constraints of the usual SDP \eqref{eqnA: pguess_main} for $\m$ will also satisfy 
\begin{equation}
    \begin{aligned}
     & \Pi K_{x, j} \Pi \, \geq \, 0 \;  \text{for all } \; x, j 
     \\
&\sum_{x}  d\, \Pi K_{x, j} \Pi =  \Pi \sum_{x}  \tr K_{x, j} \; \text{for all} \;j 
\\
&\sum_{j} \Pi K_{x,j} \Pi = \Pi M_{x} \Pi \; \text{ for all} \;  x\,,   
    \end{aligned}
\end{equation}
so we find 
\begin{equation}
   \pg \big( \Pi \ket{\phi},\,  \m   \big) \leq \pg \big( \Pi \ket{\phi},\,  \Pi\, \m \,\Pi   \big)
 \,,
\end{equation}
where 
\begin{equation}
    P \m P = \{\Pi\, M_x \,\Pi\}_x\,.
\end{equation}
For a fixed projector $\Pi$, we can choose $\ket{\phi}$ such that 
\begin{equation}
     \pg \big( \Pi \ket{\phi},\,  \Pi\, \m \,\Pi   \big) =  \pgs \big( \Pi\, \m \,\Pi   \big)
 \,,
\end{equation}
so for any $\Pi$, we find 
\begin{equation}
 \pgs \big(  \m  \big) \leq \pgs \big(  \Pi\,  \m \,\Pi  \big)
 \,,
\end{equation}
Setting
\begin{equation}
    P = \ketbra{\lmax}{\lmax} + \ketbra{\lmin}{\lmin}\,,
\end{equation}
where $\ket{\lmax}$ and $\ket{\lmin}$ are the eigenvectors of $M_1$ corresponding to the eigenvalues $\lmax \left( M_1 \right)$ and $\lmin \left( M_1 \right)$, and using Lemma \ref{lemma: qubit 2-outcome upper app}, we recover \eqref{eqnA: pguess_bound_2outcome}.
\end{proof}
Numerical evidence suggests that the upper bound in Corollary \ref{corr: qudit 2-outcome app} is saturated in general, but we have not proven this analytically. Instead, here we derive a \emph{lower} bound, which is not tight in general, on the guessing probability of a two-outcome POVM $\m$ using any state $\ket{\phi}$. We assume the eigenvalues $\{\lambda_i\}$ of $M_1$ are arranged in non-decreasing order, such that $\lmax \left( M_1\right) = \lambda_1$ and $\lmin \left( M_1 \right)= \lambda_d$. Defining $p= \frac{\lambda_1 + \lambda_d}{2}$, we consider the decomposition $\{\mathcal{K}_j\}$ with $\mathcal{K}_j= \{K_{1, j}, \, K_{2, j}\}$ such that 
\begin{equation}
    \begin{aligned}
        K_{1, 1} &= \diag\left( \min\{\lambda_i, \, p   \}_i \right)\,,
        \\
        K_{2, 1} &= \diag\left( \max\{p - \lambda_i, \, 0 \}_i \right)\,,
        \\
        K_{1, 2} &= \diag\left( \max\{ 0, \, \lambda_i -p   \}_i \right)\,,
        \\
        K_{2, 2} &= \diag\left( \min\{1-p, \,  1- \lambda_i   \}_i \right)\,,
    \end{aligned}
\end{equation}
where $\diag\left( . \right)$ is a diagonal matrix in the computational basis with the arguments as eigenvalues. This decomposition is well-defined, as all of its elements are positive semidefinite by construction, and
\begin{equation}
    \begin{aligned}
        K_{1, 1} + K_{2, 1} &= p \id\,, \qquad && K_{1, 2} + K_{2, 2} = \left( 1-p \right)\id\,, 
        \\
        K_{1, 1} + K_{1, 2} &= M_1\,, \qquad && K_{2, 1} + K_{2, 2} = \id - M_1 = M_2\,.
    \end{aligned}
\end{equation}
Note that
\begin{equation}
    K_{1, 1} + K_{2, 2} = \diag \left( \{1 - \abs{p-\lambda_i}\}_i \right)\,.
\end{equation}
For a given state $\ket{\phi}= \sum_{x} \phi_x\ket{x}$, where $\sum_x \phi_x^2=1$, we have the guessing probability 
\begin{align}
    \pg \left( \ket{\phi}, \, \m, \, \{\mathcal{K}_j\}\right) &= \sum_{x} \langle \phi | K_{x, x} | \phi\rangle = \sum_x \phi_x^2 \left( 1 - \abs{p-\lambda_x} \right) = 1 - \frac{1}{2} \sum_x \phi_x^2 \,\abs{\lambda_1 + \lambda_d -2 \lambda_x }
    \\
    &= 1 - \frac{1}{2} \sum_x \phi_x^2 \,\abs{ \,   \abs{\lambda_d - \lambda_x} - \abs{\lambda_1 - \lambda_x} \, }  \geq 1 - \frac{1}{2} \sum_x \phi_x^2 \,\abs{ \lambda_d - \lambda_1 } = 1 - \frac{\lmax \left( M_1\right) - \lmin \left( M_1\right)}{2}\,. 
\end{align}
This inequality is saturated only in the extreme cases where $\lmin \left( M_1\right)=0$ or $\lmax \left( M_1\right)=1$. In both of these cases, Eve can achieve the upper bound of Corollary \ref{corr: qudit 2-outcome app} by using the decomposition $\{\mathcal{K}_j\}$ for any state $\ket{\phi}$.

\section{Noisy projective measurements}\label{app: qudit noisy}

\subsection{Full proof of Lemma \ref{lemma: qudit lower app}}\label{app: lower qudit pguess}
\begin{lemma}\label{lemma: qudit lower app}
Let $\m_d$ be the noisy projective measurement in dimension $d$. The following lower bound holds,
 \begin{equation}\label{eqnA: lemma3}
  \pgs \big(  \m_d  \big) \geq  \frac{1}{d} \Big( \tr \sqrt{M_1} \Big)^2 \,.
\end{equation}   
\end{lemma}
\begin{proof}
We start by expressing Alice's chosen state $\ket{\phi}$ in the eigenbasis $\{\ket{x}\}$ of the POVM as 
\begin{equation}
    \ket{\phi}= \sum_{x} \langle{x | \phi} \rangle \ket{x}\,,
\end{equation}
where we assume for convenience that each term $\langle x | \phi \rangle$ is real (if they were complex, a unitary of the form $U = \sum_{k} e^{i \theta_{k}} \ketbra{k}{k}$ could be applied to make them real, leaving the POVM $\m_d$ unchanged). We also define the terms 
\begin{equation}
\id_{\neq x} = \sum_{k \neq x} \ketbra{k}{k}\,, \qquad c_{\neq x} =\langle{\phi | \id_{\neq x} | \phi} \rangle\,, \qquad \ket{\phi_{\neq x}} = \frac{\id_{\neq x} \ket{\phi}}{\sqrt{c_{\neq x}}}\,.
\end{equation}
The POVM $\m_d$ can be realized using the decomposition $\{\mathcal{K}_j\}$, $\mathcal{K}_j=\{K_{x, j}\}_x$, with
\begin{equation}\label{eqnA: qudit decomp}
    K_{x, j} = \begin{cases}
     \ketbra{\phi_x}{\phi_x}  + \frac{c_{\neq x} \, \ve}{d} \Big( \id_{\neq x} - \ketbra{\phi_{\neq x}}{\phi_{\neq x}} \Big)\,, & x = j\,,
     \\
     \ketbra{\phi_{xj}}{\phi_{xj}} \,, & x \neq j\,,
    \\
    \end{cases}
\end{equation}
where $\ket{\phi_x}$ and $\ket{\phi_{xj}}$ are unnormalized pure states given by 
\begin{equation}
\ket{\phi_x}= \sqrt{M_x}\ket{\phi} \,, \qquad\ket{\phi_{xj}}=  \sqrt{M_x} \Big( \langle{ \phi| j} \rangle\ket{x} - \langle{\phi| x} \rangle \ket{j} \Big)\,.   
\end{equation}
Each element $K_{x, j}$ is positive-semidefinite, so to prove that this decomposition is valid, we need only show that the final two conditions of \eqref{eqnA: pguess_main} hold, 
\begin{equation}
\sum_{x}  d K_{x, j} =  \id \sum_{x}  \tr K_{x, j}\text{ for all } j\,, \qquad \sum_{j} K_{x,j} = M_{x} \text{ for all } x\,. 
\end{equation}
Defining 
\begin{equation}
    A = d - \ve \left( d-1 \right)\,,
\end{equation}
we can expand $K_{x, x}$ in the basis $\{\ket{x}\}$ as 
\begin{equation}
    K_{x,x} = \,\frac{1}{d} \bigg( A  \langle{\phi | x} \rangle^2  \ketbra{x}{x}  + c_{\neq  x}\, \ve \, \id_{\neq  x} + \langle{\phi | x} \rangle \sqrt{A\ve \, c_{\neq x}}\,  \Big( \ketbra{x}{\phi_{\neq x}} + \ketbra{\phi_{\neq  x}}{x}\Big)     \bigg)\,,
\end{equation}
while for $x \neq j$, we can expand $K_{x, j}$ as 
\begin{equation}\label{eqn: Nxj}
    K_{x,j} = \, \frac{1}{d} \bigg( A  \langle{ \phi | j} \rangle^2     \ketbra{x}{x} + \langle{x | \phi} \rangle^2\, \ve \,\ketbra{j}{j} - \langle{ x| \phi} \rangle \langle{j | \phi} \rangle \sqrt{A \ve \, } \, \Big( \ketbra{j}{x} + \ketbra{x}{j} \Big)  \bigg)\,.
\end{equation}
The summation of \eqref{eqn: Nxj} over $j \neq x$ gives us
\begin{equation}
    \sum_{j \neq x} K_{x, j} = \,  \frac{1}{d} \bigg(       A\,  c_{\neq x} \ketbra{x}{x} + \langle{x | \phi} \rangle^2 \, \ve \, \id_{\neq x} - \langle{x| \phi} \rangle  \sqrt{A \ve \, c_{\neq x}}  \Big( \ketbra{\phi_{\neq x}}{x} + \ketbra{x}{\phi_{\neq x}} \Big)  \bigg)\,,
\end{equation}
while the summation over $x \neq j$ gives us
\begin{equation}
 \sum_{x \neq j}   K_{x,j} = \, \frac{1}{d} \bigg(   \langle{\phi | j} \rangle^2    A \id_{\neq j} + c_{\neq j}\, \ve \, \ketbra{j}{j} - \langle{j| \phi} \rangle \sqrt{A \ve \, c_{\neq j}} \, \Big( \ketbra{\phi_{\neq j}}{j} + \ketbra{j}{\phi_{\neq j}} \Big)  \bigg)\,.
\end{equation}
Noting that $\langle{x|\phi} \rangle^2 + c_{\neq x} =1$, the total sum of the elements over $j$ is
\begin{equation}
    \sum_{j} K_{x, j} = \frac{1}{d} \bigg( A \ketbra{x}{x} + \ve \id_{\neq x} \bigg) = M_x\,,
\end{equation}
while the total sum of the elements over $x$ is 
\begin{equation}
    \sum_{x} K_{x, j} = \frac{1}{d} \Big( A \langle{\phi| j} \rangle^2 + c_{\neq j} \, \ve \Big) \Big( \ketbra{j}{j} + \id_{\neq j} \Big) \propto \id\,.
\end{equation}
The conditions \eqref{eqn: Nxj} are satisfied, so the decomposition $\{\mathcal{K}_j\}$ from \eqref{eqnA: qudit decomp} is valid for any state $\ket{\phi}$.

We can bound the guessing probability for any state $\ket{\phi}$ using this decomposition,
\begin{equation}
    \pg \left(\ket{\phi}, \, \m \right) \geq  \sum_{x} \langle{\phi | K_{x, x} | \phi  } \rangle= \sum_{x} \langle{\phi | \sqrt{M_{x}} | \phi } \rangle^2\,.
\end{equation}
Note that 
\begin{equation}
    \langle{\phi | \sqrt{M_x} | \phi } \rangle = \frac{1}{\sqrt{d}} \bigg( \sqrt{A} \langle{\phi | x} \rangle^2 + \sqrt{\ve}\, c_{\neq x} \bigg)\,,
\end{equation}
such that 
\begin{equation}
    \sum_{x} \langle{\phi | \sqrt{M_x} | \phi } \rangle = \frac{1}{\sqrt{d}} \bigg( \sqrt{A} + (d-1)\sqrt{\ve} \bigg) = \tr \sqrt{M_x}\,.
\end{equation}
To quantify the deviation of the state $\ket{\phi}$ from the state $\ket{\psi}= \frac{1}{\sqrt{d}} \left(1, ..., 1 \right)^{T}$ which is unbiased to the basis $\{\ket{x}\}$, we introduce the real parameters
\begin{equation}
  \delta_{x} =  \langle{\phi | \sqrt{M_{x}} | \phi } \rangle - \frac{\tr \sqrt{M_x}}{d}\,, \end{equation}
where we must have 
\begin{equation}
    \sum_{x}{\delta_x} = 0\,.
\end{equation}
We then find
\begin{equation}
 \sum_{x}   \langle{ \phi | \sqrt{M_x} | \phi } \rangle^2 = \sum_{x} \left( \frac{\tr \sqrt{M_x}}{d} \right)^2 + \sum_{x} \delta_x^2 = \frac{1}{d} \left( \tr \sqrt{M_1} \right)^2 + \sum_{x} \delta_{x}^{2}\,, 
\end{equation}
where in the second equality, we use the fact that the $\tr \sqrt{M_x}$ terms are equal. Since $\sum_{x} \delta_{x}^{2}$ is non-negative, Eve will always achieve at least 
\begin{equation}
    \pgs \left( \m \right) \geq \frac{1}{d} \left( \tr \sqrt{M_1} \right)^2\,,
\end{equation}
and she can achieve strictly more when there are some nonzero parameters $\delta_x$. Since 
\begin{equation}
    \delta_x = \frac{1}{\sqrt{d}} \Big(\langle{\phi | x} \rangle^2 - \frac{1}{d} \Big)\, \bigg( \sqrt{A} - \sqrt{\ve} \bigg)\,, 
\end{equation}
when $ \langle{\phi | x} \rangle^2 \neq \frac{1}{d}$ for some $x$, i.e. when $\ket{\phi}$ is not unbiased to the measurement basis $\{\ket{x}\}$, Eve's guessing probability is bounded strictly from below by
\begin{equation}
    \pg \left( \ket{\phi}, \m \right) > \frac{1}{d} \left( \tr \sqrt{M_1} \right)^2\,.
\end{equation}

\end{proof}

\subsection{Proof of Lemma \ref{lemma: qudit upper app} based on dual variables}\label{app: upper qudit pguess}
\begin{lemma}\label{lemma: qudit upper app}
Let $\m_d$ be the noisy projective measurement in dimension $d$, and let $\ket{\psi}= \frac{1}{\sqrt{d}}\left(1, ..., 1 \right)^{T} $. Then
\begin{equation}\label{eqn: lemma4 app}
  \pg \big( \ket{\psi}, \, \m  \big) \leq  \frac{1}{d} \Big( \tr \sqrt{M_1} \Big)^2 \,.
\end{equation}  
\end{lemma}
\begin{proof}
We prove Lemma \ref{lemma: qudit upper app} by finding feasible dual variables $\{Y_x\}$ and $\{G_j\}$ that satisfy the constraints of the dual problem \eqref{eqnA: dualG} and achieve the bound
\begin{equation}\label{eqnA: dual var bound}
  \pg \big( \ket{\psi},\, \m  \big) \leq  \sum_{x}  \tr \big( Y_x \, M_x \big) = \frac{1}{d} \Big( \tr \sqrt{M_1} \Big)^2\,,
\end{equation}
where $\ket{\psi}= \frac{1}{\sqrt{d}} \left(1, ..., 1 \right)^{T}$.
First, it is helpful to set notation for Eve's POVM decomposition \eqref{eqnA: qudit decomp} for the state $\ket{\psi}$. We define
\begin{equation}\label{eqnA: psi decomp}
    K_{x, j} = \begin{cases}
     \ketbra{\psi_x}{\psi_x}  + \frac{d-1 }{d^2} \ve \,\Big( \id_{\neq x} - \ketbra{\psi_{\neq x}}{\psi_{\neq x}} \Big)\,,  & x = j\,,
     \\
     \frac{A+\ve}{d^2}\ketbra{\psi_{xj}}{\psi_{xj}} \,,   & x \neq j\,,
    \\
    \end{cases}
\end{equation}
where
\begin{equation}\label{eqn: A_meaning}
    A = d - \ve \left( d-1 \right) \,,  \qquad \id_{\neq x} := \sum_{k \neq x} \ketbra{k}{k}\,
\end{equation}
as before, and $\ket{\psi_{x}}$, $\ket{\psi_{xj}}$ and $\ket{\psi_{\neq x}}$ are (not necessarily normalized) vectors given by
\begin{equation}
\ket{\psi_x}= \sqrt{M_x}\ket{\psi}\,, \qquad \ket{\psi_{xj}} = \sqrt{\frac{d}{A+\ve}} \sqrt{M_x} \big( \ket{x} - \ket{j} \big)\,,   \qquad \ket{\psi_{\neq x}} =  \sqrt{\frac{d}{d-1}} \id_{\neq x} \ket{\psi}\,.
\end{equation}
Introducing the constants 
\begin{equation}
    \begin{aligned}
     \alpha &= \frac{d-1}{d} - \frac{d-1}{d^2} \tr \sqrt{M_1} \sqrt{\frac{d}{\ve}}\,, 
\\
\beta &= \frac{d-1}{d} - \frac{d-2}{d^2} \tr \sqrt{M_1} \sqrt{\frac{d}{\ve}}\,, 
\\
   \gamma &= \frac{\tr \sqrt{M_1}}{d \sqrt{d} } \sqrt{\frac{d-1}{\ve}} \frac{\sqrt{A }- \sqrt{\ve} }{\sqrt{A } + \sqrt{\ve} }\,,      
    \end{aligned}
\end{equation}
we can finally define our dual variables, 
\begin{equation}\label{eqnA: dual vars}
    \begin{aligned}
     Y_x &= \frac{1}{d^2} \tr \sqrt{M_1}\, {M_x}^{-\frac{1}{2}} + T_x\,, 
    \\
   T_x &= - \gamma \big( \ketbra{x}{\psi_{\neq x}} + \ketbra{\psi_{\neq x}}{x}  \big)\,, 
   \\ 
G_j &= T_j - \frac{d-1}{d} \ketbra{\psi}{\psi} - \alpha \big( \id_{\neq j} - \ketbra{\psi_{\neq j}}{\psi_{\neq j}} \big)   + \beta \Big( \id - d \sqrt{M_x} \ketbra{\psi}{\psi} \sqrt{M_x} \Big)\,.   
    \end{aligned}
\end{equation}
The variables $Y_x$, $T_x$ and $G_j$ are all Hermitian, and it is straightforward to see that $T_x$ is traceless, so it follows that $G_j$ must also be traceless. We also have
\begin{equation}
    \tr \left( M_{x} T_x \right) = \frac{1}{d} \Big( A \langle{x | T_x | x} \rangle + \ve \sum_{k \neq x} \langle{k | T_x | k} \rangle  \Big) = 0 \quad \textnormal{for all} \; \,x\,,
\end{equation}
such that 
\begin{equation}\label{eqnA: Yx correct}
    \sum_{x} \tr \left( Y_x M_x \right)  = \frac{1}{d} \left( \tr \sqrt{M_1} \right)^2\,.
\end{equation}
There just remains the more formidable task of showing that
\begin{equation}\label{eqnA: dual cond}
    Y_x - \delta_{xj}\ketbra{\psi}{\psi} - G_j \geq 0 
\end{equation}
for all $x$ and $j$. To do this, we first bound the rank of the left-hand side of \eqref{eqnA: dual cond} by proving that the complementary slackness condition
\begin{equation}
    K_{x, j} \big( Y_x - \delta_{xj}\ketbra{\psi}{\psi} - G_j \big)=0 \quad \textnormal{for all} \; \; x,\, j\,
\end{equation}
holds for the decomposition $\{\mathcal{K}_j\}$ from \eqref{eqnA: psi decomp} and the dual variables from \eqref{eqnA: dual vars}. When $x=j$, we find
\begin{align}
  K_{x, x} \big(Y_x - \ketbra{\psi}{\psi} - G_x \big) \nonumber
   &= \sqrt{M_x} \ket{\psi} \bigg( \left( \frac{1}{d^2} \tr \sqrt{M_1} - \frac{1}{d} \langle{\psi | \sqrt{M_x} | \psi } \rangle \right) \bra{\psi} 
   + \alpha \Big( \bra{\psi} \sqrt{M_x}\, \id_{\neq x} - \bra{\psi} \sqrt{M_x} \ketbra{\psi_{\neq x}}{\psi_{\neq x}} \Big) \nonumber
  \\
  & \qquad \qquad \qquad \; \; - \beta \Big( 1 - d \langle{\psi | M_x | \psi} \rangle \Big) \bra{\psi} \sqrt{M_x} 
  \bigg) \nonumber
  \\
 &  \; \; \; \; + \ve \frac{d-1}{d} \bigg( \frac{1}{d^2} \tr \sqrt{M_1}\, \id_{\neq x} \,M_{x}^{-1/2}
  - \frac{1}{d^2} \tr \sqrt{M_1} \ketbra{\psi_{\neq x}}{\psi_{\neq x}} M_x^{-1/2}
  + \alpha \Big( \id_{\neq x} - \ketbra{\psi_{\neq x}}{\psi_{\neq x}} \Big) \nonumber
\\
& \qquad \qquad \qquad \; \;- \beta \Big( \id_{\neq x}
 - \ketbra{\psi_{\neq x}}{\psi_{\neq x}} - d \id_{\neq x} \sqrt{M_x} \ketbra{\psi}{\psi} \sqrt{M_x} \nonumber
 \\
& \qquad \qquad \qquad \qquad \qquad \qquad \qquad \qquad \;\; \,+ d \langle{\psi_{\neq x} | \sqrt{M_x} | \psi } \rangle \ketbra{\psi_{\neq x}}{\psi} \sqrt{M_x}\, \Big)
  \bigg) 
  \\
  &= \ve \, \frac{d-1}{d} \bigg(  \frac{1}{d^2} \tr \sqrt{M_1} \sqrt{\frac{d}{\ve}} + \alpha - \beta  \bigg) \Big(\id_{\neq x} - \ketbra{\psi_{\neq x}}{\psi_{\neq x}} \Big)  =0\,,
\end{align}
where we use the properties 
\begin{equation}
    \langle{\psi | \sqrt{M_x} | \psi} \rangle = \frac{1}{d} \tr \sqrt{M_1}\,, \qquad \id_{\neq x} \sqrt{M_1} = \sqrt{\frac{\ve}{d}} \id_{\neq x}\,, \qquad \langle{\psi | M_x | \psi} \rangle = \frac{1}{d}\,, \qquad \alpha - \beta = - \frac{1}{d^2} \tr \sqrt{M_1} \sqrt{\frac{d}{\ve}}\,.
    \end{equation}
Note that 
\begin{equation}
    \bra{\psi} \sqrt{M_x}\, \Big( \id_{\neq x} - \ketbra{\psi_{\neq x}}{\psi_{\neq x}} \Big) = \frac{\sqrt{\ve \left( d-1 \right) }}{d} \bra{\psi_{\neq x}} \Big( \id_{\neq x} - \ketbra{\psi_{\neq x}}{\psi_{\neq x}} \Big)=0\,,
\end{equation}
so the term $K_{x,x}$ has rank $d-1$. Since 
\begin{equation}
    K_{x, x} \big(Y_x - \ketbra{\psi}{\psi} - G_x \big) = \big(Y_x - \ketbra{\psi}{\psi} - G_x \big)K_{x, x} =0\,,
\end{equation}
we learn that $Y_x - \ketbra{\psi}{\psi} - G_x $ has rank one at most. To prove that it's positive semidefinite, we need only show that its trace is non-negative,
\begin{equation}
    \tr \big(Y_x - \ketbra{\psi}{\psi} - G_x \big) =  \frac{1}{d^2} \tr \sqrt{M_1}\, \tr\left({M_x}^{-1/2}\right) - 1= \frac{d-1}{d^2} \frac{\left(\sqrt{A}- \sqrt{\ve}\right)^2}{\sqrt{A \ve}} \geq 0\,. 
\end{equation}
We proceed similarly to show that $Y_x- G_j$ is positive semidefinite for $x \neq j$. We can expand it as 
\begin{equation}
Y_x- G_j=  \frac{1}{d^2} \tr \sqrt{M_1}\, {M_x}^{-\frac{1}{2}} + T_x  - T_j + \frac{d-1}{d} \ketbra{\psi}{\psi} + \alpha \big( \id_{\neq j} - \ketbra{\psi_{\neq j}}{\psi_{\neq j}} \big) - \beta \big( \id - d \sqrt{M_j}\ketbra{\psi}{\psi} \sqrt{M_j} \big)\,, 
\end{equation}
where
\begin{equation}
    T_x -T_j = - \gamma \sqrt{\frac{d-2}{d-1}} \Big( \big( \ket{x} - \ket{j}  \big) \bra{\psi_{\neq x, j}} + \ket{\psi_{\neq x, j}}\big( \bra{x} -\bra{j}  \big)  \Big)\,. 
\end{equation}
Defining the unnormalized state
\begin{equation}
    \ket{\tilde{\psi}_{xj}}:= \frac{d}{\sqrt{A+\ve}} \ket{\psi_{xj}}\,,
\end{equation}
we again prove the complementary slackness condition \eqref{eqnA: slackness}, 
\begin{equation}
    K_{x, j}  \big( Y_x- G_j \big) =  \ketbra{\tilde{\psi}_{xj}}{\tilde{\psi}_{xj}} \big( Y_x- G_j \big)=0\,.
\end{equation}
For convenience, however, we do this component by component, proving separately that 
\begin{equation}\label{eqn: vec cond dual}
\bra{\tilde{\psi}_{xj}} Y_x - G_j  \ket{x} =  \bra{\tilde{\psi}_{xj}} Y_x - G_j  \ket{j} = \bra{\tilde{\psi}_{xj}} Y_x - G_j \ket{k} =0\,,
\end{equation}
where $k \neq x, j$. First we have
\begin{align}\label{eqnA: x zero}
\langle{\tilde{\psi}_{xj}| Y_x - G_j | x } \rangle
 &= \bra{\tilde{\psi}_{xj}} \bigg( \frac{1}{d^2}\tr \sqrt{M_1} \sqrt{\frac{d}{A}} \ket{x} - \gamma \sqrt{\frac{d-2}{d-1}} \ket{\psi_{\neq x, j}} + \frac{d-1}{d \sqrt{d}} \ket{\psi} \nonumber
 \\
& \qquad \qquad \; \; \; + \alpha \Big( \ket{x} - \frac{1}{\sqrt{d-1}} \ket{\psi_{\neq j}}\Big) - \beta  \big( \ket{x} - \sqrt{\ve} \sqrt{M_j}\ket{\psi}  \big)   \bigg)   
  \\
  &=  \frac{1}{d^2}\tr \sqrt{M_1} + \frac{d-1}{d^2 \sqrt{d}} \left( \sqrt{A} - \sqrt{\ve} \right) + \alpha \sqrt{\frac{A}{d}} \frac{d-2}{d-1} - \beta \sqrt{\frac{A}{d}}       
  \\
  &=   \frac{1}{d^2}\tr \sqrt{M_1} - \frac{1}{d^2 \sqrt{d}} \Big( \sqrt{A} + \left( d-1 \right) \sqrt{\ve} \Big) =0\,.
  \end{align}
Similarly, we have 
\begin{align}\label{eqnA: j zero}
 \,\langle{\tilde{\psi}_{xj}| Y_x - G_j | j } \rangle 
 &= \bra{\tilde{\psi}_{xj}} \bigg( \frac{1}{d^2}\tr \sqrt{M_1} \sqrt{\frac{d}{\ve}} \ket{j} + \gamma \sqrt{\frac{d-2}{d-1}} \ket{\psi_{\neq x, j}} \nonumber
 \\
 & \qquad \qquad \;\;\; + \frac{d-1}{d \sqrt{d}} \ket{\psi} - \beta  \big( \ket{j} - \sqrt{A} \sqrt{M_j}\ket{\psi}  \big)   \bigg)   
  \\
  &=  -\frac{1}{d^2}\tr \sqrt{M_1} + \frac{d-1}{d^2 \sqrt{d}} \left( \sqrt{A} - \sqrt{\ve} \right)  + \beta \sqrt{\frac{\ve}{d}}       
  \\
  &=   -\frac{d-1}{d^2}\tr \sqrt{M_1} + \frac{d-1}{d^2 \sqrt{d}} \Big( \sqrt{A} + \left( d-1 \right) \sqrt{\ve}  \Big) =0\,,
  \end{align}
and finally,
\begin{align}\label{eqnA: k zero}
\langle{\tilde{\psi}_{xj}| Y_x - G_j | k } \rangle &= \bra{\tilde{\psi}_{xj}} \bigg( \frac{1}{d^2}\tr \sqrt{M_1} \sqrt{\frac{d}{\ve}} \ket{k} -  \frac{\gamma}{\sqrt{d-1}} \big( \ket{x} - \ket{j} \big) \nonumber
\\
& \qquad \qquad \;\;\; + \frac{d-1}{d \sqrt{d}} \ket{\psi} + \alpha \Big( \ket{k} - \frac{1}{\sqrt{d-1}} \ket{\psi_{\neq j}}\Big) - \beta  \big( \ket{k} - \sqrt{\ve} \sqrt{M_j}\ket{\psi}  \big)   \bigg)  
  \\
  &=  - \frac{\gamma}{\sqrt{d(d-1)}}\big( \sqrt{A} + \sqrt{\ve} \big) + \frac{d-1}{d^2 \sqrt{d}} \left( \sqrt{A} - \sqrt{\ve} \right) - \frac{\alpha}{d-1}  \sqrt{\frac{A}{d}}         
  \\
  &=   \frac{1}{d^2}\tr \sqrt{M_1} - \frac{1}{d^2 \sqrt{d}} \Big( \sqrt{A} + \left( d-1 \right) \sqrt{\ve} \Big) =0\,
  \end{align}
for any $k \neq x, j$. Since \eqref{eqn: vec cond dual} holds, we find that $\ket{\psi_{xj}}$ is a zero eigenstate of $Y_x-G_j$. Now consider the projector $\Pi_{\neq x, j}$ defined by
\begin{equation}
  \Pi_{\neq x, j}=\id_{\neq x, j} - \ketbra{\psi_{\neq x, j}}{\psi_{\neq x, j}}\,, \qquad  \ket{\psi_{\neq x, j}} = \frac{1}{\sqrt{d-2}} \sum_{k \neq x, j} \ket{k}\,,
\end{equation}
which has rank $d-3$ and is orthogonal to the state $\ket{\psi_{xj}}$. We have 
\begin{equation}
    \Pi_{\neq x, j} \big( Y_x - G_j \big) =
     \Pi_{\neq x, j} \Big( \frac{1}{d^2} \tr \sqrt{M_1} \sqrt{\frac{d}{\ve}} + \alpha - \beta \Big) =0\,,
\end{equation}
where we used
\begin{equation}
    \id_{\neq x, j} \ket{\psi} = \sqrt{\frac{d-2}{d}} \ket{\psi_{\neq x, j}}\,, \qquad \id_{\neq x, j} \ket{\psi_{\neq j}}= \sqrt{\frac{d-2}{d-1}} \ket{\psi_{\neq x, j}}\,.
\end{equation}
Since it is orthogonal to both the projectors $\ketbra{\psi_{xj}}{\psi_{xj}}$ and $\Pi_{\neq x, j}$, we find that $Y_x - G_j$ has rank two at most. As such, since it is Hermitian, 
$Y_x - G_j$ is positive-semidefinite if its determinant is non-negative. To find this determinant, we can express $Y_x - G_j$ in any two-vector orthonormal basis that lives in its support. Such a basis is given by the pair of vectors $\ket{\psi_{\neq x, j}}$ and $\ket{\psi_{xj}^{\perp}}$, where 
\begin{equation}
    \ket{\psi_{xj}^{\perp}}= 
     \frac{1}{\sqrt{A+ \ve}} \left( {\sqrt{\ve}} \ket{x} + \sqrt{A} \ket{j}  \right)\,.
\end{equation}
In this basis, then we can write 
\begin{equation}
    Y_x - G_j = 
\begin{pmatrix}
 \langle \psi_{\neq x, j} | Y_x - G_j | \psi_{\neq x, j} \rangle  &    \langle{\psi_{\neq x, j} | Y_x - G_j | \psi_{x, j}^\perp} \rangle
 \vspace{0.3 cm}
 \\
 \,\, \, \,\langle{\psi_{x, j} ^\perp | Y_x - G_j | \psi_{\neq x, j}} \rangle  & \; \; \,  \langle{ \psi_{x, j}^\perp | Y_x - G_j | \psi_{x, j}^\perp} \rangle
    \end{pmatrix}_{\ket{\psi_{\neq x, j}}, \, \ket{\psi^{\perp}_{x, j}} }\,.
\end{equation}
The matrix coefficients are
\begin{equation}
    \begin{aligned}
     \langle{\psi_{\neq x, j} | Y_x - G_j | \psi_{\neq x, j}} \rangle &= \frac{d-2}{d^3} \left( d(d-1+\ve) + \frac{\left( A + \ve \right) \left( \sqrt{A} - \sqrt{\ve} \right)  
 }{\sqrt{\ve}}   \right)   \,  
 \\
\langle{\psi_{\neq x, j} | Y_x - G_j | \psi_{x, j}^\perp} \rangle &=   \frac{ \sqrt{d}\,\sqrt{d-2}\, \sqrt{A+\ve}\, \left(1 + \ve \right)  }{d^2 \left( \sqrt{A} + \sqrt{\ve}  \right) \sqrt{\ve}}  \tr \sqrt{M_1}   \,
\\
\langle{\psi_{x, j}^\perp | Y_x - G_j | \psi_{ x, j}^\perp} \rangle &= \frac{A+ \ve}{d^2 \sqrt{A \ve} }  \left( \tr \sqrt{M_1} \right)^2     \,,   
    \end{aligned}
\end{equation}
so the determinant is then
\begin{align}
\textnormal{det} \left( Y_x -G_j \right) &= \langle{\psi_{\neq x, j} | Y_x - G_j | \psi_{\neq x, j}} \rangle \langle{\psi_{x, j}^\perp | Y_x - G_j | \psi_{ x, j}^\perp} \rangle - \langle{\psi_{\neq x, j} | Y_x - G_j | \psi_{x, j}^\perp} \rangle^2 
\\
&= \frac{ \left(d-2 \right) \left(A+\ve \right)\, \left( \sqrt{A} - \sqrt{\ve} \right)^2 \, \left( 1 + \sqrt{A\ve} \right)^2    }{d^2  \sqrt{A \ve}\, \left(\sqrt{A} + \sqrt{\ve }\right)^4}\,,
\end{align}
which is always non-negative. We have proven then that the variables $\{Y_x\}$ and $\{G_j\}$ given in \eqref{eqnA: dual vars} are valid, so from \eqref{eqnA: Yx correct}, the inequality \eqref{eqnA: dual var bound} holds, proving Lemma \ref{lemma: qudit upper app}.
\end{proof}

\subsection{Proof of Lemma \ref{lemma: qudit upper app} based on permutations}\label{app: upper qudit pguess perm}
\begin{proof}
 We prove that the decomposition $\{\mathcal{K}_j\}$ from \eqref{eqnA: qudit decomp} in Appendix \ref{app: upper qudit pguess} is optimal for the unbiased state $\ket{\psi}$. First, note that because all of the elements of $\ket{\psi}$ are real when expressed in the computational basis, we are free to assume that all the coefficients of $K_{x, x}$ are real in this basis too, as the imaginary parts do not contribute to the guessing probability.
From now on, we assume that $\{\mathcal{K}_j\}$ has only real elements. We denote by $\sigma \left( . \right)$ a function that permutes elements in $x \in 1, ..., d$, and we denote the corresponding permutation matrix in the basis $\{\ket{x}\}$ by $\Pi_{\sigma}$. Notice that the composition of permutations $\sigma_a \circ \sigma_b \left( .\right) $ corresponds to the matrix product $\Pi_{\sigma_a \circ \sigma_b}=\Pi_{\sigma_a}\Pi_{\sigma_b}$. Starting with $\{\mathcal{K}_j\}$, we can define a new decomposition into POVMs $\mathcal{K}^{\sigma}_{j}=\{K^{\sigma}_{x, j}\}_x$  by averaging over all $d!$ possible permutations of the elements $K_{x, j}$, 
\begin{equation}\label{eqnA: Kperm}
    K^{\sigma}_{x, j} = \frac{1}{d!} \sum_{\sigma} \Pi_{\sigma}^{T} \, K_{\sigma \left(x \right), \, \sigma \left(j \right)} \,\Pi_\sigma \,.
\end{equation}
This decomposition is valid, as it satisfies 
\begin{align}
    \sum_x K^{\sigma}_{x, j} &= \frac{1}{d!} \sum_{\sigma} \Pi^{T}_{\sigma} \left( \sum_x K_{\sigma \left( x \right), \, \sigma \left( j \right) } \right) \Pi_{\sigma} = \frac{1}{d!} \sum_{\sigma} \Pi^{T}_{\sigma} \left( \frac{1}{d} \sum_x \tr K_{\sigma \left(x \right), \, \sigma \left( j \right)} \id \right) \Pi_{\sigma} = \frac{1}{d} \tr K^{\sigma}_{\sigma x, \, j} \id\,,
    \\
   \sum_j K^{\sigma}_{x, j} &=  \frac{1}{d!} \sum_{\sigma} \Pi^{T}_{\sigma} \left( \sum_j K_{\sigma \left( x \right), \, \sigma \left( j \right) } \right) \Pi_{\sigma} = \frac{1}{d!} \sum_{\sigma} \Pi^{T}_{\sigma}  M_{\sigma \left( x \right)}  \Pi_{\sigma} = M_x\,,
\end{align}
where the final equality uses the fact that, since $\m_d$ is diagonal in the basis $\{\ket{x}\}$, we have $\Pi^{T}_{\sigma} M_{\sigma \left( x\right)} \Pi_{\sigma} = M_x$ for any $x$ and any permutation $\sigma$. Note that the trace of each $K_{x, j}$ is symmetric under the interchange of $x$ and $j$, so we have
\begin{equation}
    \sum_x \tr K^{\sigma}_{x, j} = \sum_j \tr K^{\sigma}_{x, j} = \tr M_x =1\,.
\end{equation}
The new decomposition $\{\mathcal{K}^{\sigma}_{j}\}$ gives the same guessing probability as $\{\mathcal{K}_j\}$, because the unbiased state $\ket{\psi}$ is invariant under all permutations $\sigma$, i.e.
\begin{equation}
    \pg \left( \ket{\psi}, \, \m_d\right) = \sum_x \langle \psi | K^{\sigma}_{x, x} | \psi \rangle = \frac{1}{d!}  \sum_{x, \sigma} \langle \psi | \Pi_{\sigma}^{T} \, K_{\sigma \left(x \right), \, \sigma \left(x \right)} \,\Pi_\sigma | \psi \rangle = \frac{1}{d!}  \sum_{x, \sigma} \langle \psi | \, K_{\sigma \left(x \right), \, \sigma \left(x \right)} \, | \psi \rangle = \sum_{x} \langle \psi | \, K_{x, \, x} \, | \psi \rangle\,,
\end{equation}
so $\{\mathcal{K}^{\sigma}_j\}$ is an optimal decomposition. From this point on, we drop the superscript $\sigma$ and assume that our decomposition $\{\mathcal{K}_j\}$ has the form \eqref{eqnA: Kperm}. We have the property that, for any permutation $\sigma_a$,
\begin{equation}
    K_{\sigma_a \left( x\right), \, \sigma_a \left( j \right)} = \frac{1}{d!} \sum_{\sigma} \Pi_{\sigma}^{T} \, K_{\sigma \circ \sigma_a\left(x \right), \, \sigma \circ \sigma_a \left(j \right)} \,\Pi_\sigma = \Pi_{\sigma_a} \left( \frac{1}{d!} \sum_{\sigma}  \Pi_{\sigma_a}^{T}\Pi_{\sigma}^{T} \, K_{\sigma \circ \sigma_a\left(x \right), \, \sigma \circ \sigma_a \left(j \right)} \,\Pi_\sigma \Pi_{\sigma_a} \right) \Pi_{\sigma_a}^{T}= \Pi_{\sigma_a} K_{x, j  } \Pi_{\sigma_a}^{T}\,,
\end{equation}
where in the second equality we use the decomposition of the identity $\id=\Pi_{\sigma_a} \Pi^{T}_{\sigma_a}$ and in the third the fact that $\Pi_{\sigma} \Pi_{\sigma_a}= \Pi_{\sigma \circ \sigma_a}$. In particular, if we denote by $\sigma_x$ any permutation that leaves the element $x$ unchanged, we find
\begin{equation}
  K_{\sigma_x \left( x\right), \, \sigma_x \left( x \right)}  =  K_{x, x} =  \Pi_{\sigma_x} K_{x, x} \Pi_{\sigma_x}^{T}\,.
\end{equation}
Fix $\sigma_x \left( k \right) =m$ and $\sigma_x \left( l \right)=n$, where $k \neq l$, $m \neq n$ and $k, l, m, n \neq x$. Then 
\begin{align}
    \left[ K_{x, x} \right]_{x, k} &= \left[ K_{x, x} \right]_{x, m} && \textnormal{for all} \;\; k, m \neq x\,, 
    \\
    \left[ K_{x, x} \right]_{k, k} &= \left[ K_{x, x} \right]_{m, m} && \textnormal{for all} \;\; k, m \neq x\,, 
    \\
    \left[ K_{x, x} \right]_{k, l} &= \left[ K_{x, x} \right]_{m, n} && \textnormal{for all} \;\; k \neq l, \;\;  m \neq n, \;\; k, l, m, n \neq x\,, 
\end{align}
where $\left[ K_{x, x} \right]_{k, l}$ are the elements of the matrix $K_{x, x}$ in the basis $\{\ket{x}\}$. Similarly, denoting by $\sigma_{xj}$ any permutation of elements that leave both $x$ and $j$ unchanged, when $x \neq j$ we have 
\begin{align}
    \left[ K_{x, j} \right]_{x, k} &= \left[ K_{x, j} \right]_{x, m} && \textnormal{for all} \;\; k, m \neq x, j\,, 
    \\
     \left[ K_{x, j} \right]_{l, j} &= \left[ K_{x, j} \right]_{n, j} && \textnormal{for all} \;\; k, m \neq x, j\,,
     \\
    \left[ K_{x, j} \right]_{k, k} &= \left[ K_{x, j} \right]_{m, m} && \textnormal{for all} \;\; k, m \neq x, j\,, 
    \\
    \left[ K_{x, j} \right]_{k, l} &= \left[ K_{x, j} \right]_{m, n} && \textnormal{for all} \;\; k \neq l, \;\;  m \neq n, \;\; k, l, m, n \neq x, j\,. 
\end{align}
Using the fact that all the decomposition elements are positive semidefinite, for $x \neq j$, we can write $K_{x, x}$ and $K_{x, j}$ in the basis $\{\ket{x}\}$ as
\begin{equation}\label{eqn: Kmats}
 K_{x, x} = \hspace{-2mm}
   \begin{array}{c}
    \scriptstyle |x \rangle \;\;
    \\[13ex]
    \end{array}
    \hspace{-3mm}
    \begin{pmatrix}
        \begin{array}{c|cccc}
            \tau_x    & \gamma_x & \gamma_x & \dots  & \gamma_x \\
            \hline
            \gamma_x  & \alpha_x & \beta_x  & \dots  & \beta_x  \\
            \gamma_x  & \beta_x  & \alpha_x & \dots  & \beta_x  \\
            \vdots    & \vdots   & \vdots   & \ddots & \vdots   \\
            \gamma_x  & \beta_x  & \beta_x  & \dots  & \alpha_x
        \end{array}
    \end{pmatrix}   \,, 
    \qquad 
    K_{x, j} = 
    \hspace{-2mm}
    \begin{array}{c}
      \scriptstyle |x\rangle \\
      \scriptstyle |j\rangle \\[13ex]
    \end{array}
    \hspace{-2mm}
    \begin{pmatrix}
        \begin{array}{cc|cccc}
            f_{xj} & g_{xj} & s_{xj} & s_{xj} & \dots & s_{xj} \\
            g_{xj} & h_{xj} & t_{xj} & t_{xj} & \dots & t_{xj} \\
            \hline
            s_{xj} & t_{xj} & a_{xj} & b_{xj} & \dots & b_{xj} \\
            s_{xj} & t_{xj} & b_{xj} & a_{xj} & \dots & b_{xj} \\
            \vdots & \vdots & \vdots & \vdots & \ddots & \vdots \\
            s_{xj} & t_{xj} & b_{xj} & b_{xj} & \dots & a_{xj}
        \end{array}
    \end{pmatrix}\,,
\end{equation}
where the entries $\tau_x, \, \gamma_x, \, \alpha_x, \, \beta_x, \, f_{xj}, \, g_{xj}, \, h_{xj}, \, s_{xj}, \, t_{xj}, \, a_{xj}$ and $b_{xj}$ are all real. 
Denote by $\sigma =l \leftrightarrow k$ the permutation whose only action is to swap elements $l$ and $k$. 
Applying the property $K_{\sigma \left( x\right), \, \sigma \left( j \right)}= \Pi_{\sigma} K_{x, j  } \Pi_{\sigma}^{T}$ again and choosing $\sigma=x \leftrightarrow k$ for $k \neq x$, we find
\begin{equation}
\begin{aligned}
\tau_k &= \left[ K_{kk}\right]_{kk} = \left[ K_{xx}\right]_{xx} = \tau_x\,,
\\
\gamma_k &= \left[ K_{kk}\right]_{xk} = \left[ K_{xx}\right]_{kx} = \gamma_x\,,
\\
\alpha_k &= \left[ K_{kk}\right]_{xx} = \left[ K_{xx}\right]_{kk} = \alpha_x\,,
\\
\beta_k &= \left[ K_{kk}\right]_{xl} = \left[ K_{xx}\right]_{kl} = \beta_x\,, \qquad \qquad l \neq x, k\,.
\end{aligned}
\end{equation}
Similarly, for $x \neq j$, choosing the permutations $x \leftrightarrow k$ and $j \leftrightarrow l$, where $k \neq l$ and $k, l \neq x, j$, we find
\begin{equation}
\begin{aligned}
    f_{kl} &= \left[ K_{kl} \right]_{kk} = \left[ K_{xj} \right]_{xx} = f_{xj}\,,
    \\
    h_{kl} &= \left[ K_{kl} \right]_{ll} = \left[ K_{xj} \right]_{jj} = h_{xj}\,,
    \\
    g_{kl} &= \left[ K_{kl} \right]_{kl} = \left[ K_{xj} \right]_{xj} = g_{xj}\,,
    \\
    s_{kl} &= \left[ K_{kl} \right]_{km} = \left[ K_{xj} \right]_{xm} = s_{xj}\,, \qquad && m \neq k, l, x, j\,,
    \\
    t_{kl} &= \left[ K_{kl} \right]_{lm} = \left[ K_{xj} \right]_{jx} = t_{xj}\,, \qquad && m \neq k, l, x, j\,,
    \\
    a_{kl} &= \left[ K_{kl} \right]_{mm} = \left[ K_{xj} \right]_{mm} = a_{xj}\,, \qquad && m \neq k, l, x, j\,,
    \\
    b_{kl} &= \left[ K_{kl} \right]_{mn} = \left[ K_{xj} \right]_{mn} = b_{xj}\,, \qquad && m \neq n, \, \;\; m, n \neq k, l, x, j\,.
\end{aligned}
\end{equation}
Since all of the entries in \eqref{eqn: Kmats} are independent of the subscripts $x$ and $j$, we can drop them, obtaining a much simpler decomposition
\begin{equation}\label{eqnA: Kmats2}
 K_{x, x} = \hspace{-2mm}
   \begin{array}{c}
    \scriptstyle |x \rangle \;\;
    \\[13ex]
    \end{array}
    \hspace{-3mm}
    \begin{pmatrix}
        \begin{array}{c|cccc}
            \tau    & \gamma & \gamma & \dots  & \gamma \\
            \hline
            \gamma  & \alpha & \beta  & \dots  & \beta  \\
            \gamma  & \beta  & \alpha & \dots  & \beta  \\
            \vdots    & \vdots   & \vdots   & \ddots & \vdots   \\
            \gamma  & \beta  & \beta  & \dots  & \alpha
        \end{array}
    \end{pmatrix}   \,, 
    \qquad 
    K_{x, j} = 
    \hspace{-2mm}
    \begin{array}{c}
      \scriptstyle |x\rangle \\
      \scriptstyle |j\rangle \\[13ex]
    \end{array}
    \hspace{-2mm}
    \begin{pmatrix}
        \begin{array}{cc|cccc}
            f & g & s & s & \dots & s \\
            g & h & t & t & \dots & t \\
            \hline
            s & t & a & b & \dots & b \\
            s & t & b & a & \dots & b \\
            \vdots & \vdots & \vdots & \vdots & \ddots & \vdots \\
            s & t & b & b & \dots & a
        \end{array}
    \end{pmatrix}\,.
\end{equation}
The coefficients in \eqref{eqnA: Kmats2} bear no relation to constants used in other sections of this work. Using the notation
\begin{equation}
    \ket{\psi_{\neq x}} = \frac{1}{\sqrt{d-1}} \sum_{k \neq x} \ket{k}\,, \quad \ket{\psi_{\neq x, j}} = \frac{1}{\sqrt{d-2}} \sum_{k \neq x, j} \ket{k}\,, 
\end{equation}
and 
\begin{equation}
\id_{\neq x} = \id - \ketbra{x}{x}\,, \quad \; \id_{\neq x, j} = \id - \ketbra{x}{x} - \ketbra{j}{j}\,, \quad \; \Pi_{\neq x} = \id_{\neq x} - \ketbra{\psi_{\neq x}}{\psi_{\neq x}}\,, \quad \;\Pi_{\neq x, j} = \id_{\neq x, j} - \ketbra{\psi_{\neq x, j}}{\psi_{\neq x, j}}\,,   
\end{equation}
we can write, for $j \neq x$,
\begin{equation}
\begin{aligned}
K_{x, x} &= \tau \ketbra{x}{x} + \gamma \sqrt{d-1} \big( \ketbra{\psi_{\neq x}}{x} + \ketbra{x}{\psi_{\neq x}} \big) + \big( \alpha + \left( d-1 \right) \beta \big) \ketbra{\psi_{\neq x}}{\psi_{\neq x} }+ \left( \alpha - \beta \right) \,\Pi_{\neq x}\,, \vspace{1.5 cm}
\\
K_{x, j} &= f \ketbra{x}{x} + h \ketbra{j}{j} + g \big( \ketbra{x}{j} + \ketbra{j}{x} \big) + s \sqrt{d-1} \big( \ketbra{\psi_{\neq x, j}}{x} + \ketbra{x}{\psi_{\neq x, j}} \big) + t \sqrt{d-2} \big( \ketbra{\psi_{\neq x, j}}{j} + \ketbra{j}{\psi_{\neq x, j}}\big) 
\\
& \quad + \big( a + \left(d-3 \right)b \big) \ketbra{\psi_{\neq x, j}}{\psi_{\neq x, j}} + \left(a-b \right)  \,\Pi_{\neq x, j}\,.
\end{aligned}
\end{equation}
To solve for the constants, we return to the SDP constraints
\begin{equation}
\begin{aligned}
 \sum_{j} K_{x, j} &= M_x \qquad && \textnormal{for all} \;\; x\,,
 \\
 \sum_x K_{x, j}&= \frac{1}{d} \tr \sum_x K_{x, j} \id = \frac{1}{d} \id \qquad && \textnormal{for all} \;\; j\,,
\end{aligned}
\end{equation}
where
\begin{equation}
    M_x = \frac{1}{d} \big(  A \ketbra{x}{x} + \ve \id_{\neq x}\big)\,, \qquad \quad  A = d - \ve \left( d-1\right)\,.
\end{equation}
Component by component, choosing $k \neq l, \, k, l \neq x$, from the first condition we have 
\begin{align}
  \frac{A}{d} &= \left[ M_x \right]_{x, x} = \left[ K_{x, x} \right]_{x, x} + \sum_{j \neq x} \left[ K_{x, j} \right]_{x, x} = \tau + \left(d-1 \right) f\,, 
  \\
  0 &= \left[ M_x \right]_{x, k} = \left[ K_{x, x} \right]_{x, k} + \left[ K_{x, k} \right]_{x, k} + \sum_{j \neq x, k} \left[ K_{x, j} \right]_{x, k} = \gamma + g + \left( d-2 \right) s\,, 
  \\
  \frac{\ve}{d} &= \left[ M_x \right]_{k, k} = \left[ K_{x, x} \right]_{k, k} + \left[ K_{x, k} \right]_{k, k} + \sum_{j \neq x, k} \left[ K_{x, j} \right]_{k, k} = \alpha + h + \left( d-2\right) a \,,
  \\
 0 &= \left[ M_x \right]_{k, l} = \left[ K_{x, x} \right]_{k, l} + 2\left[ K_{x, k} \right]_{k, l} + \sum_{j \neq x, k, l} \left[ K_{x, j} \right]_{k, k} = \beta + 2t + \left( d-3\right) b \,. 
\end{align}
Similarly, again choosing $k \neq l, \, k, l \neq x$, from the second condition we find
\begin{align}
    \frac{1}{d} &= \left[ \frac{\id}{d}\right]_{jj} = \left[ K_{jj} \right]_{jj} + \sum_{x \neq j} \left[ K_{xj} \right]_{jj} = \tau + (d-1)h\,,
    \\
    0 &= \left[ \frac{\id}{d}\right]_{jk} = \left[ K_{jj} \right]_{jk} + \left[ K_{kj}\right]_{jk} + \sum_{x \neq j, k} \left[ K_{xj} \right]_{jk} = \gamma + g + \left( d-2 \right) t\,,
    \\
    \frac{1}{d} &= \left[ \frac{\id}{d}\right]_{kk} = \left[ K_{jj} \right]_{kk} + \left[ K_{kj}\right]_{kk} + \sum_{x \neq j, k} \left[ K_{xj} \right]_{kk} = \alpha + f + \left( d-2 \right) a\,,
    \\
    0 &= \left[ \frac{\id}{d}\right]_{kl} = \left[ K_{jj} \right]_{kl} + 2 \left[ K_{kj}\right]_{kl} + \sum_{x \neq j, k, l} \left[ K_{xj} \right]_{kl} = \beta + 2 s + \left( d-3 \right) b\,.
\end{align}
We can then solve for
\begin{equation}\label{eqnA: constants}
\begin{aligned}
    t &= s\,,
    \\
    \alpha &= -\left( d-2\right) a -h + \frac{\ve}{d}\,,
    \\
    \beta &= - \left( d-3 \right) b -2s\,,
    \\
    \gamma &= -g -\left( d-2\right) s\,, 
    \\
    \tau &= \frac{1}{d} - \left(d-1 \right)h\,,
    \\
    f &= h + \frac{1 - \ve}{d}\,.
\end{aligned}
\end{equation}
We can rewrite $K_{x, x}$ as 
\begin{equation}
    K_{x, x} = \begin{pmatrix}
        \tau & \gamma \sqrt{d-1} \vspace{0.15 cm}
        \\
        \gamma \sqrt{d-1} & \alpha + \left( d-2 \right) \beta 
    \end{pmatrix}_{\ket{x}, \, \ket{\psi_{\neq x}}} \; + \; \left( \alpha - \beta \right) \, \Pi_{\neq x}\,,
\end{equation}
where the first matrix is expressed in the basis $\{\ket{x}, \, \ket{\psi_{\neq x}}\}$. Similarly, using $s=t$ from \eqref{eqnA: constants}, for $x \neq j$, we can write $K_{x, j}$ as
\begin{equation}
    K_{x, j} = \begin{pmatrix}
        f & g & s \sqrt{d-2}
        \\
        g & h & s \sqrt{d-2}
        \\
        s \sqrt{d-2} & s \sqrt{d-2} & a+\left( d-3\right)b
    \end{pmatrix}_{\ket{x}, \, \ket{j}, \, \ket{\psi_{\neq x, j}}} + \; \left( a - b \right) \, \Pi_{\neq x, j}\,,
\end{equation}
where the constants $\alpha, \, \beta, \, \gamma, \, \tau$ and $f$ are as given in \eqref{eqnA: constants}.
Any $2 \times 2$ matrix is positive semidefinite if and only its trace and its determinant are non-negative, so the following constraints arise from enforcing $K_{x, x} \geq 0$,
\begin{align}\label{eqnA: Kxx pos}
    \tau \geq 0\,, \qquad \alpha \geq \beta \geq - \frac{\alpha}{d-2}\,, \qquad \tau \left( \alpha + \left( d-2\right) \beta \right) \geq \gamma^2 \left(d- 1\right)\,.
\end{align} 
Furthermore, by Sylvester’s criterion, the matrix $K_{x, j}$ for $j \neq x$ is positive semidefinite if and only if all of its principal minors are non-negative. This gives us the following conditions,
\begin{equation}\label{eqnA: Kxj pos}
    \begin{aligned}
        a \geq b &\geq -\frac{a}{d-3}\,,
        \\
      f \geq 0\,, \quad h &\geq 0\,, \quad   fh - g^2 \geq 0\,, 
        \\
        \big(a+b\left(d-3\right) \big) \left(fh-g^2 \right) &\geq s^2 \left(d-2 \right) \left(f+h-2g \right)\,.
    \end{aligned}
\end{equation}
Note that the second inequality implies $f+h \geq 2\sqrt{fh}\geq 2g$, with equality only when $f=g=h$. However, from \eqref{eqnA: constants} we know that $f - h = \frac{1- \ve}{d} \neq 0$, so we have $f+h-2g>0$. Moreover, the non-negativity of the principal minors $(1,3)$ and $(2,3)$ can be deduced from the last inequality combined with the fact that $f \left(f+h-2g \right)\geq fh-g^2$ and $h \left(f+h-2g \right)\geq fh-g^2$.

\begin{sublemma}\label{sublem: newK X}
    Let $\{\mathcal{K}_j\}$ be a decomposition for the POVM $\m$, with $\mathcal{K}_j=\{K_{x, j}\}_x$. If there exists a matrix $X \neq 0$ such that $K_{a, b} \geq X$ and $K_{b, a} \geq X$ for some pair $a \neq b$ then $\pg \left( \ket{\phi}, \, \m, \, \{ \mathcal{\tilde{K}}_{j}\}\right) \geq \pg \left( \ket{\phi}, \, \m, \, \{\mathcal{{K}}_{j}\}\right)$ for any state $\ket{\phi}$, where $\mathcal{\tilde{K}}_j= \{\tilde{K}_{x, j}\}_x$ has elements 
    \begin{equation}
        \begin{aligned}
        \tilde{K}_{x, x} &= K_{x, x} + X \left( \delta_{x, a} + \delta_{x,b} \right)\,,
        \\
            \tilde{K}_{x, j} &= K_{x, j} - X \left( \delta_{x, a} \delta_{j, b} + \delta_{x, b} \delta_{j, a}\right)\,, \quad \qquad j \neq x\,.    \end{aligned}
    \end{equation}
\end{sublemma}
\begin{proof}
    Since $K_{a, b} - X\geq 0$ and $K_{b, a} - X \geq 0$, with $a \neq b$, the new decomposition $\{\mathcal{\tilde{K}}_j\}$ is well-defined. The new guessing probability for any state $\ket{\phi}$ would then satisfy 
\begin{equation}
    \pg \left( \ket{\phi}, \, \m, \, \{\mathcal{\tilde{K}}_j \}\right)
    = \sum_{x} \langle \phi | \tilde{K}_{x, x} | \phi \rangle 
    = \sum_{x}  \langle \phi | K_{x, x} | \phi \rangle + 2\langle \phi | X | \phi \rangle 
    \geq \sum_{x}  \langle \phi | K_{x, x} | \phi \rangle
    = \pg \left( \ket{\phi}, \, \m, \, \{\mathcal{K}_j \}\right)\,.
\end{equation}
\end{proof}
We will apply Sublemma \ref{sublem: newK X} to $\{\mathcal{K}_j\}$ to define a new decomposition $\{\tilde{\mathcal{K}}_j\}$ which is also optimal. Since $f+h > 2g$ and $fh-g^2 \geq 0$, we are free to define a new variable $r := \frac{fh-g^2}{f+h - 2g} \geq 0$ and the matrix  
\begin{equation}
    X_{x, j} := 
    \begin{pmatrix}
      r & r & s \sqrt{\left(d-2 \right)}
      \\
     r & r & s \sqrt{\left(d-2 \right)}
     \\
     s \sqrt{\left(d-2 \right)} & s \sqrt{\left(d-2 \right)} & a + \left(d-3 \right)b
\end{pmatrix}_{\ket{x}, \, \ket{j}, \, \ket{\psi_{\neq x, j}}} 
+ \; \left( a - b \right) \Pi_{\neq x, j}\,.
\end{equation}
Using Sylvester's criterion again,
$X_{x, j}$ is positive semi-definite if and only if 
\begin{equation}
    r \geq 0\,, \qquad a \geq b \geq -\frac{a}{d-3}\,, \qquad r \big( a + \left(d-3 \right)b\big) \geq s^2 (d-2)\,.
\end{equation}
These conditions are  identical to those required for $K_{x, j}$ to be positive semidefinite, as stated in \eqref{eqnA: Kxj pos}. Then we have
\begin{equation}\label{eqnA: K - X}
\begin{aligned}
    K_{x, j} - X_{x, j} 
    &= \begin{pmatrix}
        f-r & g-r \vspace{0.15 cm}
        \\
        g-r & h-r \vspace{0.15 cm}
    \end{pmatrix}_{\ket{x}, \, \ket{j}}\,,
\hspace{5mm} \hspace{5mm}
   K_{j, x} - X_{x, j} 
   &= \begin{pmatrix}
        h-r & g-r \vspace{0.15 cm}
        \\
        g-r &f-r \vspace{0.15 cm}
    \end{pmatrix}_{\ket{x}, \, \ket{j}} \,. 
    \end{aligned}
\end{equation}
Both $K_{x, j} - X_{x, j}$ and $K_{j, x}- X_{x, j}$ are positive semidefinite if and only if 
\begin{equation}
    f+ h - 2r \geq 0\,, \qquad fh - g^2 \geq r\left( f+h -2g\right)\,.
\end{equation}
The second inequality is saturated by the definition of $r$, so both $K_{x, j} - X_{x, j}$ and $K_{j, x} - X_{x, j}$ have rank one, and the trace is 
\begin{equation}
    f + h -2r = \frac{\left(f-g\right)^2 + \left(h-g\right)^2}{f+h-2g} \geq 0\,.
\end{equation}
We are now ready to define the new POVMs $\tilde{\mathcal{K}}_j = \{\tilde{K}_{x, j}\}_x$, with 
\begin{equation}\label{eqnA: KTilde}
\begin{aligned}
 \tilde{K}_{x, x} & := K_{x, x} +\sum_{j \neq x} X_{x, j}\,,
 \\
 \tilde{K}_{x, j} & := K_{x, j} - X_{x, j}\,, \qquad \quad  x \neq j\,.
 \end{aligned}
\end{equation}
Using \eqref{eqnA: K - X} and \eqref{eqnA: KTilde}, we obtain
\begin{equation}
\begin{aligned}
\sum_{j \neq x}  \tilde{K}_{x, j} &= \left(d-1\right) \left(f-r \right)\ketbra{x}{x} + \left(g-r\right) \sum_{j\neq x} \big(\ketbra{x}{j}+\ketbra{j}{x} \big) + \left(h-r \right) \sum_{j\neq x} \ketbra{j}{j} 
\\
&= \begin{pmatrix}
 \left( d-1\right) \left( f-r\right) & \left( g-r\right) \sqrt{d-1} \vspace{0.1 cm}
 \\
 \left( g-r\right) \sqrt{d-1} & h-r
\end{pmatrix}_{ \ket{x}, \, \ket{\psi_{\neq x}}} + \; \left( h - r \right) \Pi_{\neq x}\,.
\end{aligned}
\end{equation}
Since both $M_{x}$ and $\sum_{j \neq x}  \tilde{K}_{x, j}$ exhibit  the same matrix structure as $K_{x,x}$, and given that $\tilde{K}_{x, x} = M_{x}-\sum_{j \neq x} \tilde{K}_{x, j}$,  it follows that $\tilde{K}_{x, x}$ enjoys the same symmetries as $K_{x,x}$. Thus, when $j \neq x$, we can write our new POVM elements in the computational basis in terms of a set of new coefficients $\{\tilde{\tau}, \, \tilde{\gamma}, \, \tilde{\alpha}, \, \tilde{\beta}, \, \tilde{f}, \, \tilde{g}, \, \tilde{h}, \, \tilde{s}, \, \tilde{t}, \, \tilde{a}, \, \tilde{b}\}$ as  
\begin{equation}\label{eqnA: Kmats3}
 \tilde{K}_{x, x} = \hspace{-2mm}
   \begin{array}{c}
    \scriptstyle |x \rangle \;\;
    \\[13ex]
    \end{array}
    \hspace{-3mm}
    \begin{pmatrix}
        \begin{array}{c|cccc}
            \tilde{\tau}    & \tilde{\gamma} & \tilde{\gamma} & \dots  & \tilde{\gamma} \\
            \hline 
          \vspace{-0.35 cm}  \\
             \tilde{\gamma}  & \tilde{\alpha} & \tilde{\beta}  & \dots  & \tilde{\beta}  \\
            \tilde{\gamma}  & \tilde{\beta}  & \tilde{\alpha} & \dots  & \tilde{\beta}  \\
            \vdots    & \vdots   & \vdots   & \ddots & \vdots   \\
            \tilde{\gamma}  & \tilde{\beta}  & \tilde{\beta}  & \dots  & \tilde{\alpha}
        \end{array}
    \end{pmatrix}   \,, 
    \qquad 
    \tilde{K}_{x, j} = 
    \hspace{-2mm}
    \begin{array}{c}
      \scriptstyle |x\rangle \\
      \scriptstyle |j\rangle \\[13ex]
    \end{array}
    \hspace{-2mm}
    \begin{pmatrix}
        \begin{array}{cc|cccc}
            \tilde{f} & \tilde{g} & \tilde{s} & \tilde{s} & \dots & \tilde{s} \\
            \tilde{g} & \tilde{h} & \tilde{t} & \tilde{t} & \dots & \tilde{t} \\
            \hline
            \vspace{-0.35 cm}  \\
            \tilde{s} & \tilde{t} & \tilde{a} & \tilde{b} & \dots & \tilde{b} \\
            \tilde{s} & \tilde{t} & \tilde{b} & \tilde{a} & \dots & \tilde{b} \\
            \vdots & \vdots & \vdots & \vdots & \ddots & \vdots \\
            \tilde{s} & t & \tilde{b} & \tilde{b} & \dots & \tilde{a}
        \end{array}
    \end{pmatrix}\,.
\end{equation}
Comparing $\tilde{K}_{x, j}$ in \eqref{eqnA: Kmats3} to $K_{x, j} - X_{x, j}$ in \eqref{eqnA: K - X}, we see that $\tilde{s}= \tilde{t}= \tilde{a}=\tilde{b} =0$, and
 since $\tilde{K}_{x, j}$ is rank-one, we know that $\tilde{f}\tilde{h}= \tilde{g}^2$.
Replacing them in \eqref{eqnA: constants}, we find 
\begin{equation}\label{eqnA: constants2}
\begin{aligned}
    \tilde{\alpha} &=  - \tilde{h} + \frac{\ve}{d}\,,
    \\
    \tilde{\beta} &= 0\,,
    \\
    \tilde{\gamma} &= -\tilde{g}\,, 
    \\
    \tilde{\tau} &= \frac{1}{d} - \left(d-1 \right)\tilde{h}\,,
    \\
    \tilde{f} &= \tilde{h} + \frac{1 - \ve}{d}\,,
\end{aligned}
\end{equation}
with $\tilde{g}^2 = \tilde{f} \tilde{h}$. We will now drop the tilde for ease of notation and assume that our constants satisfy \eqref{eqnA: constants2}. The positivity conditions \eqref{eqnA: Kxx pos} for $K_{x, x}$ give us 
\begin{equation}
    h \leq \frac{\ve}{d}, \, \qquad h \leq \frac{1}{(d-1)d}\,, \qquad h \leq \frac{\ve}{d^2}\,.
\end{equation}
Since the third condition is strictly stronger than the other two, we need only consider $h$ in the range $0 \leq h \leq \frac{\ve}{d^2}$.
Filling in the guessing probability in terms of $h$, we have 
\begin{align}
    \pg \left( \ket{\psi}, \, \m_d\right) &= \sum_x \langle  \psi | K_{x, x} | \psi\rangle = \frac{1}{d} \sum_{x} \begin{pmatrix}
        1 & 1 & 1 & \hdots & 1
    \end{pmatrix} \begin{pmatrix}
            \tau    & \gamma & \gamma & \dots  & \gamma \\
             \gamma  & \alpha & 0  & \dots  & 0  \\
            \gamma  & 0  & \alpha & \dots  & 0  \\
            \vdots    & \vdots   & \vdots   & \ddots & \vdots   \\
            \gamma  & 0  & 0  & \dots  & \alpha
    \end{pmatrix} \begin{pmatrix}
        1 \\
        1\\
        1\\
        \vdots \\
        1
    \end{pmatrix}
    \\
    &= \tau + \left(d-1\right) \left(2 \gamma + \alpha \right)  =  1 - \left(d-1\right) \left( \sqrt{h + \frac{1 - \ve}{d}} - \sqrt{h}\right)^2 \,,
    \end{align}
where we have chosen the negative root $\gamma = - \sqrt{fh}$, since it gives the higher guessing probability, and used $f= h + \frac{1- \ve}{d}$. We can now explicitly solve the guessing probability by maximising over the variable $h$,
\begin{equation}
    \pg \left( \ket{\psi}, \, \m_d\right) = \max_{h}  \; 1 - (d-1) \left( \sqrt{h + \frac{1 - \ve}{d}} - \sqrt{h}\right)^2 \qquad \text{subject to} \;\; 0 \leq h \leq \frac{\ve}{d^2} \,.
\end{equation}
Since the function $\left( \sqrt{h + \frac{1 - \ve}{d}} - \sqrt{h}\right)^2$ decreases continuously for $h \geq 0$, the optimal value is $h= \frac{\ve}{d^2}$, so we find
\begin{equation}
\pg \left( \ket{\psi}, \, \m_d\right) =  1 - \frac{d-1}{d^2} \left( \sqrt{A} - \sqrt{\ve} \right)^2  = \frac{1}{d^2} \left( \sqrt{A} + (d-1) \sqrt{\ve} \right)^2 = \frac{1}{d} \left( \tr \sqrt{M_1} \right)^2\,,
\end{equation}
proving Lemma \ref{lemma: qudit upper app}. The optimal POVMs $\mathcal{K}_j$ then look like
\begin{equation}
    K_{x, x} = \frac{1}{d^2}\begin{pmatrix}
       A & \sqrt{\left(d-1\right)A \ve} \vspace{0.1 cm}
        \\
       \sqrt{\left(d-1\right)A \ve} & \ve (d-1) 
    \end{pmatrix}_{\ket{x}, \, \ket{\psi_{\neq x}}} + \; \frac{\ve \left(d-1\right)}{d^2 } \, \Pi_{\neq x}\,, \qquad  K_{x, j} = \frac{1}{d^2} \begin{pmatrix}
        A & - \sqrt{A \ve} 
        \\
        -\sqrt{A \ve} & \ve 
    \end{pmatrix}_{\ket{x}, \, \ket{j}}\,,
\end{equation}
which is exactly the form of the decomposition \eqref{eqnA: psi decomp} used in Appendix \ref{app: upper qudit pguess}.
\end{proof}

\section{Other entropies }\label{app: dilation}
We prove in Theorem \ref{thm: qudit noisy} that 
\begin{equation}
    \pgs \left( \m_d \right) = \frac{1}{d} \left( \tr \sqrt{M_1} \right)^2\,.
\end{equation}
From \cite{Meng_2024}, $\pgs \left( \m_d \right)$ is equal to the guessing probability of the state $\rho_\psi$ minimized over all projective measurements, i.e.
\begin{equation}
    \pgs \left( \m_d \right) = \pgs \left( \rho_{\psi} \right)\,,
\end{equation}
where
\begin{equation}
    \rho_\psi= \Delta_{\ve} \left( \ketbra{\psi}{\psi}\right)= \left( 1 - \ve \right) \ketbra{\psi}{\psi}+ \frac{\id}{d} \ve\,.
\end{equation}
This leads us to wonder whether the maximal intrinsic randomness of $\m_d$ and $\rho_\psi$ are equivalent when quantified by a conditional entropy \emph{other} than the min-entropy. Following \cite{Meng_2024}, in this section, we consider the maximal conditional von Neumann and max-entropies of $\m_d$. When Eve receives an ensemble $\{ p(x)\, \rho_{x, E}\}$ of subnormalized states labelled by  
$x \in X$, the von Neumann entropy of the register $X$ conditioned on Eve's information is given by
\begin{equation}
    H \left( X | E, \, \{ p(x) \, \rho_{x, E}   \}  \right)  = H \left( \{ p(x) \} \right) + \sum_{x} p(x) S \left( \rho_{x, E} \right) - S \left( \sum_{x} p(x) \rho_{x, E} \right)\,,
\end{equation}
where $S (\sigma)= -\tr \left( \sigma \log \sigma \right)$ is the von Neumann entropy of the state $\sigma$ and $H \left( \{p(x)\}\right)= - \sum_{x} p(x) \log p(x)$ is the Shannon entropy of the probability distribution $\{p(x)\}$. From \cite{Konig_2009}, the max-entropy of $X$ conditioned on Eve's information can be written as 
\begin{equation}\label{eqnA: hmax def}
    \hmax \left( X | E, \{p(x) \, \rho_{x, E} \}\right) = \log \psecr \left( \{p(x) \, \rho_{x, E} \}
 \right) \,, \qquad \psecr \left( \{p(x) \, \rho_{x, E} \}
 \right) = \max_{\sigma} \left( \sum_{x} \sqrt{p(x)} F \left( \rho_{x, E}, \, \sigma\right) \right)^2\,,
\end{equation}
where $F(\rho, \sigma)$ is the quantum fidelity 
\begin{equation}
    F(\rho, \sigma)= \tr \sqrt{\sqrt{\rho}\, \sigma \sqrt{\rho}}\,.
\end{equation}
From \cite{Meng_2024}, the maximal conditional von Neumann and max-entropies of the state $\rho_\psi$ are 
\begin{equation}
\begin{aligned}
H^{*} \left( X | E, \, \rho_{\psi} \right) &= \log d - S \left( \rho_{\psi}\right)\,,
\\
  \hmaxs \left( X | E, \, \rho_{\psi} \right) &=  \log d + \log \lmax \left( \rho_{\psi} \right)\,.    \end{aligned}
\end{equation}
The maximal randomness of the noisy measurement $\m_d$ would involve a minimization by Eve over all valid Naimark dilations $\Pi_{SA}, \, \ket{\varphi}_{AE}$ and a minimization by Alice over all states $\ket{\phi}_S$ (we temporarily reinstate the subscript $S$ to denote Alice's system). In the case of the conditional von Neumann and max-entropies, then, our figures of merit would be 
\begin{equation}
    H^{*} \left( X | E, \, \m_d \right) = \max_{\ket{\phi}_{S}} \, H \left( X | E, \, \ket{\phi}_{S}, \, \m_{d, S} \right)\,, \qquad \hmaxs \left( X | E , \, \m_d \right) = \max_{\ket{\phi}_{S}} \, \hmax \left( X | E , \,\ket{\phi}_{S}, \, \m_{d, S} \right)\,,
\end{equation}
where $H \left( X | E, \, \ket{\phi}_{S}, \, \m_{d, S} \right)$ and $\hmax \left( X | E, \,\ket{\phi}_{S}, \, \m_{d, S} \right)$ are given by
\begin{equation}\label{eqnA: h opt}
\begin{aligned}
H^{*} \big( X | E, \, \ket{\phi}_S, \, \m_{d, S} \big)= \; &  \;\;\;  \min_{ \Pi_{SA}, \, \ket{\varphi}_{AE} }
& & H \left( X | E, \, \{ p(x) \, \rho_{x, E}  \}\right) \\
& \quad  \,  \text{subject to}
& &  p(x) \rho_{x, E } = \tr_{SA} \Big( \Pi_{x, SA} \otimes \id_{E} \, \, \ketbra{\phi}{\phi}_{S} \otimes \ketbra{\varphi}{\varphi}_{AE} \Big) \\
& 
& & \tr_A \Big( \Pi_{x, SA}\big(\id_S\otimes \tr_{SE} \left( \ketbra{\phi}{\phi}_{SAE} \right) \big) \Big)=M_{x, S} \; \text{ for all } \; x\,
\end{aligned}
\end{equation}
and  
\begin{equation}\label{eqnA: hmax opt}
\begin{aligned}
\hmaxs \big( X | E, \, \ket{\phi}_{S}, \, \m_{d, S} \big)= \; &  \;\;\;\min_{ \Pi_{SA}, \, \ket{\varphi}_{AE} }
& & \hmax \left( X | E, \, \{ p(x) \, \rho_{x, E}  \}\right) \\
& \quad  \,  \text{subject to}
& &  p(x) \rho_{x, E } = \tr_{SA} \Big( \Pi_{x, SA} \otimes \id_{E} \, \, \ketbra{\phi}{\phi}_{S} \otimes \ketbra{\varphi}{\varphi}_{AE} \Big) \\
& 
& & \tr_A \Big( \Pi_{x, SA}\big(\id_S\otimes \tr_{SE} \left( \ketbra{\phi}{\phi}_{SAE} \right) \big) \Big)=M_{x, S} \; \text{ for all } \; x\,
\end{aligned}
\end{equation}
respectively. We do not solve for \eqref{eqnA: h opt} or \eqref{eqnA: hmax opt} in this work, but we take a step in this direction by upper bounding $H \left( X | E, \, \ket{\psi}, \, \m_d \right)$ and $\hmax \left( X | E , \, \ket{\psi}, \, \m_d \right)$, where $\ket{\psi}$ is the unbiased state (we drop the subscript $S$ again when the system is clear from the context).

When considering the intrinsic randomness of the state $\rho_{\psi}$, we can freely fix Eve's purification of the state, as all purifications are equivalent up to local unitaries. To quantify the intrinsic randomness of a non-extremal measurement like $\m_d$, however, we cannot fix Eve's Naimark dilation without loss of generality (see \cite[Appendix B]{frauchiger_2013}). To emphasize this point, we present a dilation for $\m_2$ that gives Eve no side-information at all when Alice uses $\ket{\psi}$, even though the auxiliary state $\sigma_{A}
$ is mixed. Consider the dilation 
\begin{equation}
    \Pi_{x, SA} = \sum_{k} \ketbra{k}{k} \otimes \ketbra{k-x+1}{k-x+1}\,, \qquad \ket{\varphi}_{AE}= \sum_{j} \left( \sqrt{M_1} \otimes \id \right) \ket{j, j}\,. 
\end{equation}
This dilation is valid because
\begin{equation}
    \tr_{A} \Big( \Pi_{x, SA} \, \id_{S} \otimes \sigma_{A} \Big) = \frac{1}{d} \big( \left( A - \ve\right) \ketbra{x}{x} + \ve \id \big) = M_x\,.
\end{equation}
Eve receives the conditional states 
\begin{equation}
    p(x) \rho_{x, E} = \tr_{SA} \bigg( \Pi_{x, SA} \otimes \id_{E} \, \ketbra{\phi}{\phi}_{S} \otimes \ketbra{\varphi}{\varphi}_{AE} \bigg) = \sum_{j} \langle \phi | j+x-1 \rangle^2 \sqrt{M_1} \ketbra{j}{j} \sqrt{M_1}\,, 
\end{equation}
so when $\ket{\phi}= \ket{\psi}$, we see that Eve's conditional states are completely indistinguishable.

We now consider a different dilation of $\m_d$ that minimizes the conditional min-entropy when Alice chooses $\ket{\psi}$. From \cite[Appendix A]{Senno_2023}, we can derive a canonical Naimark dilation corresponding to Eve's chosen POVM decomposition $\{p(j)\,, \,   \mathcal{N}_j\}$. Consider a bipartite auxiliary system $A_1 A_2$ and define the purification $\ket{\psi}_{A_1 A_2 E}$ and the action of a unitary $U$ as 
\begin{equation}
\ket{\psi}_{A_1 A_2 E} = \ket{0}_{A_1} \otimes  \sum_{j} \sqrt{p(j)} \ket{j j}_{A_2 E}\,, \quad  U \ket{\varphi, 0, j}_{S A_1 A_2} = \sum_{x} \sqrt{N_{x, j, S}} \ket{\varphi, x, j } \;\; \text{for all} \;\; \ket{\varphi}\,.   
\end{equation}
The projective measurement on the joint state and auxiliary system is then given by 
\begin{equation}
    \Pi_{x, SA} = U^{\dagger} \left( \id_{S} \otimes \ketbra{x}{x}_{A_1} \otimes \id_{A_2} \right) U\,.
\end{equation}
Eve's conditional post-measurement states, when Alice chooses the state $\ket{\phi}$, are given by 
\begin{equation}
    p(x) \rho_{x, E} =  \tr_{S A_1 A_2} \Big( \Pi_{x, S A_1 A_2} \otimes \id_{E} \ketbra{\phi}{\phi}_{S } \otimes \ketbra{\Psi}{\Psi}_{A_1 A_2 E}\Big) = \sum_{j} p(j) \langle{\phi | N_{x, j, S} | \phi } \rangle \ketbra{j}{j}\,.
\end{equation}
\begin{figure}[h]
    \centering
\begin{tikzpicture}
\begin{scope}[scale=0.8]
    \begin{axis}[
        axis lines=middle,
        xlabel={$\ve$},
        xlabel style={at={(ticklabel* cs:1.02)}, anchor=west},
        ymin=0, ymax=1,
        xmin=0, xmax=1,
        domain=0:1,
        samples=200,
        legend pos=north east,
        legend style={font=\small}
    ]
    \addplot[
    dotted,
    thick,  
    black,
    domain=0:1,
    samples=50   
] {ln(2-x)/ln(2)};
        \addlegendentry{$\hmax \left( \ket{\psi}, \, \m_2 \right)$ upper bound};
        \addplot[
            dashed,
            dash pattern=on 5pt off 3pt,
            black
        ] {- ((1+sqrt(x*(2-x)))/2)*ln((1+sqrt(x*(2-x)))/2)/ln(2) - ((1-sqrt(x*(2-x)))/2)*ln((1-sqrt(x*(2-x)))/2)/ln(2)};
        \addlegendentry{$H\left( \ket{\psi}, \, \m_2 \right)$ upper bound};
        \addplot[
            black
        ] {1 + ((2-x)/2)*ln((2-x)/2)/ln(2) + (x/2)*ln(x/2)/ln(2)};
        \addlegendentry{$H^{*} \left( \rho_{\psi} \right)$};
        
        \addplot[
    dashed,
    mark=*,
    mark size=1.25pt,
    mark repeat=10,
    black
] {-ln(1+ sqrt(x*(2-x)) )/ln(2)+1};
        \addlegendentry{$\hmins \left( \m_2 \right)$=$\hmins \left( \rho_{\psi} \right)$};
        
    \end{axis}
\end{scope}
\end{tikzpicture}
\caption{Conditional entropies arising from the noisy qubit measurement $\m_d$ or the noisy qubit state $\rho_{\psi}$. Note that the upper bound on $\hmax \left( \ket{\psi}, \, \m_2\right)$ is equal to $\hmaxs \left( \rho_\psi\right)$.}
\label{fig: entropies}
\end{figure}
In the case of the noisy projective measurement $\m_d$ and the unbiased $\ket{\psi}$ state, using the `square root' decomposition with $\{\mathcal{K}_j\}$ from \eqref{eqnA: qudit decomp}, we find $p(x)=1/d$ and
\begin{equation}\label{eqnA: cond states}
\begin{aligned}
\rho_{x, E} &= d \sum_{j} \langle \phi | K_{x, j} | \phi \rangle \ketbra{j}{j} = \frac{1}{d^2} \Big( \left( \sqrt{A} + \Big( d-1\right) \sqrt{\ve}  \Big)^2 \ketbra{x}{x} + \frac{1}{d^2} \left( \sqrt{A} - \sqrt{\ve} \right)^2    \id_{\neq x}   
\\
&= d \sum_{j} \langle j | \sqrt{\rho_{\psi}} | x \rangle^2 \ketbra{j}{j} = \mathcal{C} \left( \frac{1}{d} \sqrt{\rho_{\psi}} \ketbra{x}{x} \sqrt{\rho_{\psi}} \right)\,,
\end{aligned}
\end{equation}
where $\mathcal{C} \left( . \right)$ is the completely positive trace-preserving map
\begin{align}
    \mathcal{C} \left( \sigma \right) = \sum_{j} \langle j | \sigma | j \rangle \ketbra{j}{j}\,.
\end{align}
From the data-processing inequality for the von Neumann entropy (see e.g. \cite[Theorem 11.15 (3)]{nielsen00}), such a classical post-processing cannot decrease the conditional von Neumann entropy, so we have 
\begin{equation}
\begin{aligned}
     H \left(X | E, \, \{ p(x), \, \rho_{x, E}   \}  \right) &=  H \left(X | E, \, \{ p(x), \, \mathcal{C} \left( d \sqrt{\rho_{\psi}} \ketbra{x}{x} \sqrt{\rho_{\psi}} \right)  \} \right) 
     \\
     &\geq H \left(X | E, \, \{ p(x), \, d \sqrt{\rho_{\psi}} \ketbra{x}{x} \sqrt{\rho_{\psi}}  \} \right) = \log d - S \left( \rho_{\psi} \right) = H^{*} \left( \rho_{\psi} \right)\,.
     \end{aligned}
\end{equation}
On the other hand, since Eve can minimize over all dilations, we can upper bound $H \left( \ket{\psi},\,  \m_d \right)$ by 
\begin{equation}
\begin{aligned}
H(\ket{\psi}, \, \m_d) \leq  H \left(X | E, \, \{ p(x), \, \rho_{x, E}   \}  \right) =   H_2 \left( \frac{1}{d} \left( \tr \sqrt{M_1} \right)^2 \right) + \left( 1 - \frac{1}{d} \left( \tr \sqrt{M_1} \right)^2\right) \log \left( d-1\right)\,,    
\end{aligned}
\end{equation}    
where $H_2 \left(. \right)$ is the binary Shannon entropy. In the qubit case, the right-hand side is known as the coherence of formation of $\rho_{\psi}$ with respect to the unbiased basis $\{\ket{x}\}$ \cite{Yuan_2015}.

Similarly, we upper bound $\hmax \left( \ket{\psi}, \, \m_d \right)$ by solving the optimization problem \eqref{eqnA: hmax def} for $\hmax \left( X | E, \, \{ p(x) \rho_{x, E} \}\right)$. We find 
\begin{equation}
\begin{aligned}
    \psecr \left( \{ p(x) \, \rho_{x, E}\} \right) &= \max_{\sigma} \left( \sum_{x} \sqrt{p(x)} \tr \sqrt{ \sqrt{ \rho_{x, E} } \, \sigma \,  \sqrt{ \rho_{x, E} }   } \right)^2 \geq   \left( \frac{1}{d} \sum_{x} \tr  \sqrt{ \rho_{x, E}  } \right)^{2}
    \\
    &= \frac{1}{d^2} \left(   \sqrt{A} + (d-1)\sqrt{\ve} + (d-1) \left( \sqrt{A} - \sqrt{\ve} \right)  \right)^2 = A\,,
    \end{aligned}
\end{equation}
where in the first line we use $\sigma = \id/d$. From the other direction, $\psecr \left( \{ p(x) \, \rho_{x, E} \} \right)$ can be expressed \cite[Lemma 6.6]{Tomamichel_2016} as the following minimization problem,
\begin{equation}
    \psecr \left( \{ p(x) \, \rho_{x, E} \} \right) =  \min_{Y_{XE} >0} \sum_{x} p(x) \tr \Big( \ketbra{x}{x} \otimes \rho_{x, E} \,\, Y_{XE}^{-1} \Big) \lmax \left( Y_{E} \right)\,, \qquad Y_{E} = \tr_{X} Y_{XE}\,. 
\end{equation}
We can upper bound $\psecr\left( \{ p(x) \, \rho_{x, E} \} \right)$ by taking
\begin{equation}
    Y_{XE} = \sum_{x} \ketbra{x}{x} \otimes \sqrt{p(x) \rho_{x, E}}\,,
\end{equation}
such that 
\begin{equation}
    \psecr \left( \{ p(x) \, \rho_{x, E} \} \right) \leq  \bigg( \sum_{x} \tr \sqrt{p(x) \rho_{x, E}} \bigg) \lmax \left( \sum_{x} \sqrt{p(x) \rho_{x, E}}   \right) = \tr \left( \sqrt{\frac{A}{d} \id}\right) \lmax \left( \sqrt{\frac{A}{d} \id}\right) = A\,.
\end{equation}
Since the upper and lower bounds on $\psecr \left( \{ p(x) \, \rho_{x, E} \} \right)$ meet, we have
\begin{equation}
  \hmax \left( \ket{\psi}, \, \m_d \right)\leq \hmax \left( X | E, \, \{ p(x) \rho_{x, E}\} \right) = \log \psecr \left( \{p(x) \, \rho_{x, E} \}\right) = \log A  = \hmaxs \left( \rho_{\psi}\right)\,. 
\end{equation}
In Figure \ref{fig: entropies}, we compare our upper bounds on the randomness of $\m_2$ with the state $\ket{\psi}$ with the maximal conditional entropies for the state $\rho_\psi$.

\section{Coarse-grained noisy measurements}\label{app: coarse grain}
Imagine, in the case where $d$ is even and greater than 2, that rather than carry out the full measurement $\m_d$, Alice performs a coarse-grained version of it with only two outcomes. We label this coarse-grained POVM $\hat{\m}$, such that
\begin{equation}
    \hat{\m} = \{\hat{M}_a\}_{a}\,, \quad \hat{M}_{a} = \sum_{x \in S_{a}} M_x\,, \quad M_{x}= \left(1 - \ve \right)\ketbra{x}{x} + \frac{\ve}{d}\id\,, \quad a \in \{1, 2\}\,,
\end{equation}
where
\begin{equation}
   S_1= \left\{\, 1, \, ..., \, \frac{d}{2} \, \right\}\,, \quad S_{2} = \left\{\,  \frac{d+2}{2}, \,  ..., \, d \,\right\} 
\end{equation}
are subsets from the set $\{1, ..., d\}$ with $d/2$ elements each. We denote by $\mathcal{H}_1$ and $\mathcal{H}_2$ the Hilbert spaces spanned by the vectors in $S_1$ and $S_2$, respectively, and, for clarity, here we label the elements of the noisy qubit projective measurement $\m_2$ by $M_{x, 2}$. Using square brackets to represent a matrix operating on the direct sum space $\mathcal{H}_1 \oplus \mathcal{H}_2$, let $\mathcal{C}\left( .\right)$ be a map that `inflates' any $2 \times 2$ matrix $F$ by
\begin{equation}
    \mathcal{C} \left( F \right) = \begin{bmatrix}
        F_{11} \id & F_{12} \id\\
        F_{21} \id & F_{22} \id\\
    \end{bmatrix}\,, \quad \; \textnormal{where} \; F = \begin{pmatrix}
        F_{11} & F_{12}\\
        F_{21} & F_{22}\\
    \end{pmatrix}\,.
\end{equation}
Because of the symmetries inherent in $\m_d$ and in Alice's coarse-graining method, we can write the elements of $\hat{\m}$ simply as
\begin{equation}
    \hat{M}_{a} = 
    \mathcal{C} \left( M_{a, 2}\right)\,.
\end{equation}
From Corollary \ref{corr: qudit 2-outcome app}, we know that the guessing probability of $\hat{\m}$ is bounded by 
\begin{equation}\label{eqnA: cgqudit bound}
    \pgs \left( \hat{\m} \right) \leq 1- \frac{1}{2} \left( \sqrt{\frac{2-\ve}{2}} - \sqrt{\frac{\ve}{2}} \right)^2 = \frac{1}{2} \Big(1 + \sqrt{\ve \left( 2 - \ve\right)} \Big)= \pgs \left( \m_2 \right)\,.
\end{equation}
In the following, we prove by finding a suitable decomposition for Eve that 
\begin{equation}\label{eqnA: cg equal}
    \pgs \left( \hat{\m} \right) = \pgs \left( \m_2 \right)\,.
\end{equation}
We can represent Alice's chosen state $\ket{\phi}$ as the direct sum 
\begin{equation}\label{eqnA: cg phi}
 \ket{\phi} =  \begin{bmatrix}
 \alpha_{1} \ket{\phi_1}\\
  \alpha_2 \ket{\phi_2}   
 \end{bmatrix}\,, \qquad  \sum_{i} \alpha_{i}^2 =1\,, \qquad |\langle \phi_1 | \phi_1 \rangle|^2 = |\langle \phi_2 | \phi_2 \rangle|^2 =1\,.
\end{equation}
To simplify matters, we are free to assume that $\ket{\phi_1}=\ket{\phi_2}$, because if they were not equal, we could simply apply a unitary of the form $\id \oplus U$ to the state $\ket{\phi}$ such that 
\begin{equation}
    U \ket{\phi_2} = \ket{\phi_1}
\end{equation}
without changing the POVM $\hat{\m}$. By a similar argument, we can also assume that $\alpha_2$ and $\alpha_2$ are real. Eve can then use the `inflated' decomposition $\{\mathcal{K}_j\}$, with $\hat{\mathcal{K}}_j=\{ \hat{K}_{x, j} \}_{x}$ such that
\begin{equation}
     \hat{K}_{x, j}  = \mathcal{C} \left( K_{x, j, 2} \right)\,,
\end{equation}
where $\mathcal{K}_{j, 2}=\{K_{x, j, 2}\}_{x}$ are Eve's qubit POVMs \eqref{eqnA: qudit decomp} for the state $\ket{\alpha}= \left(\alpha_1, \, \alpha_2 \right)^{T}$, taking the components $\{\alpha_i\}$ from \eqref{eqnA: cg phi}. The elements $K_{j, j, 2}$ which are relevant for the guessing probability are 
\begin{equation}
    K_{1, 1, 2} = 
    \frac{1}{2}\begin{pmatrix}
     \alpha_1^2 \left( 2- \ve\right) & \alpha_1 \alpha_2 \sqrt{\ve \left(2 - \ve \right)} \vspace{0.1 cm}\\
     \alpha_1 \alpha_2 \sqrt{\ve \left( 2 - \ve \right)} & \alpha_2^2 \ve
    \end{pmatrix}\,, 
    \qquad 
    K_{2, 2, 2} = 
    \frac{1}{2}\begin{pmatrix}
     \alpha_1^2 \ve & \alpha_1 \alpha_2 \sqrt{\ve \left(2 - \ve \right)} \vspace{0.1 cm}\\
     \alpha_1 \alpha_2 \sqrt{\ve \left( 2 - \ve \right)} & \alpha_2^2 \left( 2 - \ve\right)
    \end{pmatrix}\,. 
\end{equation}
Eve's guessing probability is then bounded by 
\begin{equation}
\begin{aligned}
\pg \left( \ket{\phi}, \, \hat{\m} \right) &\geq \sum_{j} \langle \phi | \hat{K}_{j, j} | \phi \rangle = \sum_{j} 
\begin{bmatrix}
 \alpha_1 \bra{\phi_1}, & \alpha_2 \bra{\phi_1}   \end{bmatrix}
 \;
\mathcal{C}  \left( K_{j, j, 2} \right) 
\; \begin{bmatrix}
    \alpha_1 \ket{\phi_1}\\
    \alpha_2 \ket{\phi_1}
\end{bmatrix}
\\
&= \alpha_1^4 + \alpha_2^{4} + 2 \alpha_1^2 \alpha_2^2 \, \sqrt{\ve \left( 2- \ve\right)}= 1 - \alpha_1^2 \left(1 - \alpha_1^2 \right) \left( \sqrt{2- \ve} - \sqrt{\ve} \right)^2\,. \label{eqnA: alphas}
\end{aligned}
\end{equation}
The right-hand side of \eqref{eqnA: alphas} is minimized if and only if $\alpha_1^2= \alpha_2^2= \frac{1}{2}$, so we find  
\begin{equation}\label{eqnA: cg lower bound}
 \pg \left( \ket{\phi}, \, \hat{\m} \right)  \geq \frac{1}{2} \Big( 1 + \sqrt{\ve \left( 2- \ve\right)} \Big) = \pgs \left(\m_2 \right)\,,
\end{equation}
with strict inequality when $\alpha_1^2, \alpha_2^2 \neq 1/2$. Combining \eqref{eqnA: cg lower bound} with Corollary \ref{corr: qudit 2-outcome app}, we prove \eqref{eqnA: cg equal}.

Eve could, however, have taken the decomposition $\{\mathcal{K}_j\}$ for the original POVM $\m_d$, from \eqref{eqnA: qudit decomp}, and coarse-grained it just like Alice does to obtain $\hat{\m}$. Let's define Eve's coarse-grained decomposition as 
\begin{equation}
    \hat{K}_{a, b} := \sum_{x \in \mathcal{S}_a,\, j \in \mathcal{S}_b} K_{x, j}\,.
\end{equation}
Fixing the unbiased state $\ket{\psi}$, Eve's guessing probability with the POVMs $\hat{\mathcal{K}}_b=\{\hat{K}_{a, b}\}_a$ is then 
\begin{align}
    \pg \left(\ket{\psi}, \, \hat{\m}, \, \{\hat{\mathcal{K}}_{b} \}\right) &= \sum_{a} \langle{\psi | \hat{K}_{a, a} | \psi } \rangle = \sum_{a} \sum_{x \in S_{a},\, j \in S_a} \langle{\psi | K_{x, j} | \psi } \rangle 
    \\
    &= \sum_{a} \bigg( \frac{1}{d^3} \sum_{x \in S_a} \Big( \sqrt{A} + \left(d-1 \right) \sqrt{\ve} \Big)^2 + \frac{1}{d^3} \sum_{x \in S_a} \sum_{j \in S_a, \, j \neq x} \big( \sqrt{A} - \sqrt{\ve} \big)^2 \bigg) 
    \\
    &= \frac{1}{2d^2} \bigg( Ad + 2d \sqrt{A \ve} + \ve \Big( 2 \left( d-1 \right)^2 + d-2 \Big) \bigg) 
    \\
    &= \frac{1}{2} + \frac{\sqrt{\ve}}{2d} \Big( 2 \sqrt{A} + \sqrt{\ve} \left( d-2 \right) \Big)\,, 
\end{align} 
where in the last line we use \eqref{eqn: A_meaning}\,.
We aim to show that this guessing probability is strictly suboptimal, i.e.
\begin{equation}
  \pg \left(\ket{\psi}, \, \hat{\m}, \, \{\hat{\mathcal{K}}_b \} \right) =  \frac{1}{2} + \frac{\sqrt{\ve}}{2d} \Big( 2 \sqrt{A} + \sqrt{\ve} \left( d-2 \right) \Big) < \frac{1}{2} \Big(1 + \sqrt{\ve \left( 2 - \ve\right)} \Big) = \pgs \big( \hat{\m} \big)\,,
\end{equation}
or 
\begin{equation}\label{eqn: ineq_cg}
  \frac{1}{d} \Big( 2 \sqrt{A} + \sqrt{\ve} \left( d-2 \right) \Big)  < \sqrt{2-\ve}\,.
\end{equation}
Introducing the parameter
\begin{equation}
    t = \frac{2}{d}\,, \qquad 0 < t <1\,, 
\end{equation}
we see that the left-hand side of \eqref{eqn: ineq_cg} can be rewritten as a convex mixture
\begin{equation}
 \frac{1}{d} \Big( 2 \sqrt{A} + \sqrt{\ve} \left( d-2 \right) \Big)=    t \sqrt{A} + \left( 1 - t\right) \sqrt{\ve}  = t f \left( A \right) + \left( 1 - t\right) f \left( \ve \right)\,,
\end{equation}
where we define the square root function $f \left(x \right) = \sqrt{x}$. Note that $f \left( x \right)$ is strictly concave, so, when $0 < \ve < 1$, we have 
\begin{equation}
   t f \left( A \right) + \left( 1 - t\right) f \left( \ve \right) < f \big(tA + \left( 1-t\right) \ve \big) =  f \left( 2 - \ve \right) = \sqrt{2 - \ve}\,,
\end{equation}
with equality when $\ve=0$ or $1$. We conclude that, for the unbiased state $\ket{\psi}$, the coarse-graining $\{\hat{\mathcal{K}}_b\}$ of the decomposition $\{\mathcal{K}_j\}$, which was optimal for the original POVM $\m_d$, is strictly sub-optimal for the coarse-grained POVM $\hat{\m}$.

\section{Noisy state and measurement}\label{app: mixed noise}
\begin{figure}[h]
    \centering
\begin{tikzpicture}[line cap=round, line join=round, >=Triangle]
     \centering
    \begin{scope}[scale=1]
  \clip(-2.9,-2.6) rectangle (2.8,2.8);
  \draw(0,0) circle (2cm);
  \draw (-2,0) node[anchor=east] {\small{$\ket{{-}}$}};
  \draw (2,0) node[anchor=west] {\small{$\ket{{+}}$}};
  \draw (0,2) node[anchor=south] {\small{$ {|0\rangle}$}};
  \draw (0,-2) node[anchor=north] {\small{$|{1}\rangle$}};

  \draw[dashed] (0,0)-- (0, 1.7);
\draw [->] (0,0)-- (1.05357, 1.7);
  \draw [->] (0,0)-- (-1.05357, 1.7);
  \draw [dotted] (-1.05357, 1.7)-- (1.05357, 1.7);

  \draw [fill] (0, 1.7) circle (1.5pt);

\draw[dashed] (0,0)-- (1.7, 0);
   \draw[dashed] (0,0)-- (-1.7, 0);
 
\draw [->] (0,0)-- (1.7, 1.05357);
  \draw [->] (0,0)-- (1.7, -1.05357);
  \draw [->] (0,0)-- (-1.7, 1.05357);
  \draw [->] (0,0)-- (-1.7, -1.05357);
  \draw [dotted] (1.7, -1.05357)-- (1.7, 1.05357);
  \draw [dotted] (-1.7, -1.05357)-- (-1.7, 1.05357);

  \draw [fill] (1.7,0) circle (1.5pt);
  \draw [fill] (-1.7,0) circle (1.5pt);
\end{scope}
\end{tikzpicture}
\caption{A joint decomposition of a noisy state and measurement pair with noise parameter $\ve=0.15$.}
\label{fig: joint noise app}
\end{figure}
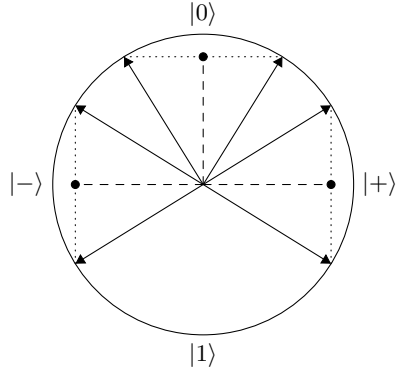
Here, we consider a simple classical model where Eve can decompose both Alice's state and her measurement at the same time. Take the noise parameter $\ve$, with $0 < \ve < 1$, and $\ket{\psi}= \frac{1}{\sqrt{2}}\left(1, 1 \right)^{T}$, and let Alice use the qubit POVM $\m_2$, which is diagonal in the basis $\{\ket{x}\}$, 
and the noisy qubit state $\rho_\psi$, where 
\begin{equation}
    \rho_{\phi} = \left(1 - \ve \right) \ketbra{\phi}{\phi} + \frac{\ve}{2} \id\,. 
\end{equation}
We allow Eve to choose a state decomposition $\{\ket{\varphi_{i,\, \lambda}}\}$ and a decomposition into POVMs $\mathcal{N}_{j, \lambda}=\{{N}_{x, j, \lambda} \}_x$ (note that both are normalized) according to some joint probability distribution $\{p(i, j, \lambda)\}$, where we have explicitly added an extra hidden variable $\lambda$ to the model considered in \cite{Senno_2023}. Eve's guessing probability is then 
\begin{equation}\label{eqnA: pguess mixed}
\begin{aligned}
\pg \big( \rho,\, \m \big)= \; & \;\;  \max_{ \{p(i,j, \lambda), \, \ket{\varphi_{i, \lambda}}, \, \mathcal{N}_j \} }
& & \sum_{i,j, \lambda}p(i,j, \lambda) \max_x \bra{\varphi_{i,\, \lambda}}N_{x,j, \lambda}\ket{\varphi_{i, \lambda}} 
\\
& \qquad \;\; \text{subject to}
& &  N_{x, j, \lambda}  \geq 0
\\
& 
& &  \sum_{x} N_{x, j, \lambda} = \id
\\
& 
& & \sum_{i,j, \lambda} p(i,j, \lambda)\ketbra{\varphi_{i, \lambda}}{\varphi_{i, \lambda}} = \rho \\
& 
& & \sum_{i,j} p(i,j) \, N_{x,j, S} = M_{x, S} \;  \text{ for all } \; x
\\
& 
& & \sum_{i,j}p(i,j, \lambda)\langle{\varphi_{i, \lambda}|N_{x,j, \lambda}|\varphi_{i, \lambda}} \rangle=\tr \big( M_{x} \, \rho \big)\,.
\end{aligned}
\end{equation}
Define 
\begin{equation}
    \ve^* := 1 - 1/\sqrt{2}\,.
\end{equation}
For $\ve \leq \ve^{*}$, let's say Eve uses the following decomposition, which is independent of $\lambda$, 
\begin{equation}\label{eqnA: decomp mixed noise}
     \begin{aligned}
      p(i, j) &=           \frac{1}{2} \delta_{i, j}\,,
    \\
    \ket{\varphi_i} &= \sqrt{2} \sqrt{\rho_\psi} \ket{i}\,,
    \\
    N_{x, j} &= 
    \begin{cases}
    2 \sqrt{M_j} \ketbra{\psi}{\psi} \sqrt{M_j}\,, & x =j\,,\\
    \id - 2 \sqrt{M_j} \ketbra{\psi}{\psi} \sqrt{M_j}\,, & x  \neq j\,.\\
    \end{cases}    
     \end{aligned}
\end{equation}
A schematic of such a joint decomposition is shown in Figure \ref{fig: joint noise app}.
We can then bound Eve's guessing probability by 
\begin{equation}
 \pg \big(\rho_\psi, \, \m_2 \big) \geq 2 \sum_{j} \bra{j} \,\sqrt{\rho_\psi} \,\sqrt{M_j} \,\ket{\psi}^2 = \frac{1}{2} \bigg( 1 - \ve + \sqrt{\ve (2- \ve)} \bigg)^2 = \frac{1}{2} \bigg( 1 + 2(1- \ve) \sqrt{\ve (2 - \ve)} \bigg)\,.
 \end{equation}
The right-hand side increases continuously in the range $\ve \leq \ve^{*}$, as we have
\begin{equation}
  \frac{1}{2}  \diff{}{\varepsilon}  \bigg( 1 + 2(1- \ve) \sqrt{\ve (2 - \ve)} \bigg) = \frac{1}{\sqrt{ \varepsilon (2 - \varepsilon ) } } \Big( 2 \ve^2 - 4\ve +1 \Big) \geq 0\,, 
\end{equation}
and we find that for $\ve=\ve^{*}$, Eve can achieve the perfect guessing probability of 1 with this attack, as her decomposition reduces to  
\begin{equation}\label{eqnA: decomp mixed ep}
     \begin{aligned}
      p(i, j) &=           \frac{1}{2} \delta_{i, j}\,,
    \\
    \ket{\varphi_i} &= \frac{1}{\sqrt{2}} \Big( \left( \sqrt{2-\ve^*}- \sqrt{\ve^*}  \right) \ket{\psi} + \sqrt{2 \ve} \ket{i} \Big)\,,
    \\
    N_{x, j} &= 
    \begin{cases}
    \ketbra{\varphi_j}{\varphi_j}\,, & x =j\,,\\
    \id - \ketbra{\varphi_j}{\varphi_j}\,, & x  \neq j\,.\\
    \end{cases}    
     \end{aligned}
\end{equation}
When $\ve \geq \ve^{*}$, Eve can continue to use the decomposition \eqref{eqnA: decomp mixed ep} for $\ve^{*}$, but to comply with the constraints in \eqref{eqnA: pguess mixed}, she should perform a convex combination of the decomposition \eqref{eqnA: decomp mixed ep} and the decomposition \eqref{eqnA: decomp mixed noise} for $\ve= \ve^{*}$, where $\ket{\psi}$ is replaced by the orthogonal state $\ket{\psi^{\perp}}= \frac{1}{\sqrt{2}}\left(1, -1 \right)^{T}$. Concretely, her decomposition is
\begin{equation}\label{eqnA: lambda decomp}
     \begin{aligned}
      p(i, j, \lambda) &=           \frac{1}{2} \delta_{i, j} \, p(\lambda)\,,
    \\
    \ket{\varphi_{i, \lambda}} &= 
    \begin{cases}
        \frac{1}{\sqrt{2}} \Big( \left( \sqrt{2-\ve^*}- \sqrt{\ve^*}  \right) \ket{\psi} + \sqrt{2 \ve} \ket{i} \Big)\,, & \lambda=1\,,
        \\
        \frac{1}{\sqrt{2}} \Big( \left( -1 \right)^{i+1}\left( \sqrt{2-\ve^*}- \sqrt{\ve^*}  \right) \ket{\psi^{\perp}} + \sqrt{2 \ve} \ket{i} \Big)\,, & \lambda=2\\
    \end{cases}
\\
    N_{x, j, \lambda} &= 
    \begin{cases}
    \ketbra{\varphi_{j, \lambda}}{\varphi_{j, \lambda}}\,, & x =j\,,\\
    \id - \ketbra{\varphi_{j, \lambda}}{\varphi_{j, \lambda}}\,, & x  \neq j\,.\\
    \end{cases}    
     \end{aligned}
\end{equation}
This `randomized' decomposition is shown in Figure \ref{fig: noisy_squares}.
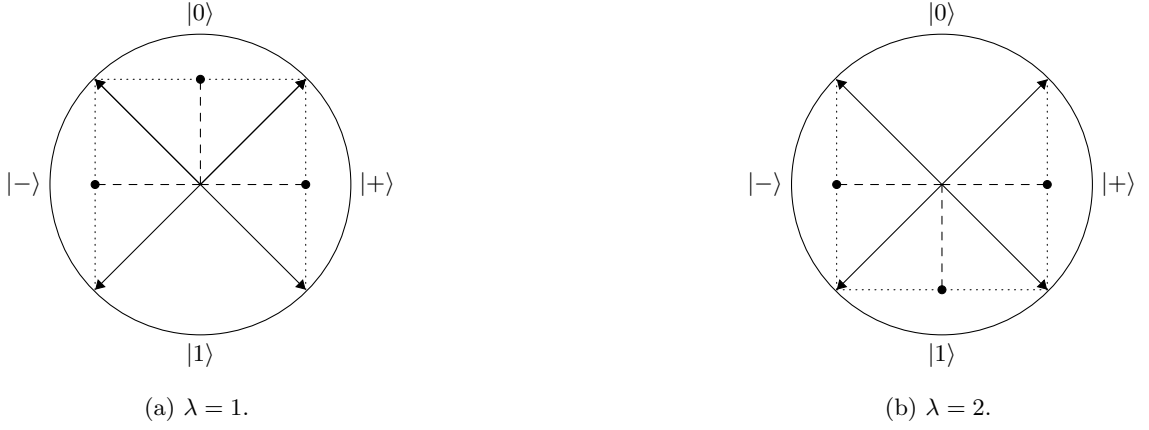
\begin{figure}
     \centering
     \begin{subfigure}[b]{0.45\textwidth}
         \centering
        \begin{tikzpicture}[line cap=round, line join=round, >=Triangle]
     \centering
    \begin{scope}[scale=1]
  \clip(-2.9,-2.6) rectangle (2.8,2.8);
  \draw(0,0) circle (2cm);
  \draw[dashed] (0,0)-- (1.4, 0);
   \draw[dashed] (0,0)-- (-1.4, 0);
\draw [->] (0,0)-- (1.4,1.4);
  \draw [->] (0,0)-- (1.4,-1.4);
  \draw [->] (0,0)-- (-1.4,1.4);
  \draw [->] (0,0)-- (-1.4,-1.4);
  \draw [dotted] (1.4,-1.4)-- (1.4,1.4);
  \draw [dotted] (-1.4,-1.4)-- (-1.4,1.4);
  
  \draw (-2,0) node[anchor=east] {\small{$\ket{{-}}$}};
  \draw (2,0) node[anchor=west] {\small{$\ket{{+}}$}};
  \draw (0,2) node[anchor=south] {\small{$ {|0\rangle}$}};
  \draw (0,-2) node[anchor=north] {\small{$|{1}\rangle$}};
  \draw [fill] (1.4,0) circle (1.5pt);
  \draw [fill] (-1.4,0) circle (1.5pt);

  \draw [fill] (0, 1.4) circle (1.5pt);
  \draw[dashed] (0,0)-- (0, 1.4);
\draw [->] (0,0)-- (1.4, 1.4);
  \draw [->] (0,0)-- (-1.4, 1.4);
  \draw [dotted] (-1.4, 1.4)-- (1.4, 1.4);

  \end{scope}
\end{tikzpicture}
        \caption{$\lambda=1$.}
        \label{fig: noisy_square_up}
        \end{subfigure}
        \hfill
        \begin{subfigure}[b]{0.45\textwidth}
            \centering
            \begin{tikzpicture}[line cap=round, line join=round, >=Triangle]
     \centering
    \begin{scope}[scale=1]
  \clip(-2.9,-2.6) rectangle (2.8,2.8);
  \draw(0,0) circle (2cm);
  \draw[dashed] (0,0)-- (1.4, 0);
   \draw[dashed] (0,0)-- (-1.4, 0);
\draw [->] (0,0)-- (1.4,1.4);
  \draw [->] (0,0)-- (1.4,-1.4);
  \draw [->] (0,0)-- (-1.4,1.4);
  \draw [->] (0,0)-- (-1.4,-1.4);
  \draw [dotted] (1.4,-1.4)-- (1.4,1.4);
  \draw [dotted] (-1.4,-1.4)-- (-1.4,1.4);
  
  \draw (-2,0) node[anchor=east] {\small{$\ket{{-}}$}};
  \draw (2,0) node[anchor=west] {\small{$\ket{{+}}$}};
  \draw (0,2) node[anchor=south] {\small{$ {|0\rangle}$}};
  \draw (0,-2) node[anchor=north] {\small{$|{1}\rangle$}};
  \draw [fill] (1.4,0) circle (1.5pt);
  \draw [fill] (-1.4,0) circle (1.5pt);

  \draw [fill] (0, -1.4) circle (1.5pt);
  \draw[dashed] (0,0)-- (0, -1.4);
  
  \draw [dotted] (1.4, -1.4)-- (-1.4, -1.4);
\end{scope}
\end{tikzpicture}
            \caption{$ \lambda=2$.}
            \label{fig: noisy_square_down}
        \end{subfigure}
    \caption{Eve's decompositions conditioned on the random variable $\lambda$ in \eqref{eqnA: lambda decomp}.}
    \label{fig: noisy_squares}
\end{figure}
If Eve chooses the distribution 
\begin{equation}
    p\left(\lambda_1 \right)= \frac{1}{2} \Big( \sqrt{2} \left(1 - \ve \right) +1 \Big)\,, \qquad p\left(\lambda_2\right)= 1 - p\left(\lambda_1 \right)\,,
\end{equation}
this decomposition complies with all the constraints in \eqref{eqnA: pguess mixed}, so we find 
\begin{equation}
    \pg \big(\rho_{\psi}, \, \m_2 \big) = 1
\end{equation}
for  all $\ve^{*} \leq \ve$. This contrasts with the `single noise' case of a noisy projective measurement with the optimal pure state, or a noisy pure state with the optimal projective measurement,
as in both of these cases, Eve can only achieve perfect guessing probability when $\ve=1$. 

With this classical model for Eve in mind, we consider how we might divide an amount of noise $\delta$ evenly between a state and a measurement device. Using the depolarizing channel for qubits,
\begin{equation}
    \Delta_{\ve} \left( \sigma \right) = \left( 1 - \ve \right) \sigma + \frac{\ve}{2} \id\,,
\end{equation}
for fixed $\delta$, we want to solve for $\ve$ such that the probabilities are preserved, i.e. 
\begin{equation}
 \tr \big( \Delta_{\ve} \left( \rho \right) \, \Delta_{\ve} \left( M_x \right)\big) =   \tr \big( \Delta_{\delta} \left( \rho \right) \, M_x \big) = \tr \big( \rho \,\Delta_{\delta} \left( M_x \right) \big) \,.
\end{equation}
The solution is 
\begin{equation}
    \ve = 1 - \sqrt{1- \delta}\,.
\end{equation}
We then compare the `single noise' case, where $\m_2$, with noise $\delta$, is applied to the perfect state $\ket{\psi}$, to the `joint noise' case, where $\m_2$, with noise $\ve$, is applied to the equally noisy state $\rho_\ve$. The noise parameter $\ve^{*}$ for which Eve reaches perfect guessing probability corresponds to $\delta=1/2$. Considering the joint decomposition described above, we can lower bound Eve's guessing probability as a function of $\delta$ by
\begin{equation}
\pg \big(\rho_\psi, \, \m_2 \big) \geq  \frac{1}{2} \bigg( 1 + 2 \sqrt{\delta \left( 1 - \delta\right)} \bigg)\,    
\end{equation}
for $0 \leq \delta \leq 1/2$. By contrast, the optimal guessing probability in the single noise case is given as a function of $\delta$ by
\begin{equation}
    \pgs \left( \rho_{\psi} \right) = \pgs \left( \m_2 \right) = \frac{1}{2} \bigg( 1 + \sqrt{\delta \left( 2- \delta \right)} \bigg)\,.
\end{equation}
Since 
\begin{equation}
\pg \big(\rho_\psi, \, \m_2 \big) -  \pgs \left( \m_2 \right) \geq \frac{1}{2} \sqrt{\delta} \, \Big( 2 \sqrt{1- \delta} - \sqrt{2 - \delta}  \Big)   \geq 0
\end{equation}
in the range $0 \leq \delta \leq 1/2$, we find in this qubit case study that Eve always achieves more when the noise is shared between the devices.

\end{widetext}

\end{document}